\documentclass[12pt]{article}
\usepackage{amsmath}
\usepackage{amssymb}
\usepackage{amsthm}
\usepackage{graphicx}
\usepackage{subfigure}
\usepackage{cite}
\usepackage{url}
\usepackage[small]{caption}
\usepackage{mathtools,xparse}
\theoremstyle{definition}
\newtheorem{theorem}{Theorem}

\begin{document}
\baselineskip 0.6cm

\def\bra#1{\langle #1 |}
\def\ket#1{| #1 \rangle}
\def\inner#1#2{\langle #1 | #2 \rangle}
\def\app#1#2{%
  \mathrel{%
    \setbox0=\hbox{$#1\sim$}%
    \setbox2=\hbox{%
      \rlap{\hbox{$#1\propto$}}%
      \lower1.1\ht0\box0%
    }%
    \raise0.25\ht2\box2%
  }%
}
\def\approxprop{\mathpalette\app\relax}
\DeclarePairedDelimiter{\norm}{\lVert}{\rVert}

\begin{titlepage}

\begin{flushright}
\end{flushright}

\vskip 1.2cm

\begin{center}
{\Large \bf Toward a Holographic Theory for General Spacetimes}

\vskip 0.7cm

{\large Yasunori Nomura$^{a,b}$,
  Nico Salzetta$^{a,b}$,
  Fabio Sanches$^{a,b}$,
  and Sean J. Weinberg$^c$}

\vskip 0.5cm

$^a$ {\it Berkeley Center for Theoretical Physics, Department of Physics,\\
  University of California, Berkeley, CA 94720, USA}

\vskip 0.2cm

$^b$ {\it Theoretical Physics Group, Lawrence Berkeley National Laboratory, 
 CA 94720, USA}

\vskip 0.2cm

$^c$ {\it Department of Physics, University of California, 
  Santa Barbara, CA 93106, USA}

\vskip 0.8cm

\abstract{We study a holographic theory of general spacetimes that does 
 not rely on the existence of asymptotic regions.  This theory is to be 
 formulated in a holographic space.  When a semiclassical description 
 is applicable, the holographic space is assumed to be a holographic 
 screen:\ a codimension-1 surface that is capable of encoding states 
 of the gravitational spacetime.  Our analysis is guided by conjectured 
 relationships between gravitational spacetime and quantum entanglement 
 in the holographic description.  To understand basic features of this 
 picture, we catalog predictions for the holographic entanglement structure 
 of cosmological spacetimes.  We find that qualitative features of 
 holographic entanglement entropies for such spacetimes differ from 
 those in AdS/CFT but that the former reduce to the latter in the 
 appropriate limit.  The Hilbert space of the theory is analyzed, and 
 two plausible structures are found:\ a direct sum and ``spacetime equals 
 entanglement'' structure.  The former preserves a naive relationship 
 between linear operators and observable quantities, while the latter 
 respects a more direct connection between holographic entanglement 
 and spacetime.  We also discuss the issue of selecting a state in 
 quantum gravity, in particular how the state of the multiverse may 
 be selected in the landscape.}

\end{center}
\end{titlepage}

\section{Introduction}
\label{sec:intro}

As with any other classical object, spacetime is expected to consist 
of a large number of quantum degrees of freedom.  The first explicit 
hint of this came from the discovery that empty spacetime can carry 
entropy~\cite{Bekenstein:1972tm,Bekenstein:1973ur,Bardeen:1973gs,%
Hawking:1974rv,Hawking:1974sw,Gibbons:1977mu}.  What theory describes 
these degrees of freedom as well as the excitations on them, i.e.\ 
matter?

Part of the difficulty in finding such a theory is the large 
redundancies present in the description of gravitational spacetime. 
The holographic principle~\cite{'tHooft:1993gx,Susskind:1994vu,%
Bousso:2002ju} suggests that the natural space in which the microscopic 
degrees of freedom for spacetime (and matter) live is a non-dynamical 
spacetime whose dimension is one less than that in the original 
description (as demonstrated in the special case of the AdS/CFT 
correspondence~\cite{Maldacena:1997re}).  This represents a 
huge redundancy in the original gravitational description beyond 
that associated with general coordinate transformations.  For 
general spacetimes, causality plays a central role in fixing 
this redundancy~\cite{Fischler:1998st,Bousso:1999xy}.  A similar 
idea also plays an important role in addressing problems in the 
semiclassical descriptions of black holes~\cite{Susskind:1993if} 
and cosmology~\cite{Nomura:2011dt,Bousso:2011up}.

In this paper, we explore a holographic theory for general spacetimes. 
We follow a ``bottom-up'' approach given the lack of a useful description 
in known frameworks, such as AdS/CFT and string theory in asymptotically 
Minkowski space.  We assume that our holographic theory is formulated 
on a holographic screen~\cite{Bousso:1999cb}, a codimension-1 surface on 
which the information about the original spacetime can be encoded.  This 
construction can be extended beyond the semiclassical regime by considering 
all possible states on all possible slices---called leaves---of holographic 
screens~\cite{Nomura:2011dt,Nomura:2011rb}, where the nonuniqueness 
of erecting a holographic screen is interpreted as the freedom in fixing 
the redundancy associated with holography.  The resulting picture is 
consistent with the recently discovered area theorem applicable to the 
holographic screens~\cite{Bousso:2015mqa,Bousso:2015qqa,Sanches:2016pga}.

To study the structure of the theory, we use conjectured relationships 
between spacetime in the gravitational description and quantum 
entanglement in the holographic theory.  Recently, it has become 
increasingly clear that quantum entanglement among holographic 
degrees of freedom plays an important role in the emergence of 
classical spacetime~\cite{Ryu:2006bv,Ryu:2006ef,Hubeny:2007xt,%
VanRaamsdonk:2009ar,Swingle:2012wq,Lewkowycz:2013nqa,Maldacena:2013xja,%
Sanches:2016sxy,Freedman:2016zud,Almheiri:2016blp,Nomura:2016aww,%
Harlow:2016vwg}.  In particular, Ref.~\cite{Sanches:2016sxy} showed 
that the areas of the extremal surfaces anchored to the boundaries 
of regions on a leaf of a holographic screen satisfy relations obeyed 
by entanglement entropies, so that they can indeed be identified as 
the entanglement entropies associated with the corresponding regions 
in the holographic space.  We analyze properties of these surfaces 
and discuss their implications for a holographic theory of general 
spacetimes.

We lay down our general framework in Section~\ref{sec:framework}. 
We then study the behavior of extremal surfaces in cosmological 
Friedmann-Robertson-Walker (FRW) spacetimes in Section~\ref{sec:FRW}. 
Here we focus on initially expanding flat and open universes, in which 
the area of the leaves monotonically increases.  We first consider 
universes dominated by a single component in the Friedmann equation, 
and we identify how screen entanglement entropies---the entanglement 
entropies among the degrees of freedom in the holographic space---encode 
information about the spacetimes.  We discuss next how the screen 
entanglement entropies behave in a transition period in which the 
dominant component of the universe changes.  We find an interesting 
theorem when the holographic screen is spacelike:\ the change of a 
screen entanglement entropy is always monotonic.  The proof of this 
theorem is given in Appendix~\ref{app:spacelike}.  If the holographic 
screen is timelike, no such theorem holds.

In Section~\ref{sec:beyond}, we study the structure of the holographic 
theory for general spacetimes, building on the results obtained earlier. 
In particular, we discuss how the holographic entanglement entropies for 
general spacetimes differ from those in AdS/CFT and how, nevertheless, 
the former reduce to the latter in an appropriate limit.  We emphasize 
that the holographic entanglement entropies for cosmological spacetimes 
obey a volume law, rather than an area law, implying that the relevant 
holographic states are not ground states of local field theories. 
This is the case despite the fact that the dynamics of the holographic 
theory respects some sense of locality, indicated by the fact that 
the area of a leaf increases in a local manner on a holographic screen.

The Hilbert space of the theory is analyzed in 
Section~\ref{subsec:structure} under two assumptions:
\begin{itemize}
\item[(i)]
The holographic theory has (effectively) a qubit degree of freedom 
per each volume of $4 \ln 2$ in Planck units.  These degrees of 
freedom appear local at lengthscales larger than a microscopic 
cutoff $l_{\rm c}$.
\item[(ii)]
If a holographic state represents a semiclassical spacetime, the 
area of an extremal surface anchored to the boundary of a region 
$\Gamma$ on a leaf $\sigma$ and contained in the causal region 
associated with $\sigma$ represents the entanglement entropy of 
$\Gamma$ in the holographic theory.
\end{itemize}
We find that these two assumptions strongly constrain the structure 
of the Hilbert space, although they do not determine it uniquely. 
There are essentially two possibilities:
\begin{itemize}
\item[] {\bf Direct sum structure} 
--- Holographic states representing different semiclassical spacetimes 
${\cal M}$ live in different Hilbert spaces ${\cal H}_{\cal M}$ even 
if these spacetimes have the same boundary space (or leaf) $B$
\begin{equation}
  {\cal H}_B = \bigoplus_{\cal M} {\cal H}_{\cal M}.
\label{eq:DS}
\end{equation}
In each Hilbert space ${\cal H}_{\cal M}$, the states representing 
the semiclassical spacetime comprise only a tiny subset of all the 
states---the vast majority of the states in ${\cal H}_{\cal M}$ do 
not allow for a semiclassical interpretation, which we call ``firewall'' 
states borrowing the terminology in Refs.~\cite{Almheiri:2012rt,%
Almheiri:2013hfa,Marolf:2013dba}.  In fact, the states allowing 
for a semiclassical spacetime interpretation do not even form a 
vector space---their superposition may lead to a firewall state if 
it involves a large number of terms, of order a positive power of 
${\rm dim}\,{\cal H}_{\cal M}$.  This is because a superposition 
involving such a large number of terms significantly alters the 
entanglement entropy structure, so under assumption~(ii) above 
we cannot interpret the resulting state as a semiclassical state 
representing ${\cal M}$.  In this picture, small excitations 
over spacetime ${\cal M}$ can be represented by standard linear 
operators acting on the (suitably extended) Hilbert space 
${\cal H}_{\cal M}$, which can be trivially promoted to linear 
operators in ${\cal H}_B$.
\item[] {\bf Spacetime equals entanglement} 
--- Holographic states that represent different semiclassical 
spacetimes but have same boundary space $B$ are all elements 
of a single Hilbert space ${\cal H}_B$.  And yet, the number of 
independent microstates representing {\it each} of these spacetimes, 
${\cal M}, {\cal M}', {\cal M}'', \cdots$, is the dimension of 
${\cal H}_B$:
\begin{equation}
  \ket{\Psi^{\cal M}_i},\,\, \ket{\Psi^{{\cal M}'}_{i'}},\,\, 
  \ket{\Psi^{{\cal M}''}_{i''}},\,\, \cdots \,\in\, {\cal H}_B;
\qquad
  i, i', i'',\cdots = 1, \cdots, {\rm dim}\,{\cal H}_B,
\label{eq:RD}
\end{equation}
which implies that the microstates representing different spacetimes 
are not independent.  This picture arises if we require the 
converse of assumption~(ii) and is called ``spacetime equals 
entanglement''~\cite{Nomura:2016aww}:\ if a holographic state 
has the form of entanglement entropies corresponding to a certain 
spacetime, then the state indeed represents that spacetime.  The 
structure of Eq.~(\ref{eq:RD}) is then obtained because arbitrary 
unitary transformations acting in each cutoff size cell in $B$ 
do not change the entanglement entropies, implying that the 
number of microstates for any geometry is ${\rm dim}\,{\cal H}_B$ 
(so they span a basis of ${\cal H}_B$).  Despite the intricate 
structure of the states, this picture admits the standard many 
worlds interpretation for classical spacetimes, as shown in 
Ref.~\cite{Nomura:2016aww}.  Small excitations over spacetime 
are represented by non-linear/state-dependent operators, 
along the lines of Ref.~\cite{Papadodimas:2015jra} (see 
also~\cite{Papadodimas:2012aq,Verlinde:2012cy,Nomura:2012ex}), 
since a superposition of background spacetimes may lead to 
another spacetime, so that operators representing excitations 
must know the entire quantum state they act on.
\end{itemize}

We note that a dichotomy similar to the one described above was 
discussed earlier in Ref.~\cite{Papadodimas:2015jra}, but the 
interpretation and the context in which it appears here are 
distinct.  First, the state-dependence of the operators representing 
excitations in the second scenario (as well as that of the 
time evolution operator) becomes relevant when the boundary 
space is involved in the dynamics as in the case of cosmological 
spacetimes.  Hence, this particular state-dependence need not 
persist in the AdS/CFT limit.  This does not imply anything about 
the description of the interior a black hole in the CFT.  It is 
possible that the CFT does not provide a semiclassical description 
of the black hole interior, i.e.\ it gives only a distant description. 
Alternatively, there may be a way of obtaining a state-dependent 
semiclassical description of the black hole interior within a CFT, 
as envisioned in Ref.~\cite{Papadodimas:2015jra}.  We are agnostic 
about this issue.

Second, Ref.~\cite{Papadodimas:2015jra} describes the dichotomy 
as state-dependence vs.\ firewalls.  Our picture, on the other hand, 
does not have a relation with firewalls because the following two 
statements apply to {\it both} the direct sum and spacetime equals 
entanglement pictures:
\begin{itemize}
\item
Most of the states in the Hilbert space, e.g.\ in the Haar measure, 
are firewalls in the sense that they do not represent smooth 
semiclassical spacetimes, which require special entanglement 
structures among the holographic degrees of freedom.
\item
The fact that most of the states are firewalls does not mean that 
these states are realized as a result of standard time evolution, 
in which the volume of the boundary space increases in time. 
In fact, the direct sum picture even has a built-in mechanism 
of eliminating firewalls through time evolution, as we will see 
in Section~\ref{subsec:time-evo}.%
\footnote{This is natural because any dynamics leading to 
 classicalization selects only a very special set of states as 
 the result of time evolution:\ states interpreted as a superposition 
 of a small number of classical worlds, where small means a number 
 (exponentially) smaller than the dimension of the full microscopic 
 Hilbert space.}
\end{itemize}
Rather, the real tension is between the linearity/state-independence 
of operators representing observables (including the time evolution 
operator) and the spacetime equals entanglement hypothesis, i.e.\ 
the hypothesis that if a holographic state has entanglement entropies 
corresponding to a semiclassical spacetime, then the state indeed 
represents that spacetime.  If we insist on the linearity of observables, 
we are forced to take the direct sum picture; if we adopt the spacetime 
equals entanglement hypothesis, then we must give up linearity.

Our analysis in Section~\ref{sec:beyond} also includes the following. 
In Section~\ref{subsec:reconst}, we discuss bulk reconstruction 
from a holographic state, which suggests that the framework 
provides a distant description for a dynamical black hole.  In 
Section~\ref{subsec:exterior}, we consider how the theory encodes 
information about spacetime outside the causal region of a leaf, 
which is needed for autonomous time evolution.  Our analysis 
suggests a strengthened covariant entropy bound:\ the entropy 
on the {\it union} of two light sheets (future-directed ingoing 
and past-directed outgoing) of a leaf is bounded by the area of 
the leaf divided by $4$.  This bound is stronger than the original 
bound in Ref.~\cite{Bousso:1999xy}, which says that the entropy on 
{\it each} of the two light sheets is bounded by the area divided 
by $4$.  In Section~\ref{subsec:time-evo}, we analyze properties 
of time evolution, in particular a built-in mechanics of eliminating 
firewalls in the direct sum picture and the required non-linearity 
of the time evolution operator in the spacetime equals entanglement 
picture.  In Sections~\ref{subsec:exterior} and \ref{subsec:time-evo}, 
we discuss how our framework may reduce to AdS/CFT and string theory 
in an asymptotically Minkowski background in the appropriate limits. 
We argue that the dynamics of these theories (in which the boundaries 
are sent to infinity) describe that of the general holographic 
theory modded out by ``vacuum degeneracies'' relevant for the 
dynamics of the boundaries and the exteriors.

In Section~\ref{sec:discuss}, we devote our final discussion to the 
issue of selecting a state.  In general, specifying a system requires 
selection conditions on a state in addition to determining the theory. 
To address this issue in quantum gravity, we need to study the problem 
of time~\cite{DeWitt:1967yk,Wheeler:1967}.  We discuss possible 
signals from a past singularity or past null infinity, closed 
universes and ``fine-tuning'' of states, and selection conditions 
for the string theory landscape~\cite{Bousso:2000xa,Kachru:2003aw,%
Susskind:2003kw,Douglas:2003um}, especially the scenario called the 
``static quantum multiverse''~\cite{Nomura:2012zb}.  While our discussion 
in this section is schematic, it allows us to develop intuition about 
how quantum gravity might work at the fundamental level when applied 
to the real world.

Throughout the paper, we adopt the Schr\"{o}dinger picture of quantum 
mechanics and take the Planck length to be unity, $l_{\rm P} = 1$. 
When the semiclassical picture is applicable, we assume the null and 
causal energy conditions to be satisfied.  These impose the conditions 
$\rho \geq -p$ and $|\rho| \geq |p|$, respectively, on the energy 
density $\rho$ and pressure $p$ of an ideal fluid component.  The 
equation of state parameter $w = p/\rho$, therefore, takes a value 
in the range $|w| \leq 1$.

\section{Holography and Quantum Gravity}
\label{sec:framework}

The holographic principle states that quantum mechanics of a system 
with gravity can be formulated as a non-gravitational theory in spacetime 
with dimension one less than that in the gravitational description. 
The covariant entropy bound, or Bousso bound,~\cite{Bousso:1999xy} 
suggests that this holographically reduced---or ``boundary''---spacetime 
may be identified as a hypersurface in the original gravitational 
spacetime determined by a collection of light rays.  Specifically, 
it implies that the entropy on a null hypersurface generated by a 
congruence of light rays terminating at a caustic or singularity is 
bounded by its largest cross sectional area ${\cal A}$; in particular, 
the entropy on each side of the largest cross sectional surface is 
bounded by ${\cal A}/4$ in Planck units.%
\footnote{We will conjecture a stronger bound in 
 Section~\ref{subsec:exterior}.}
It is therefore natural to consider that, for a fixed gravitational 
spacetime, the holographic theory lives on a hypersurface---called 
the holographic screen---on which null hypersurfaces foliating the 
spacetime have the largest cross sectional areas~\cite{Bousso:1999cb}.

This procedure of erecting a holographic screen has a large 
ambiguity, presumably reflecting a large freedom in fixing the 
redundancy of the gravitational description associated with the 
holographic principle.  A particularly useful choice advocated in 
Refs.~\cite{Nomura:2011dt,Nomura:2011rb,Nomura:2013nya} is to adopt 
an ``observer centric reference frame.''  Let the origin of the 
reference frame follow a timelike curve $p(\tau)$ which passes through 
a fixed spacetime point $p_0$ at $\tau = 0$, and consider the congruence 
of past-directed light rays emanating from $p_0$.%
\footnote{In Refs.~\cite{Nomura:2011dt,Nomura:2011rb,Nomura:2013nya}, 
 $p(\tau)$ was chosen to be a timelike geodesic with $\tau$ being the 
 proper time measured at $p(\tau)$.  We suspect that this simplifies 
 the time evolution operator in the holographic theory.}
The expansion of the light rays $\theta$ satisfies
\begin{equation}
  \frac{\partial \theta}{\partial \lambda} + \frac{1}{2} \theta^2 
  \leq 0,
\label{eq:convergence}
\end{equation}
where $\lambda$ is the affine parameter associated with the light rays. 
This implies that the light rays emitted from $p_0$ focus toward the 
past (starting from $\theta = +\infty$ at $\lambda = 0_+$), and we may 
identify the apparent horizon, i.e.\ the codimension-2 surface with
\begin{equation}
  \theta = 0,
\label{eq:app-hor}
\end{equation}
to be an equal-time hypersurface---called a leaf---of a holographic 
screen.  Repeating the procedure for all $\tau$, we obtain a specific 
holographic screen, with the leaves parameterized by $\tau$, corresponding 
to foliating the spacetime region accessible to the observer at 
$p(\tau)$; see Fig.~\ref{fig:p-tau}.  Such a foliation is consonant 
with the complementarity hypothesis~\cite{Susskind:1993if}, which 
asserts that a complete description of a system is obtained by 
referring only to the spacetime region that can be accessed by 
a single observer.
\begin{figure}[t]
\begin{center}
  \includegraphics[height=6.5cm]{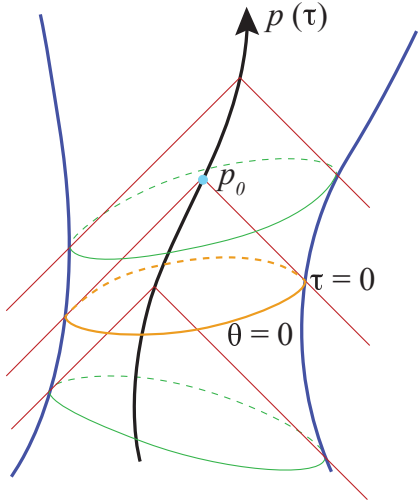}
\end{center}
\caption{For a fixed semiclassical spacetime, the holographic screen 
 is a hypersurface obtained as the collection of codimension-2 
 surfaces (labeled by $\tau$) on which the expansion of the 
 light rays emanating from a timelike curve $p(\tau)$ vanishes, 
 $\theta = 0$.  This way of erecting the holographic screen 
 automatically deals with the redundancy associated with 
 complementarity.  The ambiguity of choosing $p(\tau)$ reflects 
 a large freedom in fixing the redundancy associated with 
 holography.}
\label{fig:p-tau}
\end{figure}

With this construction, we can view a quantum state of the holographic 
theory as living on a leaf of the holographic screen obtained 
in the above observer centric manner.  We can then consider the 
collection of all possible quantum states on all possible leaves, 
obtained by considering all timelike curves in all spacetimes.  We 
take the view that a state of quantum gravity lives in the Hilbert 
space spanned by all of these states (together with other states that 
do not admit a full spacetime interpretation)~\cite{Nomura:2011dt,%
Nomura:2011rb}.  It is often convenient to consider a Hilbert space 
${\cal H}_B$ spanned by the holographic states that live on the 
``same'' boundary space $B$.%
\footnote{The exact way in which the boundary spaces are grouped 
 into different $B$'s is unimportant.  For example, one can regard 
 the boundary spaces having the same area ${\cal A}$ within some 
 precision $\delta {\cal A}$ to be in the same $B$, or one can 
 discriminate them further by their induced metrics.  This ambiguity 
 does not affect any of the results, unless one takes $\delta {\cal A}$ 
 to be exponentially small in ${\cal A}$ or discriminates induced 
 metrics with the accuracy of order the Planck length (which 
 corresponds to resolving microstates of the spacetime). 
 \label{ft:def-B}}
The relevant Hilbert space can then be written as
\begin{equation}
  {\cal H} = \sum_B {\cal H}_B,
\label{eq:H}
\end{equation}
where the sum of Hilbert spaces is defined by%
\footnote{Unlike Ref.~\cite{Nomura:2011rb}, here we do not 
 assume specific relations between ${\cal H}_B$'s; for example, 
 ${\cal H}_{B_1}$ and ${\cal H}_{B_2}$ for different boundary spaces 
 $B_1$ and $B_2$ may not be orthogonal.  Also, we have included in 
 the sum over $B$ the cases in which $B$ is outside the semiclassical 
 regime, i.e.\ the cases in which the holographic space does not 
 correspond to a leaf of a holographic screen in a semiclassical 
 regime.  These issues will be discussed in Section~\ref{sec:beyond}.}
\begin{equation}
  {\cal H}_1 + {\cal H}_2 
  = \{ v_1 + v_2\, |\, v_1 \in {\cal H}_1, v_2 \in {\cal H}_2 \}.
\label{eq:H_sum}
\end{equation}
This formulation is not restricted to descriptions based on fixed 
semiclassical spacetime backgrounds.  For example, we may consider 
a state in which macroscopically different spacetimes are superposed; 
in particular, this picture describes the eternally inflating 
multiverse as a state in which macroscopically different universes 
are superposed~\cite{Nomura:2011dt,Nomura:2012zb}.  The space in 
Eq.~(\ref{eq:H}) is called the covariant Hilbert space with observer 
centric gauge fixing.

Recently, Bousso and Engelhardt have identified two special classes 
of holographic screens~\cite{Bousso:2015mqa,Bousso:2015qqa}:\ if a 
portion of a holographic screen is foliated by marginally anti-trapped 
(trapped) surfaces, then that portion is called a past (future) 
holographic screen.  Specifically, denoting the two future-directed 
null vector fields orthogonal to a portion of a leaf by $k^a$ 
and $l^a$, with $k^a$ being tangent to light rays emanating from 
$p(\tau)$, the expansion of the null geodesic congruence generated 
by $l^a$ satisfies $\theta_l > 0$ and $< 0$ for past and future 
holographic screens, respectively.  They proved, building on earlier 
works~\cite{Hayward:1993wb,Hayward:1997jp,Ashtekar:2002ag,Ashtekar:2003hk}, 
that the area of leaves ${\cal A}(\tau)$ monotonically increases 
(decreases) for a past (future) holographic screen:
\begin{equation}
  \left\{ \begin{array}{l}
    \theta_k = 0 \\ \theta_l \gtrless 0
  \end{array} \right.
\quad\Leftrightarrow\qquad
  \frac{d}{d\tau} {\cal A}(\tau) \gtrless 0;
\label{eq:area-law}
\end{equation}
see Fig.~\ref{fig:def}.
\begin{figure}[t]
\begin{center}
  \includegraphics[height=6.5cm]{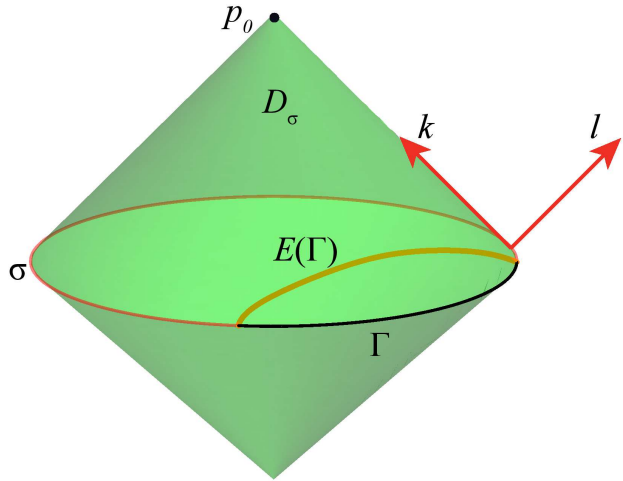}
\end{center}
\caption{The congruence of past-directed light rays emanating from $p_0$ 
 (the origin of the reference frame) has the largest cross sectional 
 area on a leaf $\sigma$, where the holographic theory lives.  At any 
 point on $\sigma$, there are two future-directed null vectors orthogonal 
 to the leaf:\ $k^a$ and $l^a$.  For a given region $\Gamma$ of the leaf, 
 we can find a codimension-2 extremal surface $E(\Gamma)$ anchored to 
 the boundary $\partial \Gamma$ of $\Gamma$, which is fully contained 
 in the causal region $D_\sigma$ associated with $\sigma$.}
\label{fig:def}
\end{figure}
In many regular circumstances, including expanding FRW universes, the 
holographic screen is a past holographic screen, so that the area of 
the leaves monotonically increases, $d{\cal A}(\tau)/d\tau > 0$.  In 
this paper we mostly focus on this case, and we interpret the area theorem 
in terms of the second law of thermodynamics applied to the Hilbert space 
of Eq.~(\ref{eq:H}).  Moreover, in Ref.~\cite{Sanches:2016pga} it was 
proved that this area theorem holds locally on the holographic screen:\ 
the area of any fixed spatial portion of the holographic screen, 
determined by a vector field tangent to the holographic screen and 
normal to its leaves, increases monotonically in time.  This implies 
that the dynamics of the holographic theory respects some notion 
of locality.

What is the structure of the holographic theory and how can we explore it? 
Recently, a conjecture has been made in Ref.~\cite{Sanches:2016sxy} which 
relates geometries of general spacetimes in the gravitational description 
to the entanglement entropies of states in the holographic theory. 
This extends the analogous theorem/conjecture in the AdS/CFT 
context~\cite{Ryu:2006bv,Ryu:2006ef,Hubeny:2007xt} to more general 
cases, allowing us to probe the holographic description of general 
spacetimes, including those that do not have an obvious spacetime 
boundary on which the holographic theory can live.  In particular, 
Ref.~\cite{Sanches:2016sxy} proved that for a given region $\Gamma$ 
of a leaf $\sigma$, a codimension-2 extremal surface $E(\Gamma)$ 
anchored to the boundary $\partial \Gamma$ of $\Gamma$ is fully 
contained in the causal region $D_\sigma$ of $\sigma$:
\begin{equation}
  D_\sigma:\,\, \mbox{the domain of dependence of an interior achronal 
    hypersurface whose only boundary is $\sigma$},
\label{eq:D_sigma}
\end{equation}
where the concept of the interior is defined so that a vector on $\sigma$ 
pointing toward the interior takes the form $c_1 k^a - c_2 l^a$ with 
$c_1, c_2 > 0$ (see Fig.~\ref{fig:def}).  This implies that the normalized 
area of the extremal surface $E(\Gamma)$
\begin{equation}
  S(\Gamma) = \frac{1}{4} \norm{E(\Gamma)},
\label{eq:S_Sigma}
\end{equation}
satisfies expected properties of entanglement entropy, such as strong 
subadditivity, so that it can be identified with the entanglement entropy 
of the region $\Gamma$ in the holographic theory.  Here, $\norm{x}$ 
represents the area of $x$.  If there are multiple extremal surfaces 
in $D_\sigma$ for a given $\Gamma$, then we must take the one with the 
minimal area.

In the rest of the paper, we study the holographic theory of quantum 
gravity for general spacetimes, adopting the framework described in 
this section.  We first analyze FRW spacetimes and then discuss lessons 
learned from that analysis later.

\section{Holographic Description of FRW Universes}
\label{sec:FRW}

In this section, we study the putative holographic description of 
$(3+1)$-dimensional FRW cosmological spacetimes:
\begin{equation}
  ds^2 = -dt^2 + a^2(t) \left[ \frac{dr^2}{1-\kappa r^2} 
    + r^2 (d\psi^2 + \sin^2\!\psi\, d\phi^2) \right],
\label{eq:FRW-metric}
\end{equation}
where $a(t)$ is the scale factor, and $\kappa < 0$, $= 0$ and $> 0$ for 
open, flat and closed universes, respectively.  The Friedmann equation 
is given by
\begin{equation}
  \left(\frac{\dot{a}}{a}\right)^2 + \frac{\kappa}{a^2} 
  = \frac{8\pi}{3}\rho,
\label{eq:Friedmann-eq}
\end{equation}
where the dot represents $t$ derivative.  Here, we include the energy 
density from the cosmological constant as a component in $\rho$ having 
the equation of state parameter $w = -1$.

As discussed in the previous section, we describe the system as viewed 
from a reference frame whose origin follows a timelike curve $p(\tau)$, 
which we choose to be at $r = 0$.  The holographic theory then lives 
on the holographic screen, an equal-time slice of which is an apparent 
horizon:\ a codimension-2 surface on which the expansion of the 
light rays emanating from $p(\tau)$ for a fixed $\tau$ vanishes. 
Under generic conditions, this horizon is always at a finite distance
\begin{equation}
  r = \frac{1}{\sqrt{\dot{a}^2(t_*) + \kappa}} 
    \equiv r_{\rm AH}(t_*) < \infty,
\label{eq:r_AH}
\end{equation}
where $t_*$ is the FRW time on the horizon.  Note that the symmetry of 
the setup makes the FRW time the same everywhere on the apparent horizon, 
and for an open universe, $\dot{a}(t_*) > \sqrt{-\kappa}$ is satisfied 
for values of $\tau$ before $p(\tau)$ hits the big crunch.  For flat 
and open universes, we find that this surface is always marginally 
anti-trapped, i.e.\ a leaf of a past holographic screen, as long as 
the universe is initially expanding.  On the other hand, for a closed 
universe the surface can change from marginally anti-trapped to marginally 
trapped as $\tau$ increases, implying that the holographic screen may be 
a past holographic screen only until certain time $\tau$.  In this section, 
we focus our attention on initially expanding flat and open universes. 
Closed universes will be discussed in Section~\ref{sec:discuss}.

Below, we study entanglement entropies for subregions in the holographic 
theory---screen entanglement entropies---adopting the conjecture of 
Ref.~\cite{Sanches:2016sxy}.  Here we focus on studying the properties 
of these entropies, leaving their detailed interpretation for later. 
We first discuss ``stationary'' aspects of screen entanglement entropies, 
concentrating on states representing spacetime in which the expansion 
of the universe is dominated by a single component in the Friedmann 
equation.  We study how screen entanglement entropies encode the 
information about the spacetime the state represents.  We then analyze 
dynamics of screen entanglement entropies during a transition period 
in which the dominant component changes.  Implications of these results 
in the broader context of the holographic description of quantum gravity 
will be discussed in the next section.

\subsection{Holographic dictionary for FRW universes}
\label{subsec:single}

Consider a Hilbert space ${\cal H}_B$ spanned by a set of quantum states 
living in the same codimension-2 boundary surface $B$.  As mentioned 
in footnote~\ref{ft:def-B}, the definition of the boundary surface 
being the same has an ambiguity.  For our analysis of states representing 
FRW spacetimes, we take the boundary $B$ to be specified by its area 
${\cal A}_B$ (within some precision $\delta {\cal A}_B$ that is not 
exponentially small in ${\cal A}_B$).  In this subsection, we focus 
on a single Hilbert space ${\cal H}_* \in \{ {\cal H}_B \}$ specified 
by a fixed (though arbitrary) boundary area ${\cal A}_*$.

Consider FRW universes with $\kappa \leq 0$ having vacuum energy 
$\rho_\Lambda$ and filled with varying ideal fluid components.%
\footnote{The $\rho_\Lambda$ here represents the energy density of 
 a (local) minimum of the potential near which fields in the FRW 
 universe in question take values.  In fact, string theory suggests 
 that there is no absolutely stable de~Sitter vacuum in full quantum 
 gravity; it must decay, at least, before the Poincar\'{e} recurrence 
 time~\cite{Kachru:2003aw}.}
For every universe with
\begin{equation}
  \rho_\Lambda < \frac{3}{2 {\cal A}_*},
\label{eq:A*-cond}
\end{equation}
there is an FRW time $t_*$ at which the area of the leaf of the past 
holographic screen is ${\cal A}_*$; see Fig.~\ref{fig:A-t}. 
\begin{figure}[t]
\begin{center}
  \includegraphics[height=6.5cm]{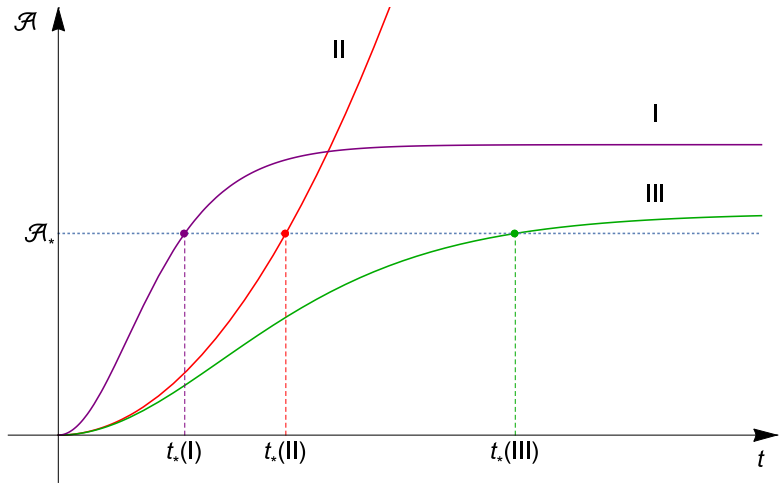}
\end{center}
\caption{Various FRW universes, ${\rm I}, {\rm II}, {\rm III}, \cdots$, 
 have the same boundary area ${\cal A}_*$ at different times, $t_*({\rm I}), 
 t_*({\rm II}), t_*({\rm III}), \cdots$.  Quantum states representing 
 universes at these moments belong to Hilbert space ${\cal H}_*$ specified 
 by the value of the boundary area.}
\label{fig:A-t}
\end{figure}
This is because the area of the leaf of the past holographic screen 
is monotonically increasing~\cite{Bousso:2015mqa}, and the final 
(asymptotic) value of the area is given by
\begin{equation}
  {\cal A}_\infty = \left\{ 
  \begin{array}{ll}
    \frac{3}{2 \rho_\Lambda}, & \mbox{ for } \rho_\Lambda > 0,\\
    +\infty, & \mbox{ for } \rho_\Lambda \leq 0.
  \end{array} \right.
\label{eq:A_infty}
\end{equation}
Any quantum state representing the system at any such moment is an element 
of ${\cal H}_*$.  A question is what features of the holographic state 
encode information about the universe it represents.

To study this problem, we perform the following analysis.  First, given 
an FRW universe specified by the history of the energy density of the 
universe, $\rho(t)$, we determine the FRW time $t_*$ at which the apparent 
horizon $\sigma_*$, identified as a leaf of the past holographic screen, 
has the area ${\cal A}_*$:
\begin{equation}
  \left\{ \begin{array}{l}
    \rho(t) \\ {\cal A}_*
  \end{array} \right.
\rightarrow\,\,
  t_*,
\label{eq:inp-outp}
\end{equation}
where we assume Eq.~(\ref{eq:A*-cond}).  We then consider a spherical 
cap region of the leaf $\sigma_*$ specified by an angle $\gamma$ 
($0 \leq \gamma \leq \pi$):
\begin{equation}
  L(\gamma):\ \,\, 
  t = t_*, \quad r = r_{\rm AH}(t_*), \quad 0 \leq \psi \leq \gamma,
\label{eq:L_gamma}
\end{equation}
where $r_{\rm AH}(t_*)$ is given by Eq.~(\ref{eq:r_AH}) (see 
Fig.~\ref{fig:cap}), and determine the extremal surface $E(\gamma)$ 
which is codimension-2 in spacetime, anchored on the boundary 
of $L(\gamma)$, and fully contained inside the causal region 
$D_{\sigma_*}$ associated with $\sigma_*$.
\begin{figure}[t]
\begin{center}
  \includegraphics[height=6cm]{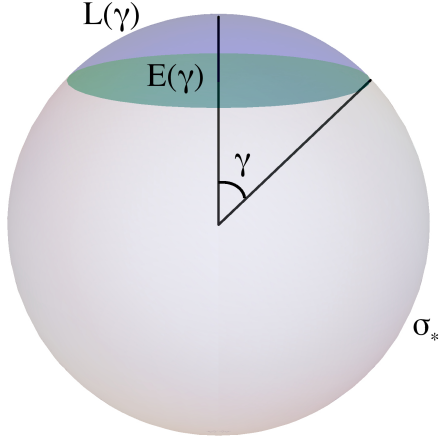}
\end{center}
\caption{A region $L(\gamma)$ of the leaf $\sigma_*$ is parameterized 
 by an angle $\gamma: [0, \pi]$.  The extremal surface $E(\gamma)$ 
 anchored to its boundary, $\partial L(\gamma)$, is also depicted 
 schematically.  (In fact, $E(\gamma)$ bulges into the time direction.)}
\label{fig:cap}
\end{figure}
According to Ref.~\cite{Sanches:2016sxy}, we interpret the quantity
\begin{equation}
  S(\gamma) = \frac{1}{4} \norm{E(\gamma)},
\label{eq:S_gamma}
\end{equation}
to represent von~Neumann entropy of the holographic state representing 
the region $L(\gamma)$, obtained after tracing out the complementary 
region on $\sigma_*$.

To determine the extremal surface $E(\gamma)$, it is useful to introduce 
cylindrical coordinates
\begin{equation}
  \xi = r \sin\psi,
\qquad
  z = r \cos\psi.
\label{eq:cylind}
\end{equation}
We find that the isometry of the FRW metric, Eq.~(\ref{eq:FRW-metric}), 
allows us to move the boundary on which the extremal surface is anchored, 
$\partial L(\gamma)$, on the $z=0$ plane:
\begin{equation}
  \partial L(\gamma):\ \,\, 
  t = t_*, \quad 
  \xi = r_{\rm AH}(t_*) \sin\gamma \equiv \xi_{\rm AH}, 
  \quad z = 0.
\label{eq:partial-L}
\end{equation}
The surface to be extremized is then parameterized by functions $t(\xi)$ 
and $z(\xi)$ with the boundary conditions
\begin{equation}
  t(\xi_{\rm AH}) = t_*,
\qquad
  z(\xi_{\rm AH}) = 0,
\label{eq:bc}
\end{equation}
and the area functional to be extremized is given by
\begin{equation}
  2\pi \int_0^{\xi_{\rm AH}} a(t)\, \xi\, 
    \sqrt{-\biggl( \frac{dt}{d\xi} \biggr)^2
    + \frac{a^2(t)}{1-\kappa(\xi^2 + z^2)} 
    \biggl\{ (1-\kappa z^2)
      + (1-\kappa \xi^2) \biggl( \frac{dz}{d\xi} \biggr)^2
      + 2\kappa \xi z \frac{dz}{d\xi} \biggr\}}\, d\xi.
\label{eq:area-gen}
\end{equation}
In all the examples we study (in this and next subsections), we find 
that the extremal surface does not bulge into the $z$ direction. 
In this case, we can set $z = 0$ in Eq.~(\ref{eq:area-gen}) and find
\begin{equation}
  \norm{E(\gamma)} = \underset{t(\xi)}{\rm ext} \left[ 
    2\pi \int_0^{r_{\rm AH}(t_*) \sin\gamma}\!\! a(t)\, \xi\, 
    \sqrt{-\biggl( \frac{dt}{d\xi} \biggr)^2
      + \frac{a^2(t)}{1-\kappa \xi^2}}\, d\xi \right].
\label{eq:norm-E}
\end{equation}

The analysis described above is greatly simplified if the expansion 
of the universe is determined by a single component in the Friedmann 
equation, i.e.\ a single fluid component with the equation of state 
parameter $w$ or negative spacetime curvature.  We thus focus on 
the case in which the expansion is dominated by a single component 
in (most of) the region probed by the extremal surfaces.  In realistic 
FRW universes this holds for almost all $t$, except for a few Hubble 
times around when the dominant component changes from one to another. 
Discussion about a transition period in which the dominant component 
changes will be given in the next subsection.

\subsubsection*{A flat FRW universe filled with a single fluid component}

Suppose the expansion of the universe is determined dominantly by 
a single ideal fluid component with $w$.  The scale factor is then 
given by
\begin{equation}
  a(t) = c\, t^{\frac{2}{3(1+w)}},
\label{eq:at-single}
\end{equation}
where $c$ is a constant, and the metric in the region $r \leq r_{\rm AH}$ 
takes the form
\begin{equation}
  ds^2 = -dt^2 + c^2\, t^{\frac{4}{3(1+w)}} \left[ dr^2 
    + r^2 (d\psi^2 + \sin^2\!\psi\, d\phi^2) \right],
\label{eq:FRW-single}
\end{equation}
where we have used the fact that $|\kappa\, r_{\rm AH}^2| \ll 1$. 
In this case, we find that the ${\cal A}_*$ dependence of screen 
entanglement entropy $S_\Gamma$ for an arbitrarily shaped region 
$\Gamma$ on $\sigma_*$---specified as a region on the $\psi$-$\phi$ 
plane---is given by
\begin{equation}
  S_\Gamma = \tilde{S}_\Gamma {\cal A}_*,
\label{eq:S-gamma}
\end{equation}
where $\tilde{S}_\Gamma$ does not depend on ${\cal A}_*$.  This can be 
seen in the following way.

Consider the causal region $D_{\sigma_*}$ associated with $\sigma_*$. 
For certain values of $w$ (i.e.\ $w \geq 1/3$), $D_{\sigma_*}$ hits the 
big bang singularity.  It is thus more convenient to discuss the ``upper 
half'' of the region:
\begin{equation}
  D^+_{\sigma_*} = \{p \in D_{\sigma_*} \:|\: t(p) \geq t_* \}.
\label{eq:D+sigma}
\end{equation}
In an expanding universe, the extremal surface anchored on the boundary 
of a region $\Gamma$ on $\sigma_*$ is fully contained in this region. 
Now, by performing $t_*$-dependent coordinate transformation
\begin{align}
  \rho &= \frac{2}{3(1+w)} c\, t_*^{-\frac{1+3w}{3(1+w)}} r,
\label{eq:rho}\\
  \eta &= \frac{2}{3(1+w)} \left[ 
    \left( \frac{t}{t_*} \right)^{\frac{1+3w}{3(1+w)}} - 1 \right],
\label{eq:eta}
\end{align}
the region $D^+_{\sigma_*}$ is mapped into
\begin{equation}
  0 \leq \eta \leq 1,
\qquad
  0 \leq \rho \leq 1 - \eta,
\label{eq:D+sigma-new}
\end{equation}
and the metric in $D^+_{\sigma_*}$ is given by
\begin{equation}
  ds^2\big|_{D^+_{\sigma_*}} = \frac{{\cal A}_*}{4\pi} 
    \left( \frac{1+3w}{2}\eta + 1 \right)^{\frac{4}{1+3w}} 
    \left[ -d\eta^2 + d\rho^2 
      + \rho^2 (d\psi^2 + \sin^2\!\psi\, d\phi^2) \right],
\label{eq:met_scal}
\end{equation}
where
\begin{equation}
  {\cal A}_* = 9\pi (1+w)^2 t_*^2.
\label{eq:A_*}
\end{equation}
Since ${\cal A}_*$ appears only as an overall factor of the metric in 
Eqs.~(\ref{eq:D+sigma-new},~\ref{eq:met_scal}), we conclude that the 
${\cal A}_*$ dependence of $S_\Gamma \propto \norm{E_\Gamma}$ is only 
through an overall proportionality factor, as in Eq.~(\ref{eq:S-gamma}).

Due to the scaling in Eq.~(\ref{eq:S-gamma}), it is useful to consider 
an object obtained by dividing $S_\Gamma$ by a quantity that is also 
proportional to ${\cal A}_*$.  We find it convenient to define the quantity
\begin{equation}
  Q_\Gamma \equiv \frac{S_\Gamma}{V_\Gamma/4},
\label{eq:Q_Gamma}
\end{equation}
where $V_\Gamma$ is the (2-dimensional) ``volume'' of the region 
$\Gamma$ or its complement $\bar{\Gamma}$ on the boundary surface 
$\sigma_*$, whichever is smaller.  This quantity is independent 
of ${\cal A}_*$, and hence $t_*$.  For the spherical region of 
Eq.~(\ref{eq:L_gamma}), we find
\begin{equation}
  Q(\gamma) = \frac{S(\gamma)}{V(\gamma)/4} 
    = \frac{\norm{E(\gamma)}}{V(\gamma)},
\label{eq:Q-gamma}
\end{equation}
where
\begin{equation}
  V(\gamma) = \frac{1}{2} \Bigl\{ 1 - {\rm sgn}\Bigl(\frac{\pi}{2} 
    - \gamma\Bigr) \cos\gamma \Bigr\} {\cal A}_*.
\label{eq:V-gamma}
\end{equation}
An explicit expression for $Q(\gamma)$ is given by
\begin{equation}
  Q(\gamma) = \frac{1}{1 - {\rm sgn}(\frac{\pi}{2}-\gamma) \cos\gamma}\,\, 
    \underset{f(x)}{\rm ext} \left[ \int_0^{\sin\gamma}\!\! x\, 
    f^{\frac{4}{1+3w}} \sqrt{1 - \Bigl(\frac{2}{1+3w}\Bigr)^2 
    \Bigl(\frac{df}{dx}\Bigr)^2}\, dx \right],
\label{eq:Q-scaling}
\end{equation}
where the extremization with respect to function $f(x)$ is performed with 
the boundary condition
\begin{equation}
  f(\sin\gamma) = 1,
\label{eq:f-bc}
\end{equation}
and we have used the fact that the extremal surface does not bulge into 
the $z$ direction in the cylindrical coordinates of Eq.~(\ref{eq:cylind}). 
From the point of view of the holographic theory, $Q_\Gamma$ represents 
the amount of entanglement entropy per degree of freedom as viewed 
from the smaller of $\Gamma$ and $\bar{\Gamma}$.  As we will discuss 
in Section~\ref{subsec:nonlocal}, the fact that this is a physically 
significant quantity has important implications for the structure of 
the holographic theory.

\begin{figure}[t]
\begin{center}
  \includegraphics[height=6.5cm]{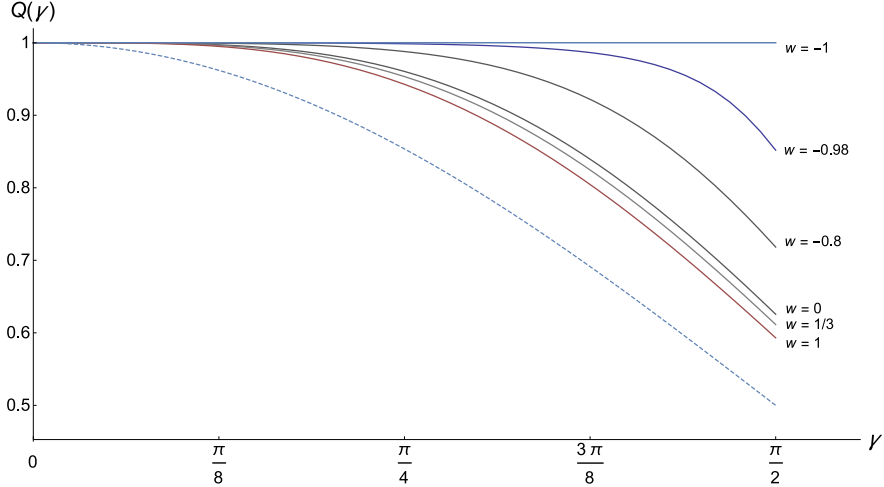}
\end{center}
\caption{The value of $Q(\gamma)$ as a function of $\gamma$ ($0 \leq 
 \gamma \leq \pi/2$) for $w = -1$ (vacuum energy), $-0.98$, $-0.8$, 
 $0$ (matter), $1/3$ (radiation), and $1$.  The dotted line indicates 
 the lower bound given by the flat space geometry, which can be realized 
 in a curvature dominated open FRW universe.}
\label{fig:Q-gamma}
\end{figure}
In Fig.~\ref{fig:Q-gamma}, we plot $Q(\gamma)$ as a function of $\gamma$ 
($0 \leq \gamma \leq \pi/2$) for various values of $w$:\ $-1$ (vacuum 
energy), $-0.98$, $-0.8$, $0$ (matter), $1/3$ (radiation), and $1$. 
The value of $Q(\gamma)$ for $\pi/2 \leq \gamma \leq \pi$ is given by 
$Q(\gamma) = Q(\pi - \gamma)$.  We find the following features:
\begin{itemize}
\item
In the limit of a small boundary region, $\gamma \ll 1$, the value of 
$Q(\gamma)$ approaches unity regardless of the value of $w$:
\begin{equation}
  Q_w(\gamma) \xrightarrow{\gamma \ll 1} 1.
\label{eq:Q-gamma-small}
\end{equation}
This implies that for a small boundary region, the entanglement entropy 
of the region is given by its volume in the holographic theory in 
Planck units:
\begin{equation}
  S_w(\gamma) \xrightarrow{\gamma \ll 1} \frac{1}{4} V(\gamma).
\label{eq:S-volume}
\end{equation}
For larger $\gamma$ ($\leq \pi/2$), $Q(\gamma)$ becomes monotonically 
small as $\gamma$ increases:
\begin{equation}
  \frac{d}{d\gamma} Q_w(\gamma) < 0.
\label{eq:Q-gamma-der}
\end{equation}
The deviation of $Q(\gamma)$ from $1$ near $\gamma = 0$ is given by
\begin{equation}
  Q_w(\gamma) \stackrel{\gamma \ll 1}{=} 1 - c\, (1+w) \gamma^4 + \cdots,
\label{eq:Q-small-gamma}
\end{equation}
where $c > 0$ is a constant that does not depend on $w$.
\item
For any fixed boundary region, $\gamma$, the value of $Q(\gamma)$ decreases 
monotonically in $w$:
\begin{equation}
  \frac{d}{dw} Q_w(\gamma) < 0.
\label{eq:Q-w-der}
\end{equation}
In particular, when $w$ approaches $-1$ (from above), $Q(\gamma)$ 
becomes unity:
\begin{equation}
  \lim_{w \rightarrow -1} Q_w(\gamma) = 1.
\label{eq:Q-w_-1}
\end{equation}
This implies that in the limit of de~Sitter FRW ($w \rightarrow -1$), 
the state in the holographic theory becomes ``randomly entangled'' 
(i.e.\ saturates the Page curve~\cite{Page:1993df}):%
\footnote{In the case of an exactly single component with $w = -1$, 
 the expansion of light rays emanating from $p_0$, i.e.\ $\theta_k$, 
 becomes $0$ only at infinite affine parameter $\lambda$.  We view this 
 as a result of mathematical idealization.  A realistic de~Sitter FRW 
 universe is obtained by introducing an infinitesimally small amount 
 of matter in addition to the $w=-1$ component, which avoids the above 
 issue.  The results obtained in this way agree with those by first 
 taking $w > -1$ and then the limit $w \rightarrow -1$.}
\begin{equation}
  \lim_{w \rightarrow -1} S_w(\gamma) = \frac{1}{4} V(\gamma).
\label{eq:S-w_-1}
\end{equation}
Note that $V(\gamma)$ is the smaller of the volume of $L(\gamma)$ and 
that of its complement on the leaf.  The value of $Q(\pi/2)$ (the case 
in which $L(\gamma)$ is a half of the leaf) is plotted as a function 
of $w$ in Fig.~\ref{fig:Q-w}.
\end{itemize}
We will discuss further implications of these findings in 
Section~\ref{sec:beyond}.
\begin{figure}[t]
\begin{center}
  \includegraphics[height=6.5cm]{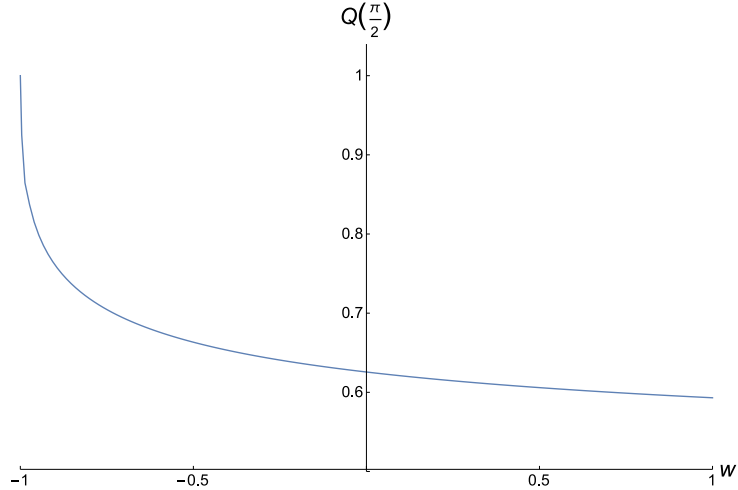}
\end{center}
\caption{The value of $Q(\pi/2)$ as a function of $w$.}
\label{fig:Q-w}
\end{figure}

We note that there are simple geometric bounds on the values 
of $Q_w(\gamma)$.  This can be seen by adopting the maximin 
construction~\cite{Sanches:2016sxy,Wall:2012uf}:\ the extremal surface 
is the one having the maximal area among all possible codimension-2 
surfaces each of which is anchored on $\partial L(\gamma)$ and has 
minimal area on some interior achronal hypersurface bounded by $\sigma$. 
This implies that the area of the extremal surface, $\norm{E(\gamma)}$, 
cannot be larger than the boundary volume $V(\gamma)$, giving 
$Q(\gamma) \leq 1$.  Also, the extremal surface cannot have a 
smaller area than the codimension-2 surface that has the minimal 
area on a constant time hypersurface $t=t_*$:\ $\norm{E(\gamma)} 
\geq \pi \{ a(t_*) r_{\rm AH}(t_*) \sin\gamma \}^2$.  Together, 
we obtain
\begin{equation}
  \frac{\sin^2\!\gamma}{2\bigl\{ 1 - {\rm sgn}(\frac{\pi}{2}-\gamma) 
    \cos\gamma \bigr\}} \leq Q_w(\gamma) \leq 1.
\label{eq:Q-bounds}
\end{equation}
The lower edge of this range is depicted by the dashed line 
in Fig.~\ref{fig:Q-gamma}.  We find that the upper bound of 
Eq.~(\ref{eq:Q-bounds}) can be saturated with $w \rightarrow -1$, 
while the lower bound cannot with $|w| \leq 1$.  If we formally 
take $w \rightarrow +\infty$, the lower bound can be reached.  A 
fluid with $w > 1$, however, does not satisfy the causal energy 
condition (although it satisfies the null energy condition), so 
we do not consider such a component.

\begin{figure}[t]
\begin{center}
  \includegraphics[height=6.5cm]{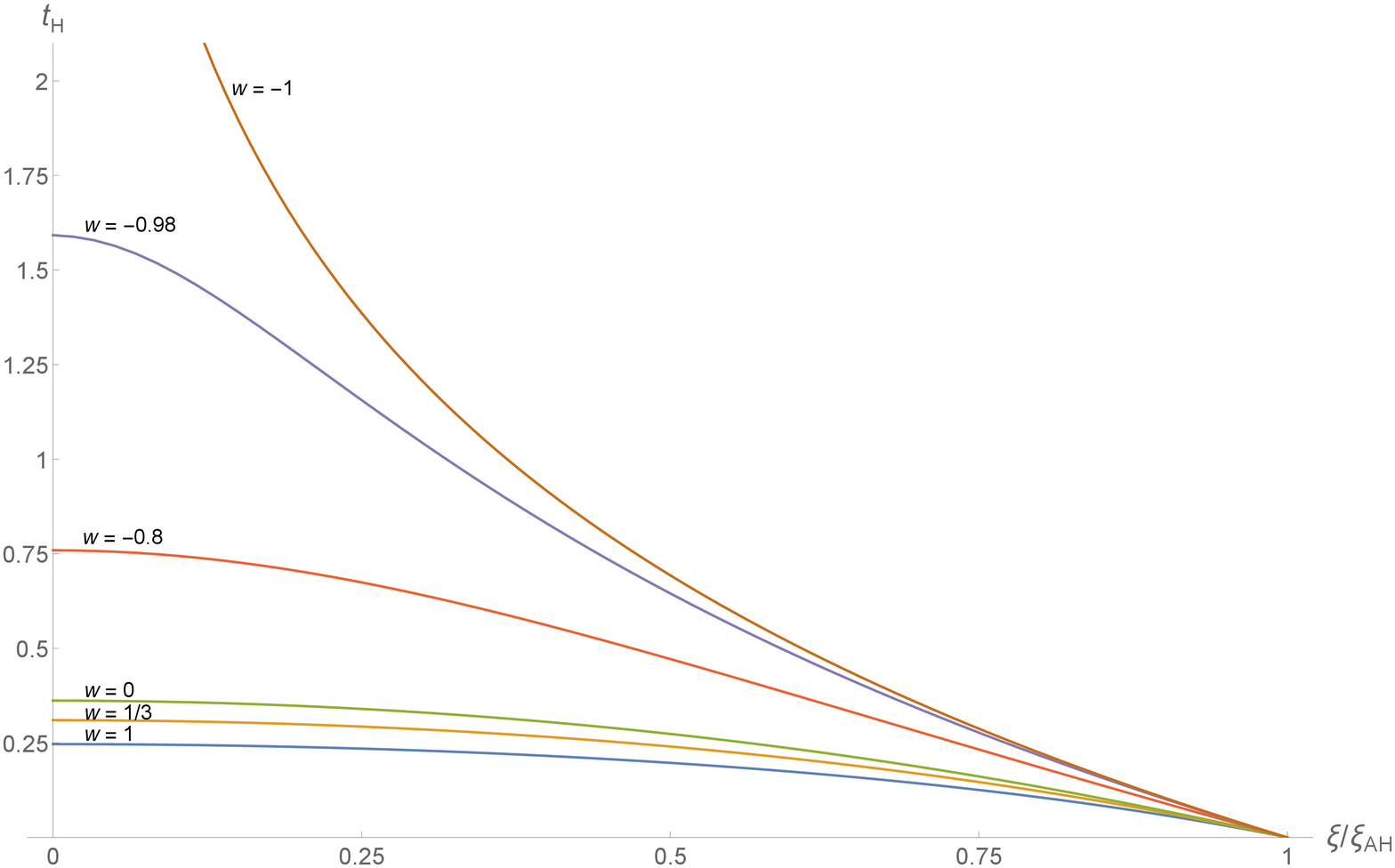}
\end{center}
\caption{The shape of the extremal surfaces $E(\pi/2)$ for $w = -1$, 
 $-0.98$, $-0.8$, $0$, $1/3$, and $1$.  The horizontal axis is the 
 cylindrical radial coordinate normalized by the apparent horizon 
 radius, $\xi/\xi_{\rm AH}$, and the vertical axis is the Hubble 
 time, $t_{\rm H}$.}
\label{fig:ext-surf}
\end{figure}
As a final remark, we show in Fig.~\ref{fig:ext-surf} the shape of 
the extremal surface for $\gamma = \pi/2$ for the same values of $w$ 
as in Fig.~\ref{fig:Q-gamma}:\ $-1$, $-0.98$, $-0.8$, $0$, $1/3$, 
and $1$.  The horizontal axis is the cylindrical radial coordinate 
normalized by the apparent horizon radius, $\xi/\xi_{\rm AH}$, and 
the vertical axis is taken to be the Hubble time defined by
\begin{equation}
  t_{\rm H} = \int_{t_*}^t \frac{\dot{a}(t)}{a(t)}\, dt 
  = \frac{2}{3(1+w)} \ln \frac{t}{t_*},
\label{eq:t_H}
\end{equation}
which reduces in the $w \rightarrow -1$ limit to the usual Hubble time 
$t_{\rm H} = H (t - t_*)$, where $H = \dot{a}/a$.  We find that the 
extremal surface bulges into the future direction for any $w$.  In fact, 
this occurs generally in an expanding universe and can be understood from 
the maximin construction:\ the scale factor increases toward the future, 
so that the area of the minimal area surface on an achronal hypersurface 
increases when the hypersurface bulges into the future direction in 
time.  The amount of the bulge is $t_{\rm H} \approx O(1)$, except when 
$w \approx -1$.  For $w \rightarrow -1$, the extremal surface probes 
$t_H \rightarrow +\infty$ as $\xi/\xi_{\rm AH} \rightarrow +0$, but 
its area is still finite, $\norm{E(\pi/2)} \rightarrow {\cal A}_*/2$, 
as the surface becomes almost null in this limit.

\subsubsection*{An open FRW universe dominated by curvature}

We now consider an open FRW universe dominated by curvature, i.e.\ the 
case in which the expansion of the universe is determined by the second 
term in the left-hand side of Eq.~(\ref{eq:Friedmann-eq}).  This implies 
that the distance to the apparent horizon is much larger than the 
curvature length scale
\begin{equation}
  \frac{-\kappa}{a^2(t)} \gg \frac{8\pi}{3} \rho(t)
\quad\Longleftrightarrow\quad
  r_{\rm AH}(t) \gg \frac{1}{\sqrt{-\kappa}} \equiv r_{\rm curv}.
\label{eq:curv-dom}
\end{equation}
(Note that $\kappa < 0$ for an open universe.)  As seen in 
Eqs.~(\ref{eq:Friedmann-eq},~\ref{eq:r_AH}), the value of $r_{\rm AH}(t)$ 
is determined by $\rho(t)$, which gives only a minor contribution to the 
expansion of the universe.  The scale factor is given by
\begin{equation}
  a(t) = \sqrt{-\kappa}\, t.
\label{eq:at-curv}
\end{equation}

The extremal surface can be found easily by noticing that the universe 
in this limit is a hyperbolic foliation of a portion of the Minkowski 
space:\ the coordinate transformation
\begin{align}
  \tilde{t} &= t \sqrt{1+ \bigl(\sqrt{-\kappa}\,r\bigr)^2},
\\
  \tilde{r} &= \sqrt{-\kappa}\, t\, r,
\end{align}
leads to the Minkowski metric $ds^2 = -d\tilde{t}^2 + d\tilde{r}^2 
+ \tilde{r}^2 (d\psi^2 + \sin^2\!\psi\, d\phi^2)$.  The extremal surface 
is thus a plane on a constant $\tilde{t}$ hypersurface, which in the 
FRW (cylindrical) coordinates is given by
\begin{equation}
  t_{\rm H} \approx \ln\frac{1}{\xi/\xi_{\rm AH}} 
\qquad
  (0 \leq \xi/\xi_{\rm AH} \leq 1),
\label{eq:ext-AdS}
\end{equation}
where $\xi_{\rm AH} = r_{\rm AH}(t_*) \sin\gamma$, and $t_{\rm H}$ is 
the Hubble time
\begin{equation}
  t_{\rm H} = \int_{t_*}^t \frac{\dot{a}(t)}{a(t)}\, dt 
  = \ln \frac{t}{t_*}.
\label{eq:t_H-AdS}
\end{equation}
The resulting $Q(\gamma)$ is
\begin{equation}
  Q(\gamma) \approx \frac{\sin^2\!\gamma}{2\bigl\{ 1 
    - {\rm sgn}(\frac{\pi}{2}-\gamma) \cos\gamma \bigr\}}.
\label{eq:Q-AdS-sol}
\end{equation}
This, in fact, saturates the lower bound in Eq.~(\ref{eq:Q-bounds}), 
plotted as the dashed line in Fig.~\ref{fig:Q-gamma}.

\subsection{Dynamics of screen entanglement entropies in a transition}
\label{subsec:transit}

Let us consider the evolution of an FRW universe.  From the holographic 
theory point of view, it is described by a time-dependent state 
$\ket{\Psi(\tau)}$ living on $\sigma(\tau)$.  Because of the area 
theorem of Refs.~\cite{Bousso:2015mqa,Bousso:2015qqa}, we can take 
$\tau$ to be a monotonic function of the leaf area, leading to
\begin{equation}
  \frac{d}{d\tau} {\cal A}(\tau) > 0,
\label{eq:A-dot}
\end{equation}
where ${\cal A}(\tau) \equiv \norm{\sigma(\tau)}$.  This evolution 
involves a change in the number of (effective) degrees of freedom, 
${\cal A}(\tau)/4$, as well as that of the structure of entanglement 
on the boundary, $Q_\Gamma(\tau)$.  For the latter, we mostly consider 
$Q(\gamma, \tau)$ associated with a spherical cap region $\Gamma 
= L(\gamma)$.  A natural question is if a statement similar to 
Eq.~(\ref{eq:A-dot}) applies for screen entanglement entropies:
\begin{equation}
  \frac{d}{d\tau} S(\gamma, \tau) \stackrel{?}{>} 0.
\label{eq:S-dot}
\end{equation}
Here,
\begin{equation}
  S(\gamma,\tau) = Q(\gamma,\tau) \frac{V(\gamma,\tau)}{4},
\label{eq:S_tau}
\end{equation}
with
\begin{equation}
  V(\gamma,\tau) = \frac{1}{2} \Bigl\{ 1 - {\rm sgn}\Bigl(\frac{\pi}{2} 
    - \gamma\Bigr) \cos\gamma \Bigr\} {\cal A}(\tau),
\label{eq:V_tau}
\end{equation}
being the smaller of the boundary volumes of $L(\gamma)$ and its 
complement.

There are some cases in which we can show that the relation in 
Eq.~(\ref{eq:S-dot}) is indeed satisfied.  Consider, for example, 
a flat FRW universe filled with various fluid components having 
differing equations of states:\ $w_i$ ($i = 1,2,\cdots$).  As time 
passes, the dominant component of the universe changes from one 
having larger $w$ to one having smaller $w$ successively.  This 
implies that $Q(\gamma, \tau)$ monotonically increases in time, so 
that Eq.~(\ref{eq:A-dot}) indeed implies Eq.~(\ref{eq:S-dot}) in this 
case.  Another interesting case is when the holographic screen is 
spacelike.  In this case, we can prove that the time dependence of 
$S(\gamma, \tau)$ is monotonic; see Appendix~\ref{app:spacelike}. 
In particular, if we have a spacelike past holographic screen (which 
occurs for $w > 1/3$ in a single-component dominated flat FRW universe), 
then the screen entanglement entropy for an arbitrary region increases 
in time:\ $dS_\Gamma(\tau)/d\tau > 0$.

What happens if the holographic screen is timelike?  One might 
think that there is an obvious argument against the inequality in 
Eq.~(\ref{eq:S-dot}).  Suppose the expansion of the early universe 
is dominated by a fluid component with $w$.  Suppose at some FRW 
time $t_0$ this component is converted into another fluid component 
having a different equation of state parameter $w'$, e.g.\ by 
reheating associated with the decay of a scalar condensate.  If 
$w' > w$, then the $Q$ value after the transition is smaller than 
that before
\begin{equation}
  Q_{w'}(\gamma) - Q_w(\gamma) < 0.
\label{eq:Delta-Q}
\end{equation}
One may think that this can easily overpower the increase of $S(\gamma, 
\tau)$ from the increase of the area:\ $d{\cal A}(\gamma, \tau)/d\tau 
> 0$~\cite{Sanches:2016pga}.  In particular, if $w$ is close to $-1$, 
then the increase of the area before the transition is very slow, so 
that the effect of Eq.~(\ref{eq:Delta-Q}) would win over that of the 
area increase.  However, as depicted in Fig.~\ref{fig:ext-surf}, when 
$w \approx -1$ the extremal surface bulges into larger $t$ by many Hubble 
times.  Hence the time between the moments in which Eq.~(\ref{eq:S_tau}) 
can be used before and after the transition becomes long, opening the 
possibility that the relevant area increase is non-negligible.

\begin{figure}[t]
\begin{center}
  \includegraphics[height=6cm]{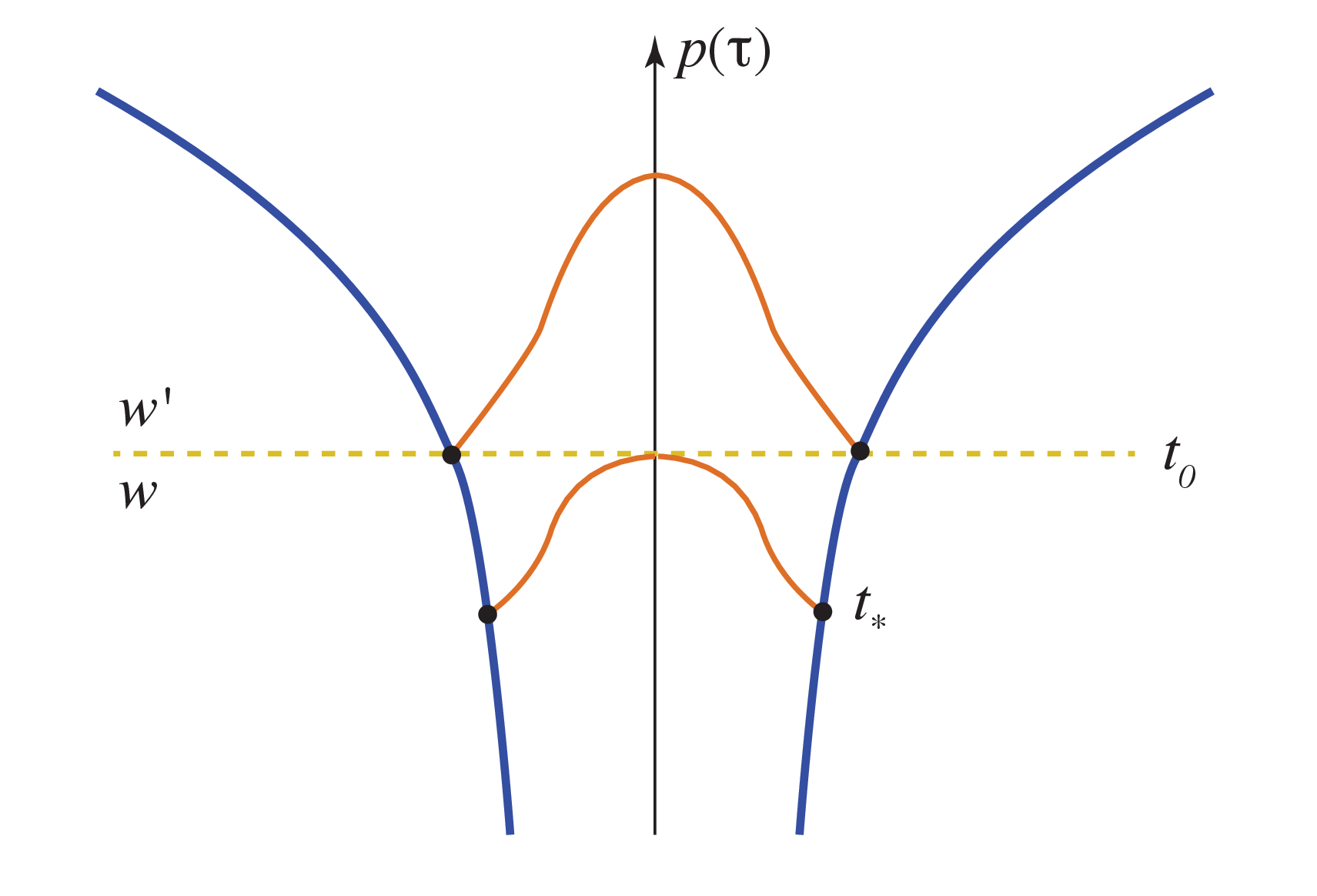}
\end{center}
\caption{An FRW universe whose dominant component changes from $w$ to $w'$ 
 at time $t_0$.  Two surfaces depicted by orange lines are the latest 
 extremal surface fully contained in the $w$ region (bottom) and the 
 earliest extremal surface fully contained in the $w'$ region (top), 
 each anchored to the leaves at $t_*$ and $t_0$.}
\label{fig:trans}
\end{figure}
To make the above discussion more explicit, let us compare the values 
of the screen entanglement entropy $S(\gamma)$ corresponding to two 
extremal surfaces depicted in Fig.~\ref{fig:trans}:\ the ``latest'' 
extremal surface that is fully contained in the $w$ region and the 
``earliest'' extremal surface fully contained in the $w'$ region, each 
anchored to the leaves at FRW times $t_*$ and $t_0$.  This provides the 
most stringent test of the inequality in Eq.~(\ref{eq:S-dot}) that can 
be performed using the expression of Eq.~(\ref{eq:S_tau}) for fixed 
$w$'s.  The ratio of the entanglement entropies is given by
\begin{equation}
  R_{w' w}(\gamma) 
  \equiv \frac{S_{\rm after}(\gamma)}{S_{\rm before}(\gamma)} 
  = \frac{Q_{w'}(\gamma)}{Q_w(\gamma)} \frac{t_0^2}{t_*^2} 
  = \frac{Q_{w'}(\gamma)}{Q_w(\gamma)}\, e^{3(1+w) t_{{\rm H},w}},
\label{eq:S-ratio}
\end{equation}
where $t_{{\rm H},w}$ is the Hubble time between $t_*$ and $t_0$, given 
by Eq.~(\ref{eq:t_H}) with $t \rightarrow t_0$.  In Fig.~\ref{fig:Delta-S}, 
we plot $R_w \equiv R_{1 w}(\pi/2)$; setting $w' = 1$ minimizes the ratio.
\begin{figure}[t]
\begin{center}
  \includegraphics[height=6.5cm]{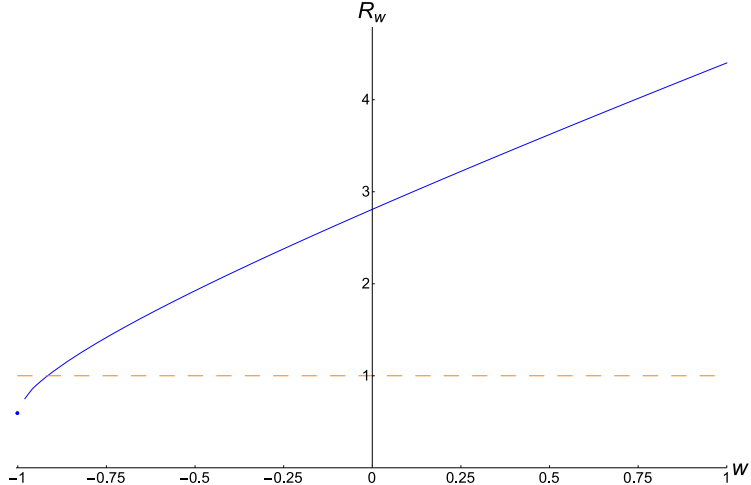}
\end{center}
\caption{The ratio of the screen entanglement entropies, $R_w = 
 R_{1 w}(\pi/2)$, before and after the transition from a universe 
 with the equation of state parameter $w$ to that with $w' = 1$, 
 obtained from Figs.~\ref{fig:Q-w} and \ref{fig:ext-surf} using 
 Eq.~(\ref{eq:S-ratio}).  The dot at $w = -1$ represents $R_{-1} 
 = R_{1 -1}(\pi/2)$ obtained in Eq.~(\ref{eq:S-ratio_-1}).}
\label{fig:Delta-S}
\end{figure}
We find that this ratio can be smaller than $1$ for $w \approx -1$.  In 
fact, for $w \rightarrow -1$ we find the value obtained naively by assuming 
that the area does not change before the transition:
\begin{equation}
  R_{-1} 
  = \frac{Q_{1}\bigl(\frac{\pi}{2}\bigr)}{Q_{-1}\bigl(\frac{\pi}{2}\bigr)} 
  = Q_{1}\Bigl(\frac{\pi}{2}\Bigr),
\label{eq:S-ratio_-1}
\end{equation}
although for $w = -1$ there is no such thing as the latest extremal 
surface that is fully contained in the region before the transition 
(since $t_{{\rm H},-1} = +\infty$).

This analysis suggests that screen entanglement entropies can in fact 
drop if the system experiences a rapid transition induced by some 
dynamics,%
\footnote{This does not mean that the second law of thermodynamics 
 is violated. The entropy discussed here is the von~Neumann entropy 
 of a significant portion (half) of the whole system, which can 
 deviate from the thermodynamic entropy of the region when the 
 system experiences a rapid change.}
although the instantaneous transition approximation adopted above 
is not fully realistic.  Of course, such a drop is expected to be 
only a temporary phenomenon---because of the area increase after 
the transition, the entropy generally returns back to the value 
before the transition in a characteristic dynamical timescale and 
then continues to increase afterward.  We expect that the relation 
in Eq.~(\ref{eq:S-dot}) is valid in a coarse-grained sense
\begin{equation}
  \frac{d}{d\tau} \bar{S}(\gamma, \tau) > 0;
\qquad
  \bar{S}(\gamma, \tau) = \frac{1}{\tau_c} 
    \int_\tau^{\tau+\tau_c}\! S(\gamma, \tau')\, d\tau',
\label{eq:S-dot-coarse}
\end{equation}
but not ``microscopically'' in general.  Here, $\tau_c$ must be taken 
sufficiently larger than the characteristic dynamical timescale, 
the Hubble time for an FRW universe.

For further illustration, we perform numerical calculations for how the 
area of a leaf hemisphere, $\norm{L(\pi/2,t)}$, and the associated screen 
entanglement entropy, calculated using $S(\pi/2,t) = \norm{E(\pi/2,t)}/4$, 
evolve in time during transitions from a $w = -1$ to a $w' = 0$ flat FRW 
universe.  Here, we take the FRW time $t$ as the time parameter.  For 
this purpose, we consider a scalar field $\phi$ having a potential 
$V(\phi)$ that has a flat portion and a well, with the initial value 
of $\phi$ being in the flat portion.  We first note that a transformation 
of the potential of the form
\begin{equation} 
  V(\phi) \rightarrow V'(\phi) = \epsilon^2 V(\phi),
\label{eq:rescale_V}
\end{equation}
leads to rescalings of the scalar field, $\phi(t)$, and the scale factor, 
$a(t)$, obtained as the solutions to the equations of motion:
\begin{equation}
  \phi'(t) = \phi(\epsilon t),
\qquad
  a'(t) = a(\epsilon t).
\label{eq:rescale_phi-a}
\end{equation}
Plugging these in Eq.~(\ref{eq:norm-E}), we find that the area 
functionals before and after the transformation Eq.~(\ref{eq:rescale_V}) 
are related by simple rescaling $t \rightarrow t/\epsilon$ and 
$\xi \rightarrow \xi/\epsilon$, so that
\begin{equation}
  \Bigl\lVert E'\Bigl(\frac{\pi}{2},t\Bigr) \Bigr\lVert
  = \frac{1}{\epsilon^2} \Bigl\lVert 
    E\Bigl(\frac{\pi}{2},\frac{t}{\epsilon}\Bigr) \Bigr\lVert.
\label{eq:rescale_E}
\end{equation}
These scaling properties imply that the leaf hemisphere area and the 
screen entanglement entropy for the transformed potential are read 
off from those for the untransformed one by
\begin{equation}
  \Bigl\lVert L'\Bigl(\frac{\pi}{2},t\Bigr) \Bigr\lVert
  = \frac{1}{\epsilon^2} \Bigl\lVert 
    L\Bigl(\frac{\pi}{2},\frac{t}{\epsilon}\Bigr) \Bigr\lVert,
\qquad
  \Bigl\lVert S'\Bigl(\frac{\pi}{2},t\Bigr) \Bigr\lVert
  = \frac{1}{\epsilon^2} \Bigl\lVert 
    S\Bigl(\frac{\pi}{2},\frac{t}{\epsilon}\Bigr) \Bigr\lVert.
\label{eq:rescale_L-S}
\end{equation}
We therefore need to be concerned only with the shape of the potential, 
not its overall scale.  In particular, we can always be in the semiclassical 
regime by performing a transformation with $\epsilon \ll 1$.

\begin{figure}[]
  \setcounter{subfigure}{0}
  \subfigure[Steep potential.]
     {\includegraphics[height=4.4cm]{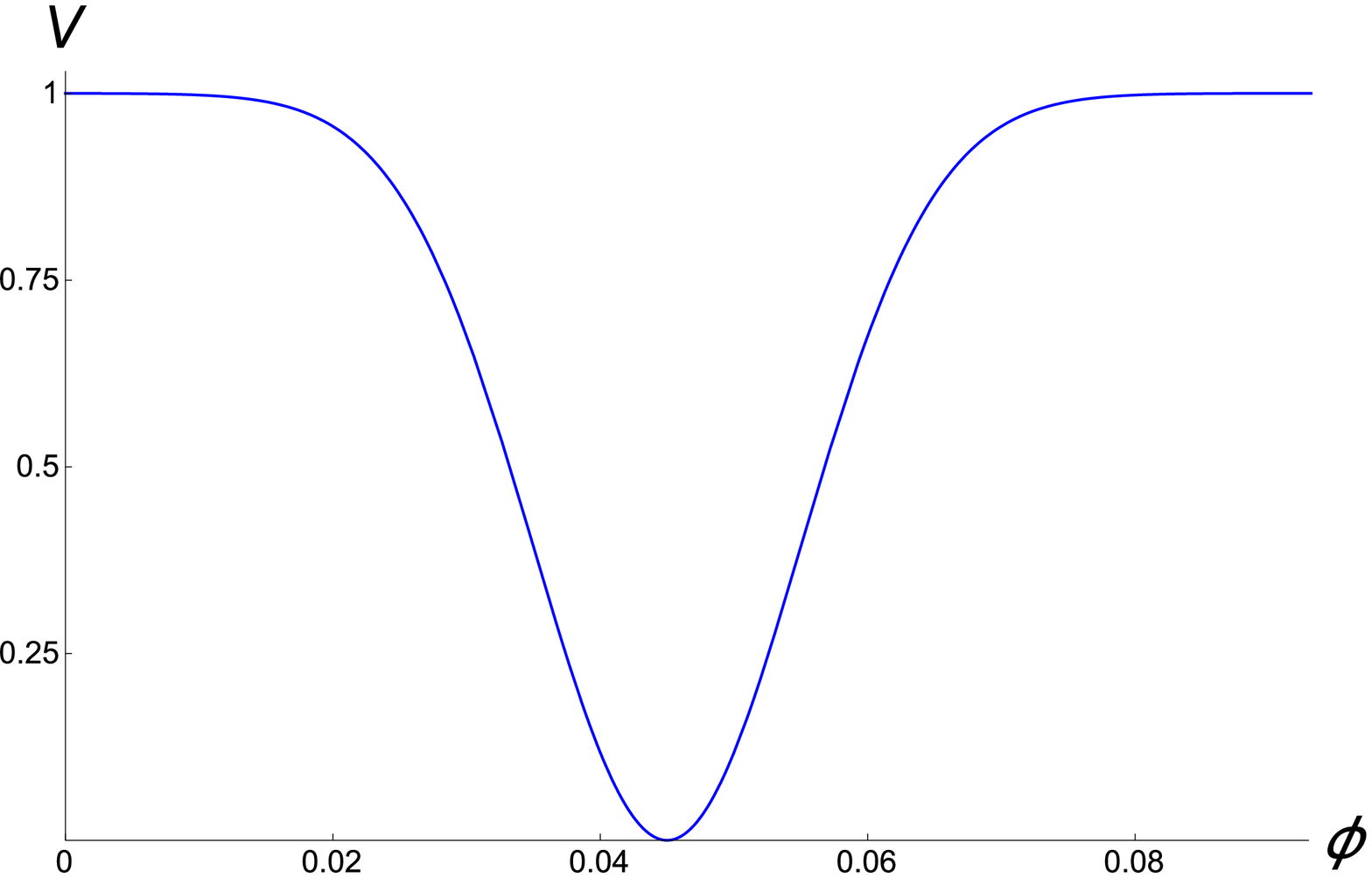}}
  \setcounter{subfigure}{4}
  \subfigure[Broad potential.]
     {\includegraphics[height=4.4cm]{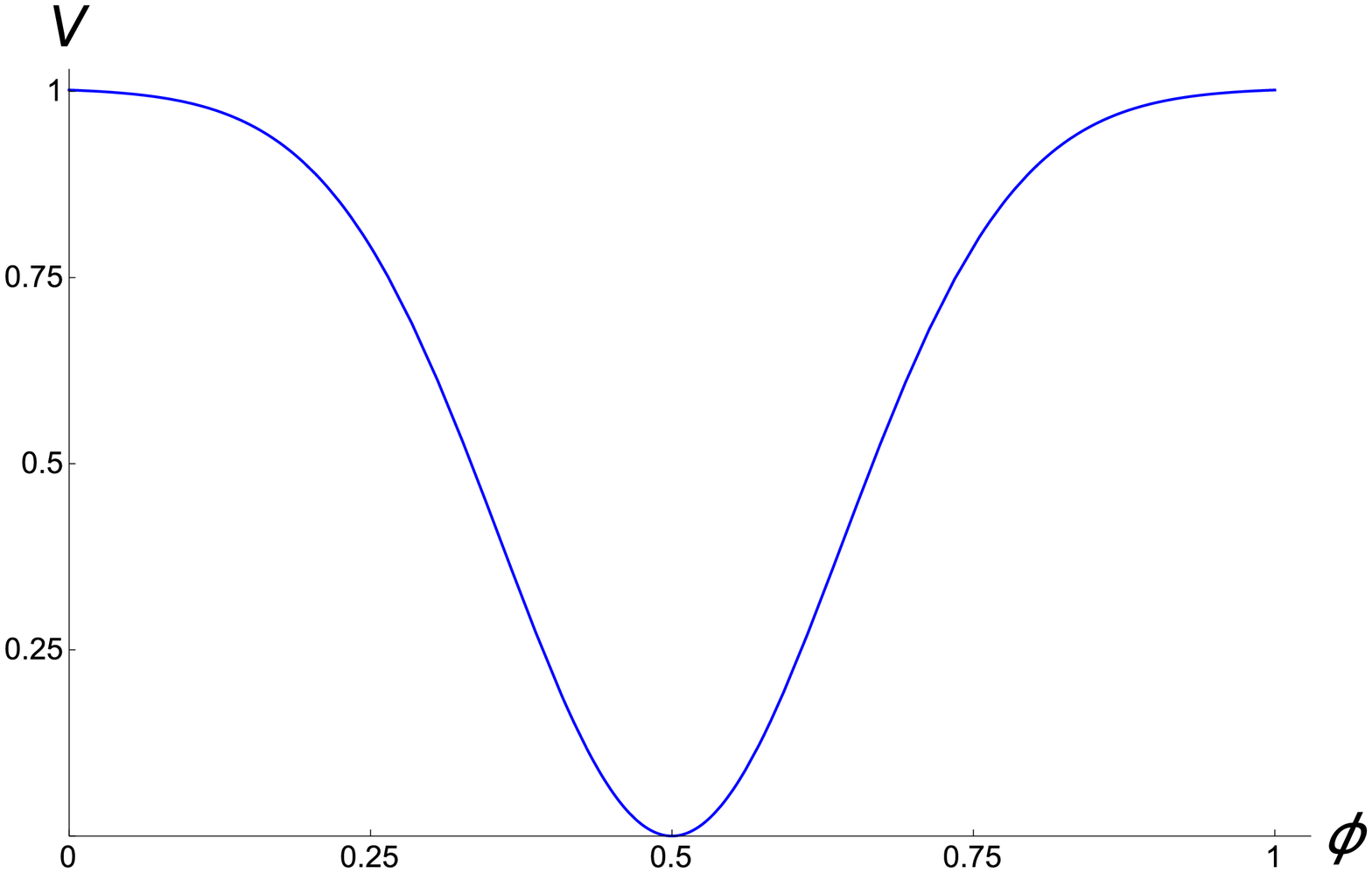}}
  \setcounter{subfigure}{1}
  \subfigure[$\phi(t)$ for the steep potential.]
     {\includegraphics[height=4.4cm]{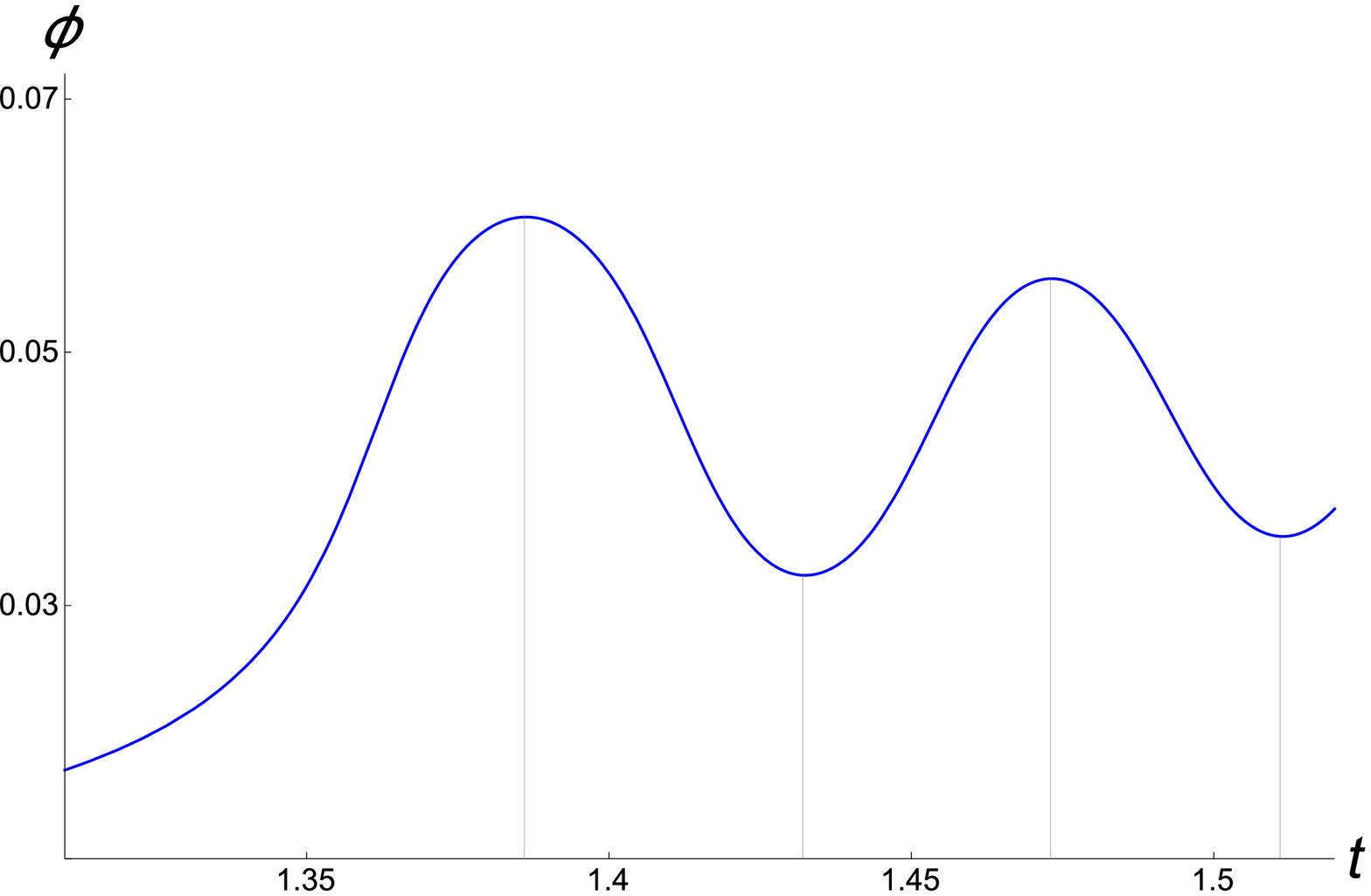}}
  \setcounter{subfigure}{5}
  \subfigure[$\phi(t)$ for the broad potential.]
     {\includegraphics[height=4.4cm]{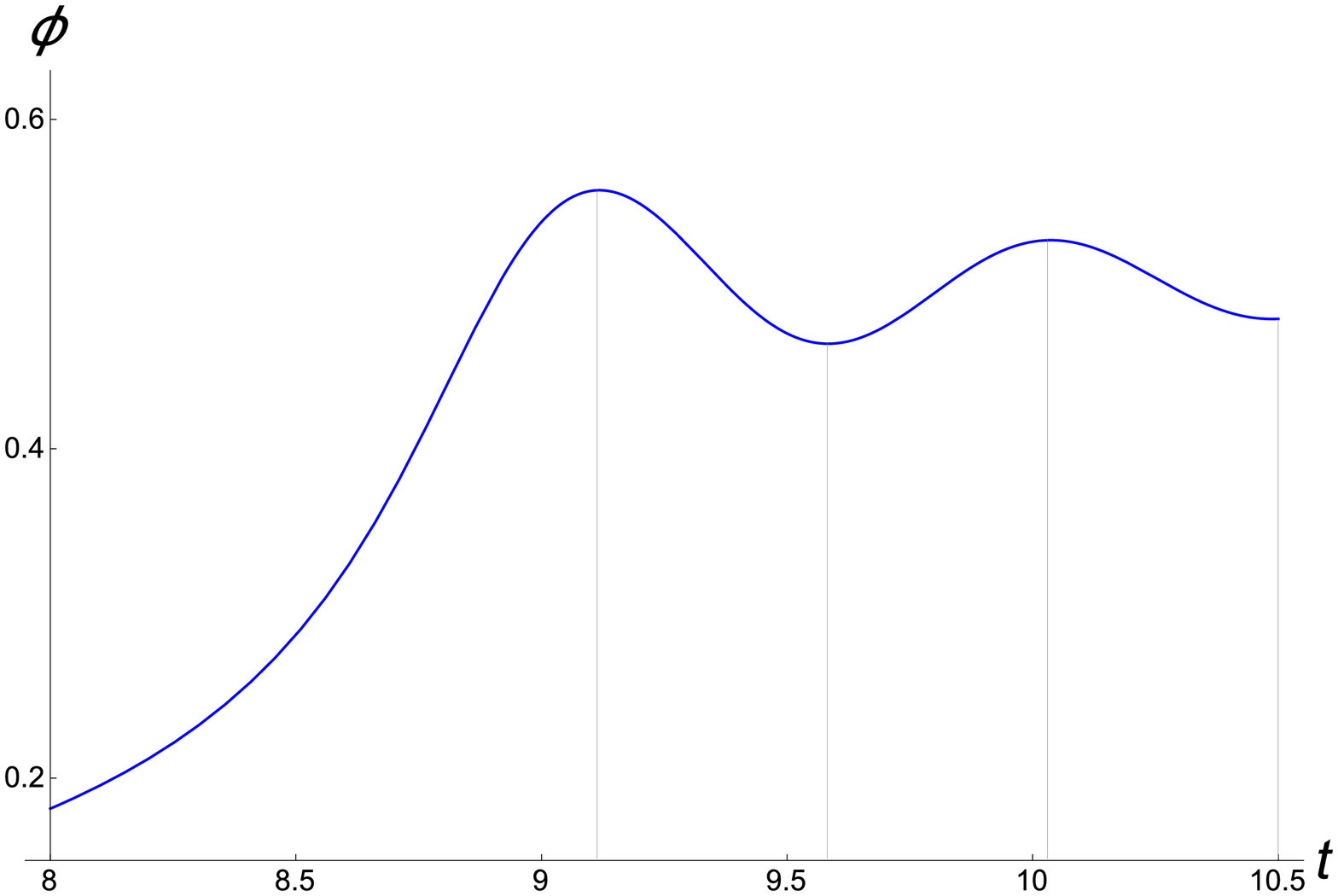}}
  \setcounter{subfigure}{2}
  \subfigure[$\norm{L(\pi/2,t)}$ for the steep potential.]
     {\includegraphics[height=4.4cm]{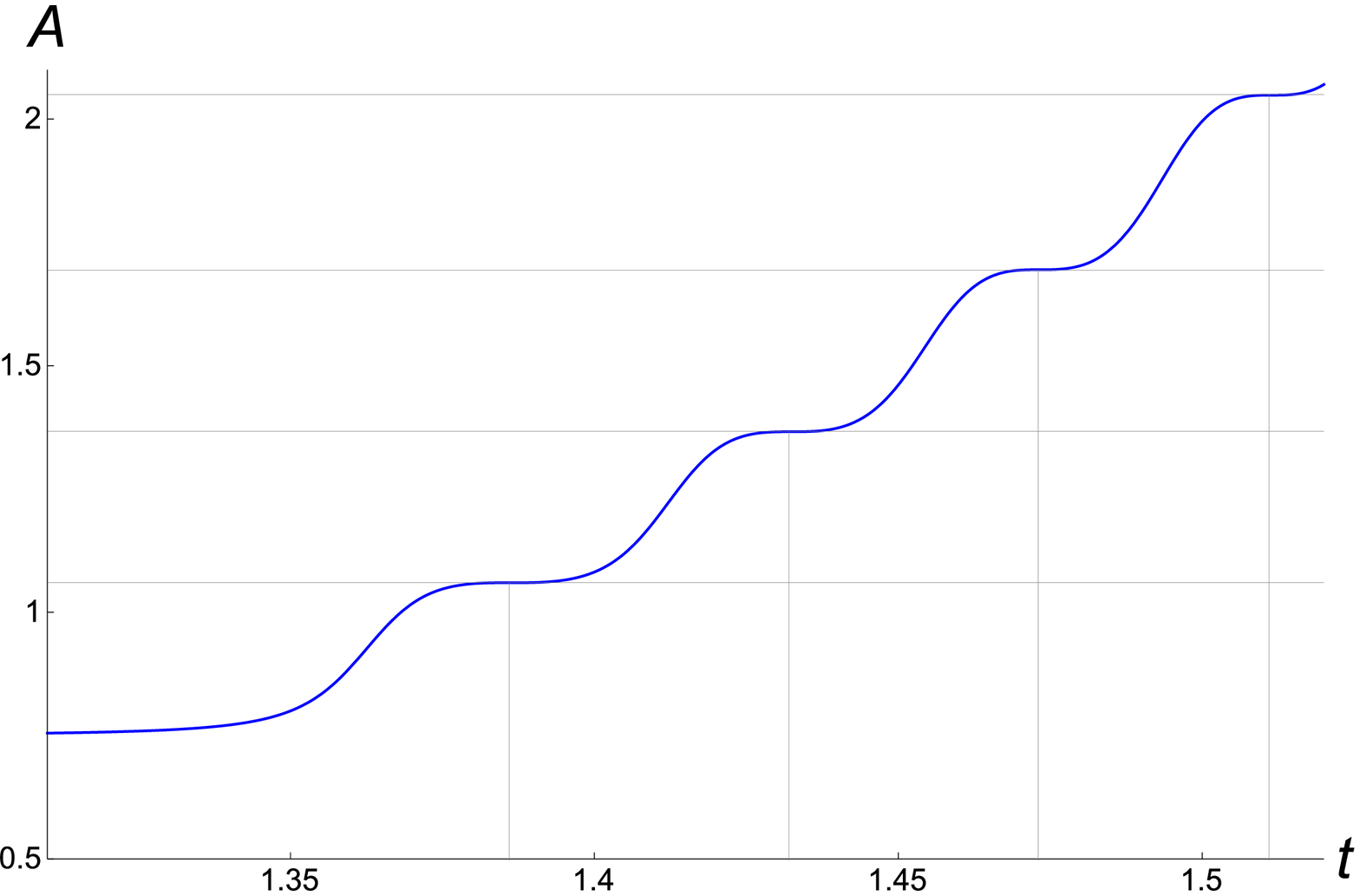}}
  \setcounter{subfigure}{6}
  \subfigure[$\norm{L(\pi/2,t)}$ for the broad potential.]
     {\includegraphics[height=4.4cm]{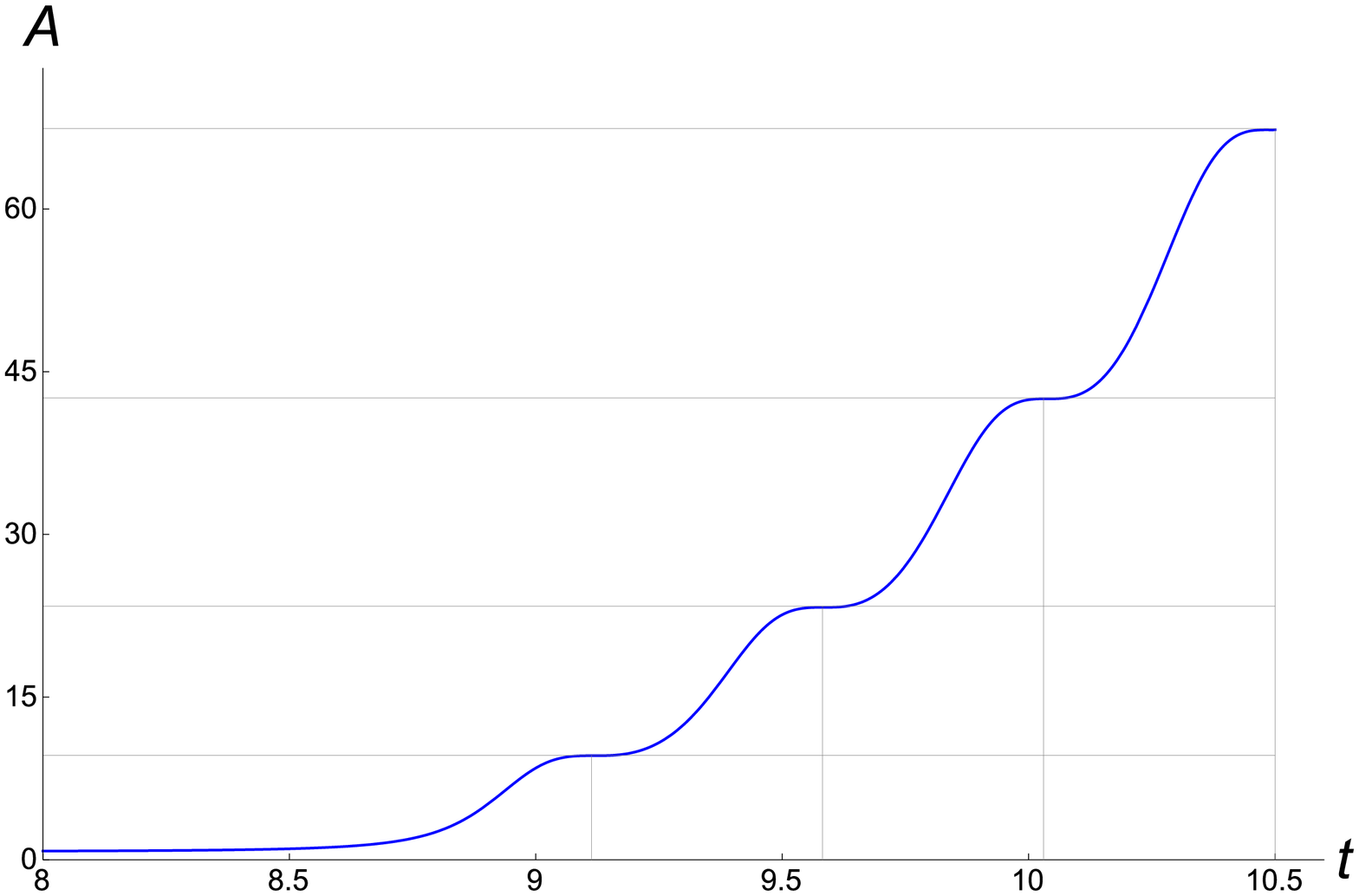}}
  \setcounter{subfigure}{3}
  \subfigure[$S(\pi/2,t)$ for the steep potential.]
     {\includegraphics[height=4.4cm]{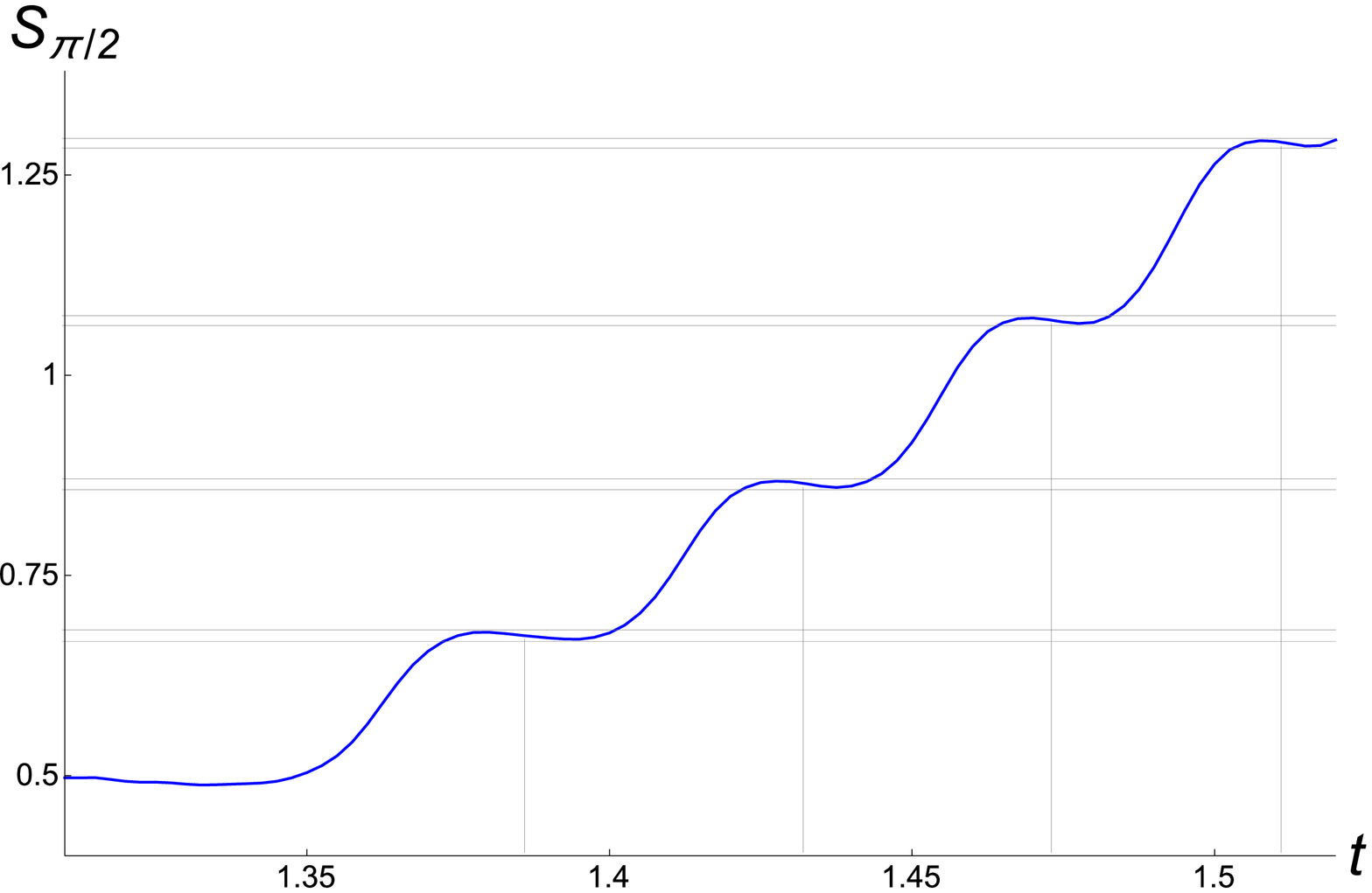}}
  \hspace{29.5mm}
  \setcounter{subfigure}{7}
  \subfigure[$S(\pi/2,t)$ for the broad potential.]
     {\includegraphics[height=4.4cm]{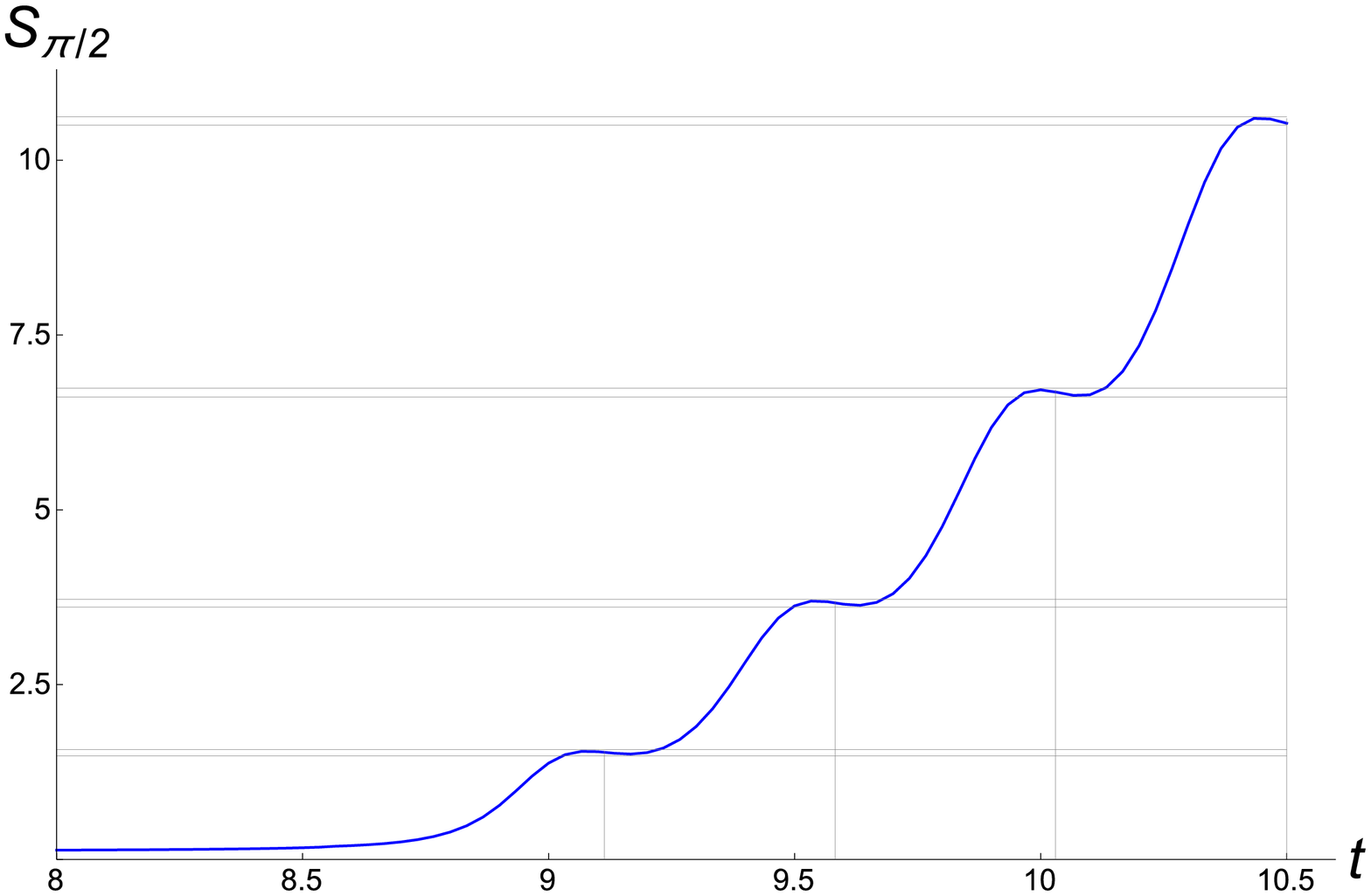}}
\caption{A steep potential~(a) leading to the time evolution of the 
 scalar field~(b), the area of a leaf hemisphere~(c), and the screen 
 entanglement entropy~(d).  The same for a broad potential~(e)--(h).}
\label{fig:oscil}
\end{figure}
In Fig.~\ref{fig:oscil}, we show the results of our calculations for 
``steep'' and ``broad'' potentials.  The explicit forms of the potentials 
are given by
\begin{equation}
  V(\phi) = 1 - e^{-k (\phi - \phi_0)^2} 
    + s (\phi - \phi_0) \tanh(p (\phi - \phi_0)),
\label{eq:potentials}
\end{equation}
with
\begin{align}
  \mbox{Steep} &: \quad 
  k = 5000,\quad s = 0.01,\quad p = 20,\quad \phi_0 = 0.045,
\label{eq:V-steep}\\
  \mbox{Broad} &: \quad 
  k = 25,\quad s = 0.01,\quad p = 2,\quad \phi_0 = 0.5,
\label{eq:V-broad}
\end{align}
although their detailed forms are unimportant.  For the steep 
potential, plotted in Fig.~\ref{fig:oscil}(a), we show the time 
evolutions of $\phi(t)$, $\norm{L(\pi/2,t)}$, and $S(\pi/2,t)$ 
in Figs.~\ref{fig:oscil}(b)--(d) for the initial conditions 
of $\phi(0) = \dot{\phi}(0) = 0$ and $a(0) = 0.01$.  The same 
are shown for the broad potential, Fig.~\ref{fig:oscil}(e), in 
Figs.~\ref{fig:oscil}(f)--(h) for the initial conditions of 
$\phi(0) = \dot{\phi}(0) = 0$ and $a(0) = 10^{-11}$.  In either 
cases, the leaf hemisphere area increases monotonically while the 
screen entanglement entropy experiences drops as the field oscillates 
around the minimum.  The fractional drops from the first, second, 
and third peaks are $\simeq 1.3\%$, $0.9\%$, and $0.6\%$, respectively, 
for the steep potential and $\simeq 2.5\%$, $1.6\%$, and $1.2\%$, 
respectively, for the broad potential.

We thus find that screen entanglement entropies may decrease in a 
transition period.  The interpretation of this result, however, needs 
care.  Since the system is far from being in a ``vacuum'' during 
a transition, true entanglement entropies for subregions in the 
holographic theory may have contributions beyond that captured 
by the simple formula of Eq.~(\ref{eq:S_gamma}).  This would 
require corrections of the formula, possibly along the lines of 
Refs.~\cite{Faulkner:2013ana,Engelhardt:2014gca,Bousso:2015eda}, 
and with such corrections the drop of the entanglement entropy we 
have found here might disappear.  We leave a detailed study of 
this issue to future work.

\section{Interpretation: Beyond AdS/CFT}
\label{sec:beyond}

The entanglement entropies in the holographic theory of FRW universes 
seen so far show features different from those in CFTs of the AdS/CFT 
correspondence.  Here we highlight these differences and see how 
properties characteristic to local CFTs are reproduced when bulk 
spacetime becomes asymptotically AdS.  We also discuss implications 
of our findings for the structure of the holographic theory.  In 
particular, we discuss the structure of the Hilbert space for quantum 
gravity applicable to general spacetimes.  While we cannot determine 
the structure uniquely, we can classify possibilities under certain 
general assumptions.  The issues discussed include bulk reconstruction, 
the interior and exterior regions of the leaf, and time evolution in 
the holographic theory.

\subsection{Volume/area law for screen entanglement entropies}
\label{subsec:nonlocal}

One can immediately see that holographic entanglement entropies for
FRW universes have two features that are distinct from those in 
AdS/CFT.  First, unlike entanglement entropies in CFTs, the holographic 
entanglement entropies for FRW universes are finite for a finite 
value of ${\cal A}_*$.  Second, as seen in Section~\ref{subsec:single}, 
e.g.\ Eq.~(\ref{eq:S-gamma}), these entropies obey a volume law, 
rather than an area law.%
\footnote{A similar property was argued for holographic entropies
 for Euclidean flat spacetime in Ref.~\cite{Li:2010dr}.}
(Note that ${\cal A}_*$ is a volume from the viewpoint of the holographic 
theory.)  In particular, in the limit that the region $\Gamma$ in 
the holographic theory becomes small, the entanglement entropy $S_\Gamma$ 
becomes proportional to the volume $V_\Gamma$ {\it with a universal 
coefficient}, which we identified as $1/4$ to match the conventional 
results in Refs.~\cite{Bekenstein:1972tm,Bekenstein:1973ur,%
Bardeen:1973gs,Hawking:1974rv,Hawking:1974sw,Gibbons:1977mu}. 
(For a small enough subsystem, we expect that the entanglement 
entropy agrees with the thermal entropy.)  From the bulk point 
of view, this is because the extremal surface $E_\Gamma$ approaches 
$\Gamma$ itself, so that $\norm{E_\Gamma} \rightarrow V_\Gamma$.

What do these features mean for the holographic theory?  The finiteness 
of the entanglement entropies implies that the cutoff length of 
the holographic theory is finite, i.e.\ the number of degrees of 
freedom in the holographic theory is finite, at least effectively. 
In particular, our identification implies that the holographic 
theory effectively has a qubit degree of freedom per volume of 
$4 \ln 2$ (in Planck units), although it does not mean that the 
cutoff length of the theory is necessarily $\simeq \sqrt{4 \ln 2}$. 
It is possible that the cutoff length is $l_{\rm c} > \sqrt{4 \ln 2}$ 
and that each cutoff size cell has $N = l_{\rm c}^2/4 \ln 2$ ($> 1$) 
degrees of freedom.  In fact, since the string length $l_{\rm s}$ 
and the Planck length are related as $l_{\rm s}^2 \sim n$, where $n$ 
is the number of species in the low energy theory (including the moduli 
fields parameterizing the extra dimensions)~\cite{Dvali:2007hz}, it 
seems natural to identify $l_{\rm c}$ and $N$ as $l_{\rm s}$ and $n$, 
respectively.

The volume law of the entangled entropies implies that a holographic 
state corresponding to an FRW universe is not a ground state 
of local field theory, which is known to satisfy an area 
law~\cite{Bombelli:1986rw,Srednicki:1993im}.  This does not 
necessarily mean that the holographic theory for FRW universes 
must be nonlocal at lengthscales larger than the cutoff $l_{\rm c}$; 
it might simply be that the relevant states are highly excited 
ones.  In fact, the dynamics of the holographic theory is 
expected to respect some aspects of locality as suggested 
by the fact that the area theorem applies locally on a holographic 
screen~\cite{Sanches:2016pga}.  Of course, it is also possible 
that the holographic states for FRW universes are states of some 
special class of nonlocal theories.

The features of screen entangled entropies described here are not 
specific to FRW universes but appear in more general ``cosmological''
spacetimes, spacetimes in which the holographic screen is at finite
distances and the gravitational dynamics is not frozen there.  If 
the interior region of the holographic screen is (asymptotically) 
AdS, these features change.  In this case, the same procedure as in 
Section~\ref{sec:framework} puts the holographic screen at spatial 
infinity (the AdS boundary), and the AdS geometry makes the area 
of the extremal surface anchored to the boundary $\partial \Gamma$ 
of a small region $\Gamma$ on a leaf proportional to the area of 
$\partial \Gamma$ with a diverging coefficient:\ $\norm{E_\Gamma} 
\sim \norm{\partial \Gamma}/\epsilon$ ($\epsilon \rightarrow 0$). 
This makes the screen entanglement entropies obey an area law, so 
that the holographic theory can now be a ground state of a local 
field theory.  In fact, the theory is a CFT~\cite{Maldacena:1997re,%
Gubser:1998bc,Witten:1998qj}, consistent with the fact that we
could take the cutoff length to zero, $l_{\rm c} \sim \epsilon 
\rightarrow 0$.

\subsection{The structure of holographic Hilbert space}
\label{subsec:structure}

We now discuss implications of our analysis for the structure of the 
Hilbert space of quantum gravity for general spacetimes.  We work in 
the framework of Section~\ref{sec:framework}; in particular, we assume 
that when a holographic state represents a semiclassical spacetime, 
the area of the extremal surface contained in $D_\sigma$ and anchored 
to the boundary of a region $\Gamma$ on a leaf represents the 
entanglement entropy of the region $\Gamma$ in the holographic 
theory, Eq.~(\ref{eq:S_Sigma}).  Note that this does not necessarily
mean that the converse is true; there may be a holographic state 
in which entanglement entropies for subregions do not correspond
to the areas of extremal surfaces in a semiclassical spacetime.

Consider a holographic state representing an FRW spacetime.  The fact 
that for a small enough region $\Gamma$ the area of the extremal surface 
anchored to its boundary approaches the volume of the region on the 
leaf, $\norm{E_\Gamma} \rightarrow V_\Gamma$, implies that the degrees 
of freedom in the holographic theory are localized and that their 
density is, at least effectively, one qubit per $4 \ln 2$ (although 
the cutoff length of the theory may be larger than $\sqrt{4 \ln 2}$). 
We take these for granted as anticipated in the original holographic 
picture~\cite{'tHooft:1993gx,Susskind:1994vu}.  This suggests that 
the number of holographic degrees of freedom which comprise FRW states 
on the leaf $\sigma_*$ with area ${\cal A}_*$ is ${\cal A}_*/4$ 
{\it for any value of $w$}.

Given these assumptions, there are still a few possibilities for the 
structure of the Hilbert space of the holographic theory.  Below we 
enumerate these possibilities and discuss their salient features.

\subsubsection{Direct sum structure}

Let us first assume that state vectors representing FRW universes 
with different $w$'s are independent of each other, as indicated 
in the left portion of Fig.~\ref{fig:possib}.  This implies that 
the Hilbert space ${\cal H}_* \in \{ {\cal H}_B \}$, which contains 
holographic states for FRW universes at times when the leaf area 
is ${\cal A}_*$, has a direct sum structure
\begin{equation}
  {\cal H}_* = \bigoplus_w {\cal H}_{*,w}.
\label{eq:direct-sum}
\end{equation}
Here, we regard universes with the equation of state parameters falling 
in a range $\delta w \ll 1$ to be macroscopically identical, where 
$\delta w$ is a small number that does not scale with ${\cal A}_*$.%
\footnote{If we consider FRW universes with multiple fluid components, 
 the corresponding spaces must be added in the right-hand side of 
 Eq.~(\ref{eq:direct-sum}).}
This is the structure envisioned originally in Ref.~\cite{Nomura:2011rb}. 
\begin{figure}[t]
\begin{center}
  \includegraphics[height=6cm]{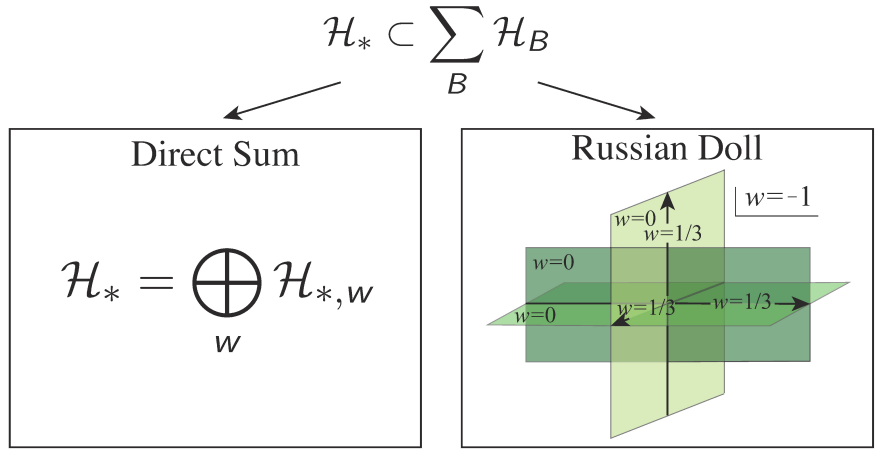}
\end{center}
\caption{Possible structures of the Hilbert space ${\cal H}_*$ 
 for a fixed boundary space $B$.  In the direct sum structure 
 (left), each semiclassical spacetime in $D_{\sigma_*}$ has its 
 own Hilbert space ${\cal H}_{*,w}$.  The Russian doll structure 
 (right) corresponds to the scenario of ``spacetime equals 
 entanglement,'' i.e.\ the entanglement entropies of the holographic 
 degrees of freedom determine spacetime in $D_{\sigma_*}$.  This 
 implies that a superposition of exponentially many semiclassical 
 spacetimes can lead to a different semiclassical spacetime.}
\label{fig:possib}
\end{figure}

What is the structure of ${\cal H}_{*,w}$?  A natural possibility 
is that each of these subspaces has dimension
\begin{equation}
  \ln {\rm dim}\, {\cal H}_{*,w} = \frac{{\cal A}_*}{4}.
\label{eq:H_w-dim}
\end{equation}
This is motivated by the fact that arbitrary unitary transformations 
acting in each cutoff size cell do not change the structure of screen 
entanglement entropies, and they can lead to $e^{{\cal A}_*/4}$ 
independent holographic states that have the screen entanglement 
entropies corresponding to the FRW universe with the equation of 
state parameter $w$.  If we regard all of these states as microstates 
for the FRW universe with $w$, then we obtain Eq.~(\ref{eq:H_w-dim}). 
This, however, does not mean that the holographic states representing 
the FRW universe with $w$ comprise the Hilbert space ${\cal H}_{*,w}$. 
Since these states form a basis of ${\cal H}_{*,w}$, their superposition 
can lead to a state which has entanglement entropies far from those 
corresponding to the FRW universe with $w$.  In fact, we can even form 
a state in which degrees of freedom in different cells are not entangled 
at all.  This is a manifestation of the fact that entanglement cannot 
be represented by a linear operator.

This implies that states representing the semiclassical FRW universe 
are ``preferred basis states'' in ${\cal H}_{*,w}$, and their 
arbitrary linear combinations may lead to states that do not admit 
a semiclassical interpretation.  We expect that these preferred axes 
are ``fat'':\ we have to superpose a large number of basis states, 
in fact exponentially many in ${\cal A}_*$, to obtain a state that 
is not semiclassical (because we need that many states to appreciably 
change the entanglement structure, as illustrated in a toy qubit 
model in Appendix~\ref{app:qubit}).  It is, however, true that 
most of the states in ${\cal H}_{*,w}$, including those having 
the entanglement entropy structure corresponding to a universe 
with another $w$, are states that do not admit a semiclassical 
spacetime interpretation.  Drawing an analogy with the work in 
Refs.~\cite{Almheiri:2012rt,Almheiri:2013hfa,Marolf:2013dba}, we 
may call them ``firewall'' states.  In Section~\ref{subsec:time-evo}, 
we argue that these states are unlikely to be produced by standard 
semiclassical time evolution.

The dimension of ${\cal H}_*$ is given by
\begin{equation}
  \ln {\rm dim}\, {\cal H}_* 
  = \ln \sum_w e^{\frac{{\cal A}_*}{4}} 
  \approx \frac{{\cal A}_*}{4} - \ln \delta w 
  \simeq \frac{{\cal A}_*}{4},
\label{eq:H_*-dim}
\end{equation}
as expected from the covariant entropy bound (unless $\delta w$ is 
exponentially small in ${\cal A}_*$, which we assume not to be the 
case).  Small excitations over the FRW universes may be represented 
in suitably extended spaces ${\cal H}_{*,w}$.  Since entropies 
associated with the excitations are typically subdominant in 
${\cal A}_*$~\cite{'tHooft:1993gx,Nomura:2013lia}, they have only 
minor effects on the overall picture, e.g.\ Eq.~(\ref{eq:H_*-dim}). 
(Note that the excitations here do not contain the degrees of freedom 
attributed to gravitational, e.g.\ Gibbons-Hawking, radiation.  These 
degrees of freedom are identified as the microscopic degrees of freedom 
of spacetimes, i.e.\ the vacuum degrees of freedom~\cite{Nomura:2014woa,%
Nomura:2014voa,Nomura:2016qum}, which are already included in 
Eq.~(\ref{eq:H_w-dim}).)  The operators representing the excitations 
can be standard linear operators acting on the Hilbert space 
${\cal H}_*$, at least in principle.

We also mention the possibility that the logarithm of the number 
of independent states $N_w$ representing the FRW universe with $w$ 
is smaller than ${\cal A}_*/4$.  For example, it might be given 
approximately by twice the entanglement entropy for a leaf hemisphere 
$S_w(\pi/2) = Q_w(\pi/2) {\cal A}_*/8$:
\begin{equation}
  \ln N_w \approx Q_w\Bigl(\frac{\pi}{2}\Bigr)\, \frac{{\cal A}_*}{4}.
\label{eq:H_w-dim_Q}
\end{equation}
The basic picture in this case is not much different from that 
discussed above; for example, the difference of the values of 
$\ln {\rm dim}\, {\cal H}_*$ is higher order in $1/{\cal A}_*$ 
(although this possibility makes the issue of the equivalence 
condition for the boundary space label $B$ nontrivial).  We will 
not consider this case separately below.

\subsubsection{Russian doll structure: spacetime equals entanglement}

In the picture described above, the structures of ${\cal H}_{*,w}$'s 
are all very similar.  Each of these spaces has the dimension of 
${\cal A}_*/4$ and has $e^{{\cal A}_*/4}$ independent states that 
represent the FRW universe with a fixed value of $w$.  An arbitrary 
linear combination of these states, however, is not necessarily a 
state representing the FRW universe with $w$.  In the previous picture, 
we identified all such states as the firewall (or unphysical) states, 
but is it possible that some of these states, in fact, represent 
other FRW universes?  In particular, is it possible that all the 
${\cal H}_{*,w}$ spaces are actually the {\it same} space, i.e.\ 
${\cal H}_{*,w_1} = {\cal H}_{*,w_2}$ for all $w_1 \neq w_2$?

A motivation to consider this possibility comes from the fact that 
if $w$ does not by itself provide an independent label for states, 
then the $e^{{\cal A}_*/4}$ independent microstates for the FRW 
universe with a fixed $w$ can form a basis for the configuration 
space of the ${\cal A}_*/4$ holographic degrees of freedom.  This 
implies that we can superpose these states to obtain many---in fact 
$e^{{\cal A}_*/4}$---independent states that have the entanglement 
entropies corresponding to the FRW universe with any $w' \neq w$, which 
we can identify as the states representing the FRW universe with $w'$.%
\footnote{The same argument applies to the FRW universes with multiple 
 fluid components, so that the states representing these universes also 
 live in the same Hilbert space as the single component universes.}
In essence, this amounts to saying that the converse of the statement 
made at the beginning of this subsection is true:\ when a holographic 
state has the form of entanglement entropies corresponding to a 
certain spacetime, then the state indeed represents that spacetime. 
This scenario was proposed in Ref.~\cite{Nomura:2016aww} and called 
``spacetime equals entanglement.''  It is depicted in the right 
portion of Fig.~\ref{fig:possib}.

One might think that the scenario does not make sense, since 
it implies that a superposition of classical universes can 
lead to a different classical universe.  Wouldn't it make any 
reasonable many worlds interpretation of spacetime impossible? 
In Ref.~\cite{Nomura:2016aww}, it was argued that this is not 
the case.  First, for a given FRW universe, we expect that the 
space of its microstates is ``fat''; namely, a superposition of 
less than $e^{O(\delta w {\cal A}_*)}$ microstates representing 
a classical universe leads only to another microstate representing 
the same universe.  This implies that the $e^{{\cal A}_*/4}$ 
microstates of a classical universe generate an ``effective vector 
space,'' unless we consider a superposition of an exponentially 
large, $\gtrsim e^{O(\delta w {\cal A}_*)}$, number of states.

What about a superposition of different classical universes?  In 
particular, if states representing universes with $w_1$ and $w_2$ 
($\neq w_1$) are superposed, then how does the theory know that 
the resulting state represents a superposition of two classical 
universes, and not another---perhaps even non-classical---universe? 
A key point is that the Hilbert space we consider has a special 
basis, determined by the ${\cal A}_*/4$ local degrees of freedom 
in the holographic space:%
\footnote{For simplicity, here we have assumed that the degrees 
 of freedom are qubits, but the subsequent argument persists as 
 long as the number of independent states for each degree of 
 freedom does not scale with ${\cal A}_*$.  In particular, it 
 persists if the correct structure of ${\cal H}_*$ appears as 
 $({\mathbf C}^N)^{\otimes {\cal A}_*/l_{\rm c}^2}$ as discussed 
 in Section~\ref{subsec:nonlocal}.}
\begin{equation}
  {\cal H}_* = ({\mathbf C}^2)^{\otimes \frac{{\cal A}_*}{4}}.
\label{eq:H*-struc}
\end{equation}
From the result in Section~\ref{subsec:single}, we know that a state 
representing the FRW universe with $w_1$ is more entangled than that 
representing the FRW universe with $w_2$ ($> w_1$).  This implies 
that when expanded in the natural basis $\{ \ket{\Psi_i} \}$ for 
the structure of Eq.~(\ref{eq:H*-struc}), i.e.\ the product state 
basis for the ${\cal A}_*/4$ local holographic degrees of freedom, 
then a state $\ket{\Psi_{w_1}}$ representing the universe with 
$w_1$ effectively has exponentially more terms than a state 
$\ket{\Psi_{w_2}}$ representing the universe with $w_2$.  Namely, 
we expect that
\begin{equation}
  \ket{\Psi_w} \approx 
    \sum_{i=1}^{e^{f(w) \frac{{\cal A}_*}{4}}} a_i\, \ket{\Psi_i},
\label{eq:psi_w}
\end{equation}
where $f(w)$ is a monotonically decreasing function of $w$ taking 
values of $O(1)$, and $a_i$ are coefficients taking generic random 
values.  The normalization condition for $\ket{\Psi_w}$ then implies
\begin{equation}
  |a_i| \approx O\bigl( e^{-f(w) \frac{{\cal A}_*}{8}} \bigr),
\label{eq:a_i}
\end{equation}
i.e.\ the size of the coefficients in product basis expansion is 
exponentially different for states with different $w$'s.  This, 
in particular, leads to
\begin{equation}
  \inner{\Psi_{w_1}}{\Psi_{w_2}} 
  \lesssim O\bigl( e^{-\{ f(w_1)-f(w_2) \} \frac{{\cal A}_*}{8}} \bigr),
\label{eq:ortho}
\end{equation}
i.e.\ microstates for different universes are orthogonal up to 
exponentially suppressed corrections.

Now consider a superposition state
\begin{equation}
  \ket{\Psi} = c_1 \ket{\Psi_{w_1}} + c_2 \ket{\Psi_{w_2}},
\label{eq:superpose}
\end{equation}
where $|c_1|^2 + |c_2|^2 = 1$ up to the correction from exponentially 
small overlap $\inner{\Psi_{w_1}}{\Psi_{w_2}}$.  We are interested 
in the reduced density matrix for a subregion $\Gamma$ in the 
holographic theory
\begin{equation}
  \rho_\Gamma = {\rm Tr}_{\bar{\Gamma}}\, \ket{\Psi} \bra{\Psi},
\label{eq:rho_Gamma}
\end{equation}
where $\Gamma$ occupies less than a half of the leaf volume.  The 
property of Eq.~(\ref{eq:ortho}) then ensures that
\begin{equation}
  \rho_\Gamma = |c_1|^2 \rho^{(1)}_\Gamma + |c_2|^2 \rho^{(2)}_\Gamma,
\label{eq:rho-factor}
\end{equation}
up to corrections exponentially suppressed in ${\cal A}_*$.  Here, 
$\rho^{(1)}_\Gamma$ ($\rho^{(2)}_\Gamma$) are the reduced density 
matrices we would obtain if the state were genuinely $\ket{\Psi_{w_1}}$ 
($\ket{\Psi_{w_2}}$).  The matrix $\rho_\Gamma$ thus takes the form 
of an incoherent classical mixture for the two universes.  Similarly, 
the entanglement entropy for the region $\Gamma$ is also incoherently 
added
\begin{equation}
  S_\Gamma = |c_1|^2 S^{(1)}_\Gamma + |c_2|^2 S^{(2)}_\Gamma 
    + S_{\Gamma,{\rm mix}},
\label{eq:EE-factor}
\end{equation}
where $S^{(1,2)}_\Gamma$ are the entanglement entropies obtained 
if the state were $\ket{\Psi_{w_{1,2}}}$, and
\begin{equation}
  S_{\Gamma,{\rm mix}} = - |c_1|^2 \ln |c_1|^2 - |c_2|^2 \ln |c_2|^2,
\label{eq:S_mix}
\end{equation}
is the entropy of mixing (classical Shannon entropy), suppressed by 
factors of $O({\cal A}_*)$ compared with $S^{(1,2)}_\Gamma$.  The 
features in Eqs.~(\ref{eq:rho-factor},~\ref{eq:EE-factor}) indicate 
that unless $|c_1|$ or $|c_2|$ is suppressed exponentially in 
${\cal A}_*$, the state $\ket{\Psi}$ admits the usual interpretation 
of a superposition of macroscopically different universes with 
$w_{1,2}$.

In fact, unless a superposition involves exponentially many microstates, 
we find
\begin{equation}
  \ket{\Psi} = \sum_i c_i \ket{\Psi_{w_i}} 
\quad\Rightarrow\quad
  \begin{array}{l}
    \rho_\Gamma  = \sum_i |c_i|^2 \rho^{(i)}_\Gamma,\\
    S_\Gamma = \sum_i |c_i|^2 S^{(i)}_\Gamma + S_{\Gamma,{\rm mix}},
  \end{array}
\label{eq:rho_n-gen}
\end{equation}
with exponential accuracy.  Here, $S_{\Gamma,{\rm mix}} = - \sum_i 
|c_i|^2 \ln |c_i|^2$ and is suppressed by a factor of $O({\cal A}_*)$ 
compared with the first term in $S_\Gamma$.  This indicates that the 
standard many worlds interpretation applies to classical spacetimes 
under any reasonable measurements (only) in the limit that 
$e^{-{\cal A}_*}$ is regarded as zero, i.e.\ unless a superposition 
involves exponentially many terms or an exponentially small coefficient. 
This is consonant with the observation that classical spacetime 
has an intrinsically thermodynamic nature~\cite{Jacobson:1995ab}, 
supporting the idea that it consists of a large number of degrees 
of freedom.  In Ref.~\cite{Nomura:2016aww}, the features described 
above were discussed using a qubit model in which the states 
representing the FRW universes exhibit a ``Russian doll'' structure 
as illustrated in Fig~\ref{fig:possib}.  We summarize this model 
in Appendix~\ref{app:qubit} for completeness.

We conclude that the states representing FRW universes with a 
leaf area ${\cal A}_*$ can all be elements of a single Hilbert 
space ${\cal H}_*$ with dimension
\begin{equation}
  \ln {\rm dim}\, {\cal H}_* = \frac{{\cal A}_*}{4}.
\label{eq:H_*-SEE}
\end{equation}
Any such universe has $e^{{\cal A}_*/4}$ independent microstates, 
which form a basis of ${\cal H}_*$.  This implies that matter and 
spacetime must have a sort of unified origin in this picture, since 
a superposition that changes the spacetime geometry must also change 
the matter content filling the universe.  How could this be the case?

Consider, as discussed in Section~\ref{subsec:nonlocal}, that the 
cutoff length of the holographic theory is of order $l_{\rm s} \sim 
\sqrt{n}$, where $n$ ($> 1$) is the number of species at energies 
below $1/l_{\rm s}$.  This implies that the ${\cal A}_*/4$ degrees 
of freedom can be decomposed as
\begin{equation}
  \frac{{\cal A}_*}{4} \sim n \frac{{\cal A}_*}{l_{\rm s}^2},
\label{eq:decomp}
\end{equation}
representing $n$ fields living in the holographic space of cutoff 
length $l_{\rm s}$.  Now, to obtain the $e^{{\cal A}_*/4}$ microstates 
for an FRW universe we need to consider rotations for all the $n$ 
degrees of freedom in each cutoff size cell.  This may suggest that 
the identity of a matter species at the fundamental level may not be 
as adamant as in low energy semiclassical field theories.  The reason 
why all the $n$ degrees of freedom can be involved could be because 
the ``local effective temperature,'' defined analogously to de~Sitter 
space, diverges at the holographic screen.

Finally, we expect that small excitations over FRW universes in 
this picture are represented by non-linear/state-dependent operators 
in the (suitably extended) ${\cal H}_*$ space, along the lines 
of Ref.~\cite{Papadodimas:2015jra} (see Refs.~\cite{Papadodimas:2012aq,%
Verlinde:2012cy,Nomura:2012ex} for earlier work).  This is because 
a superposition of background spacetimes may lead to another background 
spacetime, so that operators representing excitations should know 
the entire quantum state they are acting on.

\subsection{Bulk reconstruction from holographic states}
\label{subsec:reconst}

We have seen that the entanglement entropies of the ${\cal A}_*/4$ 
local holographic degrees of freedom in the holographic space $\sigma_*$ 
encode information about spacetime in the causal region $D_{\sigma_*}$. 
Here we discuss in more detail how this encoding may work in general.

While we have focused on the case in which the future-directed ingoing 
light rays emanating orthogonally from $\sigma_*$ (i.e.\ in the 
$k^a$ directions in Fig.~\ref{fig:def}) meet at a point $p_0$, our 
discussion can be naturally extended to the case in which the light 
rays encounter a spacetime singularity before reaching a caustic. 
\begin{figure}[]
  \includegraphics[height=6.5cm]{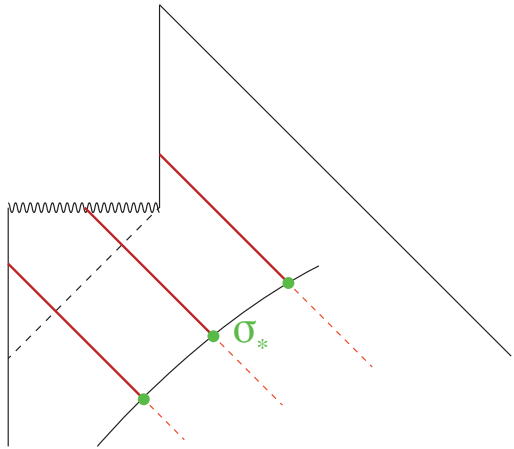}
\caption{If a black hole forms inside the holographic screen, 
 future-directed ingoing light rays emanating orthogonally from 
 the leaf $\sigma_*$ at an intermediate time may hit the singularity 
 before reaching a caustic.  While the diagram here assumes spherical 
 symmetry for simplicity, the phenomenon can occur more generally.}
\label{fig:BH}
\end{figure}
This may occur, for example, if a black hole forms in a universe 
as depicted in Fig.~\ref{fig:BH}, where we have assumed spherical 
symmetry for collapsing matter and taken $p(\tau)$ to follow its 
center.  We see that at intermediate times, the future-directed 
ingoing light rays emanating from leaves encounter the black hole 
singularity before reaching a caustic.%
\footnote{At these times, the specific construction of the 
 holographic screen in Section~\ref{sec:framework} cannot be 
 applied exactly.  This is not a problem as the fundamental 
 object is the state in the holographic space, and not $p(\tau)$. 
 The purpose of the discussion in Section~\ref{sec:framework} 
 is to illustrate our observer centric choice of fixing the 
 holographic redundancy in formulating the holographic theory.}
Our interpretation in this case is similar to the case without 
a singularity.  The entanglement entropies of the holographic 
degrees of freedom encode information about $D_{\sigma_*}$.

In what sense does a holographic state on $\sigma_*$ contain 
information about $D_{\sigma_*}$?  We assume that the theory 
allows for the reconstruction of $D_{\sigma_*}$ from the data 
in the state on $\sigma_*$.  On the other hand, it is not the 
case that the collection of extremal surfaces for all possible 
subregions on $\sigma_*$ probes the entire $D_{\sigma_*}$. 
This suggests that the full reconstruction of $D_{\sigma_*}$ 
may need bulk time evolution.

There is, however, no a priori guarantee that the operation 
corresponding to bulk time evolution is complete within ${\cal H}_*$. 
This means that there may be no arrangement of operators defined 
in ${\cal H}_*$ that represents certain operators in $D_{\sigma_*}$. 
For these subsets of $D_{\sigma_*}$, bulk reconstruction would involve 
operators defined on other boundary spaces.  In other words, the 
operators supported purely in ${\cal H}_*$ may allow for a direct 
spacetime interpretation only for a portion of $D_{\sigma_*}$, 
e.g.\ the outside of the black hole horizon in the example of 
Fig.~\ref{fig:BH} (in which case some of the operators would 
represent the stretched horizon degrees of freedom).  Our assumption 
merely says that the operators in ${\cal H}_*$ acting on the state 
contain data equivalent to specifying the system on a Cauchy surface 
for $D_{\sigma_*}$.

The consideration above implies that the information in a holographic 
state on $\sigma_*$, when interpreted through operators in 
${\cal H}_*$, may only be partly semiclassical.  We expect 
that this becomes important when the spacetime has a horizon. 
In particular, for the $w=-1$ FRW universe, the leaf $\sigma_*$ 
is formally beyond the stretched de~Sitter horizon as viewed 
from $p(\tau)$.  This may mean that some of the degrees of freedom 
represented by operators defined in ${\cal H}_*$ can only be 
viewed as non-semiclassical stretched horizon degrees of freedom.

\subsection{Information about the ``exterior'' region}
\label{subsec:exterior}

The information about $D_{\sigma_*}$, contained in the screen 
entanglement entropies, is not sufficient to determine future 
states obtained by time evolution.  This information corresponds 
to that on the ``interior'' light sheet, i.e.\ the light sheet 
generated by light rays emanating in the $+k^a$ directions 
from $\sigma_*$.%
\footnote{If the light sheet encounters a singularity before 
 reaching a caustic, then the information about the singularity 
 may also be contained.}
However, even barring the possibility of information sent into 
the system from a past singularity or past null infinity (which 
we will discuss in Section~\ref{sec:discuss}), determining a 
future state requires information about the ``exterior'' light 
sheet, i.e.\ the one generated by light rays emanating in the 
$-k^a$ directions; see Fig.~\ref{fig:exterior}.%
\footnote{This light sheet is terminated at a singularity or 
 a caustic.  Note that the information beyond a caustic is not 
 needed to specify the state~\cite{Nomura:2013nya}, since it 
 is timelike related with the information on the interior light 
 sheet~\cite{Wald:1984rg} so that the two do not provide independent 
 information.}
\begin{figure}[]
  \includegraphics[height=6.5cm]{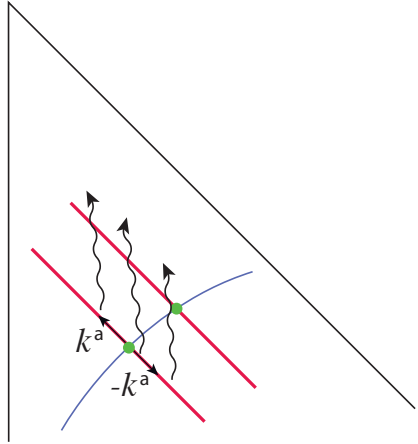}
\caption{To determine a state in the future, we need information 
 on the ``exterior'' light sheet, the light sheet generated by 
 light rays emanating from $\sigma_*$ in the $-k^a$ directions, 
 in addition to that on the ``interior'' light sheet, i.e.\ the 
 one generated by light rays emanating in the $+k^a$ directions.}
\label{fig:exterior}
\end{figure}
How is this information encoded in the holographic state?  Does 
it require additional holographic degrees of freedom beyond the 
${\cal A}_*/4$ degrees of freedom considered so far?

The simplest possibility is that the $e^{{\cal A}_*/4}$ microstates 
for each interior geometry (i.e.\ a fixed screen entanglement 
entropy structure) contain all the information associated with 
both the interior and exterior light sheets.  If this is indeed 
the case, then we do not need any other degrees of freedom in 
the holographic space $\sigma_*$ beyond the ${\cal A}_*/4$ ones 
discussed earlier.  It also implies the following properties 
for the holographic theory:
\begin{itemize}
\item {\bf Autonomous time evolution} 
--- Assuming the absence of a signal sent in from a past singularity 
or past null infinity (see Section~\ref{sec:discuss}), the evolution 
of the state is autonomous.  In particular, an initial pure state 
evolves into a pure state.
\item {\bf {\boldmath $S$}-matrix description for a dynamical black hole} 
--- As a special case, a pure state representing initial collapsing 
matter to form a black hole will evolve into a pure state representing 
final Hawking radiation, even if $p(\tau)$ hits the singularity at 
an intermediate stage (at least if the leaf stays outside the black 
hole); see Fig.~\ref{fig:BH}.
\item {\bf Strengthened covariant entropy bound} 
--- According to the original proposal of the covariant entropy 
bound~\cite{Bousso:1999xy,Bousso:2002ju}, the entropy on {\it each} 
of the interior and exterior light sheets is bounded by ${\cal A}_*/4$, 
implying that
\begin{equation}
  \ln {\rm dim}\, {\cal H}_* = 2 \times \frac{{\cal A}_*}{4} 
  = \frac{{\cal A}_*}{2},
\label{eq:ceb-1}
\end{equation}
where ${\cal H}_*$ is the Hilbert space associated with $\sigma_*$. 
The present picture instead says
\begin{equation}
  \ln {\rm dim}\, {\cal H}_* = \frac{{\cal A}_*}{4},
\label{eq:ceb-2}
\end{equation}
implying that the entropy on the {\it union} of the interior 
and exterior light sheets is bounded by ${\cal A}_*/4$:%
\footnote{This bound was anticipated earlier~\cite{Nomura:2013lia} 
 based on more phenomenological considerations.}
Note that the bound does not say that the entropy on each of the interior 
and exterior light sheets is separately bounded by ${\cal A}_*/8$, and 
so is profoundly holographic.  This bound is consistent with the fact 
that in any known realistic examples the covariant entropy bound is 
saturated only in one side of a leaf~\cite{Bousso:2010pm}.
\end{itemize}
The picture described here is, of course, a conjecture, which needs 
to be tested.  For example, if a realistic case is found in which 
the ${\cal A}_*/4$ bound is violated by the contributions from both 
the interior and exterior light sheets, then we would have to modify 
the framework, e.g., by adding an extra ${\cal A}_*/4$ degrees of 
freedom on the holographic space.  It is interesting that there is 
no known system that requires such a modification.

We finally discuss the connection with AdS/CFT.  In the limit that 
the spacetime becomes asymptotically AdS, the location of the holographic 
screen is sent to spatial infinity, so that ${\cal A}_* \rightarrow 
\infty$.  This implies that there are $N_* = e^{{\cal A}_*/4} \rightarrow 
\infty$ microstates for any spacetime configuration in $D_{\sigma_*}$ 
for a leaf $\sigma_*$, including the case that it is a portion of the 
empty AdS space.  Wouldn't this contradict the statement that the vacuum 
of a CFT is unique?

As we will discuss in Section~\ref{sec:discuss}, the degrees of freedom 
associated with $N_*$ correspond to a freedom of sending information 
into the system at a later stage of time evolution, i.e.\ that of inserting 
operators at locations other than the point $x_{-\infty}$ corresponding 
to $\tau = -\infty$ on the conformally compactified AdS boundary.  It 
is with this freedom that the CFT corresponds to the AdS limit of our 
theory including the $N_*$ degrees of freedom:
\begin{equation}
  {\rm CFT}\,\, \Longleftrightarrow\,\, 
    \lim_{{\cal M} \rightarrow {\rm asymptotic\, AdS}} {\cal T},
\label{eq:AdS-CFT-2}
\end{equation}
where ${\cal M}$ is the spacetime inside the holographic screen, and 
${\cal T}$ represents the theory under consideration.  Here, we have 
taken the holographic screen to stay outside the cutoff surface 
(corresponding to the ultraviolet cutoff of the CFT) which is also 
sent to infinity.

This implies that if we want to consider a setup in which the evolution 
of the state is ``autonomous'' within the bulk, then we need to fix a 
configuration of operators at $x \neq x_{-\infty}$, i.e.\ we need to 
fully fix a boundary condition at the AdS boundary.  The correspondence 
to our theory in this case is written as
\begin{equation}
  {\rm autonomous \,\, CFT}\,\, \Longleftrightarrow\,\, 
    \lim_{{\cal M} \rightarrow {\rm asymptotic\, AdS}} {\cal T}
    \,\scalebox{1.5}{/} N_*.
\label{eq:AdS-CFT}
\end{equation}
The conventional vacuum state of the CFT corresponds to a special 
configuration of the $N_*$ degrees of freedom that does not send in 
any signal to the system at later times (the simple reflective boundary 
conditions at the AdS boundary).  Given the correspondence between 
the $N_*$ degrees of freedom and boundary operators, we expect that 
this configuration is unique.  The state corresponding to the CFT 
vacuum in our theory is then unique:\ the vacuum state of the theory 
${\cal T}/N_*$ with the configuration of the $N_*$ degrees of freedom 
chosen uniquely as discussed above.

\subsection{Time evolution}
\label{subsec:time-evo}

Another feature of the holographic theory of general spacetimes 
beyond AdS/CFT is that the boundary space changes in time.  This 
implies that we need to consider the theory in a large Hilbert space 
containing states living in different boundary spaces, Eq.~(\ref{eq:H}). 
For states representing FRW universes, the relevant space can be 
written as
\begin{equation}
  {\cal H} = \sum_{\cal A} {\cal H}_{\cal A},
\label{eq:H-FRW}
\end{equation}
where ${\cal A}$ is the area of the leaf, and the sum of the Hilbert 
spaces is defined by Eq.~(\ref{eq:H_sum}).%
\footnote{More precisely, ${\cal H}_{\cal A}$ contains states whose 
 leaf areas fall in the range between ${\cal A}$ and ${\cal A} 
 + \delta{\cal A}$.  The precise choice of $\delta{\cal A}$ is 
 unimportant unless it is exponentially small in ${\cal A}$. 
 For example, the dimension of ${\cal H}_A$ is $e^{{\cal A}/4} 
 \delta{\cal A}$, so that the entropy associated with it is 
 ${\cal A}/4 + \ln \delta{\cal A}$, which is ${\cal A}/4$ at 
 the leading order in $1/{\cal A}$ expansion.}
While the microscopic theory involving time evolution is not yet 
available, we can derive its salient features by assuming that it 
reproduces the semiclassical time evolution in appropriate regimes. 
Here we discuss this issue for both direct sum and Russian doll 
structures.  In particular, we consider a semiclassical time evolution 
in which a state having the leaf area ${\cal A}_1$ evolves into 
that having the leaf area ${\cal A}_2$ ($> {\cal A}_1$).

\subsubsection*{Direct sum structure}

In this case there is a priori no need to introduce non-linearity 
in the algebra of observables, so we may assume that time evolution 
is described by a standard unitary operator acting on ${\cal H}$. 
In particular, time evolution of a state in ${\cal H}_{{\cal A}_1}$ 
into that in ${\cal H}_{{\cal A}_2}$ is given by a linear map from 
elements of ${\cal H}_{{\cal A}_1}$ to those in ${\cal H}_{{\cal A}_2}$. 

Consider microstates $\ket{\Psi^w_i}$ ($i = 1,\cdots,e^{{\cal A}_1/4}$) 
representing the FRW universe with $w$ when the leaf area is 
${\cal A}_1$, $\ket{\Psi^w_i} \in {\cal H}_{{\cal A}_1,w} \subset 
{\cal H}_{{\cal A}_1}$; see Eq.~(\ref{eq:direct-sum}).  Assuming that 
all these states follow the standard semiclassical time evolution,%
\footnote{Here we ignore the possibility that the equation of state 
 changes between the two times, e.g., by a conversion of the matter 
 content or vacuum decay.  This does not affect our discussion below.}
their evolution is given by
\begin{equation}
  \ket{\Psi^w_i} \rightarrow \ket{\Phi^w_i},
\label{eq:evol-DS}
\end{equation}
where $\{ \ket{\Phi^w_i} \}$ is a subset of the microstates 
$\ket{\Phi^w_j}$ ($j = 1,\cdots,e^{{\cal A}_2/4}$) representing the 
FRW universe with $w$ when the leaf area is ${\cal A}_2$, $\ket{\Phi^w_j} 
\in {\cal H}_{{\cal A}_2,w} \subset {\cal H}_{{\cal A}_2}$.  This has 
an important implication.  Suppose that the initial state of the universe 
is given by
\begin{equation}
  \ket{\Psi} = \sum_i a_i \ket{\Psi^w_i}.
\label{eq:Psi_init-DS}
\end{equation}
As we discussed before, if the effective number of terms in the sum 
is of order $e^{{\cal A}_1/4}$, namely if there are $e^{{\cal A}_1/4}$ 
nonzero $a_i$'s with size $|a_i| \sim e^{-{\cal A}_1/8}$, then the 
state $\ket{\Psi}$ is not semiclassical, i.e.\ a firewall state 
(because a superposition of that many microstates changes the structure 
of the entanglement entropies).  After the time evolution, however, 
this state becomes
\begin{equation}
  \ket{\Psi} \rightarrow \ket{\Phi} = \sum_i a_i \ket{\Phi^w_i},
\label{eq:Psi_fin-DS}
\end{equation}
where the number of terms in the sum is $e^{{\cal A}_1/4}$ 
because of the linearity of the map.  This implies that the state 
$\ket{\Phi}$ is {\it not} a firewall state, since the number of 
terms is much (exponentially) smaller than the dimensionality of 
${\cal H}_{{\cal A}_2,w}$:\ $e^{{\cal A}_1/4} \ll e^{{\cal A}_2/4}$. 
In particular, the state $\ket{\Phi}$ represents the standard 
semiclassical FRW universe with the equation of state 
parameter $w$.

This shows that this picture has a ``built-in'' mechanism of 
eliminating firewalls through time evolution, at least when the 
leaf area increases in time as we focus on here.  This process 
happens very quickly---any macroscopic increase of the leaf area 
makes the state semiclassical regardless of the initial state.

\subsubsection*{Spacetime equals entanglement}

In this case, time evolution from states in ${\cal H}_{{\cal A}_1}$ 
to those in ${\cal H}_{{\cal A}_2}$ is expected to be non-linear. 
Consider microstates $\ket{\Psi^w_i}$ ($i = 1,\cdots,e^{{\cal A}_1/4}$) 
representing the FRW universe with $w$ when the leaf area is 
${\cal A}_1$, $\ket{\Psi^w_i} \in {\cal H}_{{\cal A}_1}$.  As 
before, requiring the standard semiclassical evolution for all 
the microstates, we obtain
\begin{equation}
  \ket{\Psi^w_i} \rightarrow \ket{\Phi^w_i},
\label{eq:evol-SEE}
\end{equation}
where $\{ \ket{\Phi^w_i} \}$ is a subset of the microstates 
$\ket{\Phi^w_j}$ ($j = 1,\cdots,e^{{\cal A}_2/4}$) representing 
the FRW universe with $w$ when the leaf area is ${\cal A}_2$, 
$\ket{\Phi^w_j} \in {\cal H}_{{\cal A}_2}$.  Suppose the 
initial state
\begin{equation}
  \ket{\Psi} = \sum_i a_i \ket{\Psi^w_i} \equiv \ket{\Psi^{w'}},
\label{eq:Psi_init-SEE}
\end{equation}
represents the FRW universe with $w' < w$.  This is possible if the 
effective number of terms in the sum is of order $e^{{\cal A}_1/4}$, 
i.e.\ if there are $e^{{\cal A}_1/4}$ nonzero $a_i$'s with size 
$|a_i| \sim e^{-{\cal A}_1/8}$.  Now, if the time evolution map were 
linear, then this state would evolve into
\begin{equation}
  \ket{\Psi^{w'}} \rightarrow \ket{\Phi} = \sum_i a_i \ket{\Phi^w_i}.
\label{eq:Psi_fin-SEE}
\end{equation}
This state, however, is not a state representing the FRW universe with 
$w'$, since the effective number of terms in the sum, $e^{{\cal A}_1/4}$, 
is exponentially smaller than $e^{{\cal A}_2/4}$, the required number 
to obtain a state with $w'$ from the microstates $\ket{\Phi^w_i}$. 
To avoid this problem, the map from ${\cal H}_{{\cal A}_1}$ into 
${\cal H}_{{\cal A}_2}$ must be non-linear so that $\ket{\Psi^{w'}}$ 
evolves into $\ket{\Phi^{w'}}$ containing $e^{{\cal A}_2/4}$ terms 
when expanded in terms of $\ket{\Phi^w_i}$.

Here we make two comments.  First, the non-linearity of the map 
described above does not necessarily mean that the time evolution 
of semiclassical degrees of freedom (given as excitations on the 
background states considered here) is non-linear, since the definition 
of these degrees of freedom would also be non-linear at the fundamental 
level.  In fact, from observation this evolution must be linear, 
at least with high accuracy.  This requirement gives a strong 
constraint on the structure of the theory.  Second, the non-linearity 
seen above arises when the area of the boundary space changes, 
${\cal A}_1 \rightarrow {\cal A}_2 \neq {\cal A}_1$.  Since the 
area of the boundary is fixed in the AdS/CFT limit (with the standard 
regularization and renormalization procedure), this non-linearity 
does not show up in the CFT time evolution, generated by the 
dilatation operator with respect to the $t = -\infty$ point in 
the compactified Euclidean space.%
\footnote{This does not mean that the interior of a black hole is 
 described by state-independent operators in the CFT.  It is possible 
 that the CFT does not provide a description of the black hole 
 interior; see discussion in Section~\ref{subsec:reconst}.}

We finally discuss relations between different ${\cal H}_B$'s. 
While we do not know how they are related, for example they could 
simply exist as a direct sum in the full Hilbert space ${\cal H} 
= \bigoplus_B {\cal H}_B$, an interesting possibility is that 
their structure is analogous to the Russian doll structure within 
a single ${\cal H}_B$.  Specifically, let us introduce the notation 
to represent the Russian doll structure as
\begin{equation}
  \{ \ket{\Psi^w} \} \,\prec\, \{ \ket{\Psi^{w'}} \}
\quad\mbox{for}\quad
  w' < w,
\label{eq:new-notation}
\end{equation}
meaning that the left-hand side is a measure zero subset of the closure 
of the right-hand side.  
We may imagine that states $\ket{\Psi_B}$ representing spacetimes 
with boundary $B$ and states $\ket{\Psi_{B'}}$ representing those 
with boundary $B'$ are related similarly as
\begin{equation}
  \{ \ket{\Psi_B} \} \,\prec\, \{ \ket{\Psi_{B'}} \}
\quad\mbox{for}\quad
  \norm{B} < \norm{B'}.
\label{eq:rel-H_B}
\end{equation}
(The relation may be more complicated; for example, some of the 
$\ket{\Psi_B}$'s are related with $\ket{\Psi_{B'}}$'s and some 
with $\ket{\Psi_{B''}}$'s with $B'' \neq B'$.)  Ultimately, 
all states in realistic (cosmological) spacetimes may be related 
with those in asymptotically Minkowski space as
\begin{equation}
  \{ \ket{\Psi_B} \} \,\prec\, \{ \ket{\Psi_{B'}} \} 
  \cdots \,\prec\, \{ \ket{\Psi_{\rm Minkowski}} \},
\label{eq:H_B-H_Minkowski}
\end{equation}
since the boundary area for asymptotically Minkowski space 
is infinity, ${\cal A}_{\rm Minkowski} = \infty$.  Does string 
theory formulated in an asymptotically Minkowski background (using 
$S$-matrix elements) correspond to the present theory as
\begin{equation}
  {\rm String\,\, theory}\,\, \Longleftrightarrow\,\, 
    \lim_{{\cal M} \rightarrow {\rm asymptotic\,\, Minkowski}}{\cal T} 
  \,?
\label{eq:string}
\end{equation}
Here, the ${\cal T}/N_{\rm Minkowski}$ portion is described by the 
scattering dynamics, and the $N_{\rm Minkowski}$ degrees of freedom 
are responsible for the initial conditions, where $N_{\rm Minkowski} 
= e^{{\cal A}_{\rm Minkowski}/4}$; see the next section.  If this 
is indeed the case, then it would be difficult to obtain a useful 
description of cosmological spacetimes directly in that formulation, 
since they would correspond to a special measure zero subset of 
the possible asymptotic states.

\section{Discussion}
\label{sec:discuss}

In this final section, we discuss some of the issues that have not been 
addressed in the construction so far.  This includes the possibility 
of sending signals from a past singularity or past null infinity (in 
the course of time evolution) and the interpretation of a closed universe 
in which the area of the leaf changes from increasing to decreasing 
once the scale factor at the leaf starts decreasing.  We argue that 
these issues are related to that of ``selecting a state''---even if 
the theory is specified we still need to provide selection conditions 
on a state, usually given in the form of boundary conditions (e.g.\ 
initial conditions).  Our discussion here is schematic, but it allows 
us to develop intuition about how quantum gravity in general spacetimes 
might work at the fundamental level.

\subsubsection*{Signals from a past singularity or past null infinity}

As mentioned in Section~\ref{subsec:exterior}, the evolution of a state 
in the present framework is not fully autonomous.  Consistent with the 
covariant entropy bound, we may view a holographic state to carry the 
information on the two (future-directed ingoing and past-directed outgoing) 
light sheets associated with the leaf it represents.  However, this 
is not enough to determine a future state because there may be signals 
sent into the system from a past singularity or past null infinity 
(signals originating from the lower right direction between the two 
$45^\circ$ lines in Fig.~\ref{fig:exterior}).

\begin{figure}[]
  \includegraphics[height=6.5cm]{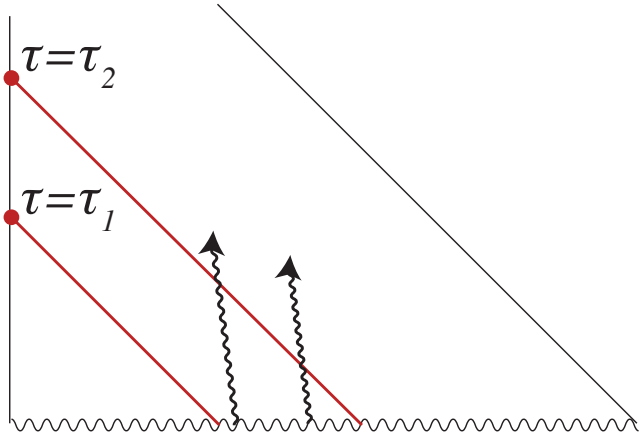}
\caption{In a universe beginning with a big bang, obtaining a future 
 state requires a specification of signals from the big bang singularity, 
 in addition to the information contained in the original state.  In 
 an FRW universe this is done by imposing spatial homogeneity and isotropy, 
 which corresponds to selecting a fine-tuned state from the viewpoint 
 of the big bang universe.}
\label{fig:sing}
\end{figure}

To be specific, let us consider a (not necessarily FRW) universe 
beginning with a big bang.  As shown in Fig.~\ref{fig:sing}, obtaining 
a future state (represented by the upper $45^\circ$ line) in general 
requires a specification of signals from the big bang singularity, 
in addition to the information contained in the original state (the 
lower $45^\circ$ line).  We usually avoid this issue by requiring the 
``cosmological principle,'' i.e.\ spatial homogeneity and isotropy, 
which determines what conditions one must put at the singularity---with 
this requirement, the state of the universe is determined by the energy 
density content in the universe at a time.  Imposing this principle, 
however, corresponds to choosing a very special state.  This is because 
there is no reason to expect that signals sent from the singularity 
at different times $\tau$ (defined holographically) are correlated 
in such a way that the system appears as homogeneous and isotropic 
in some time slicing.  In fact, this was one of the original 
motivations for inflationary cosmology~\cite{Guth:1980zm,Linde:1981mu,%
Albrecht:1982wi}.

In some cases, appropriate conditions can be obtained by assuming 
that the spacetime under consideration is a portion of larger spacetime. 
For example, if the universe is dominated by negative curvature at an 
early stage, it may arise from bubble nucleation~\cite{Coleman:1980aw}, 
in which case the homogeneity and isotropy would result from the 
dynamics of the bubble nucleation~\cite{Guth:2013sya}.  Even in this 
case, however, we would still need to specify similar conditions in 
the ambient space in which our bubble forms, and so on.  More generally, 
the analysis here says that to obtain a future state, we need to specify 
the information coming from the directions tangent to the past-directed 
light rays.  This, however, is morally the same as the usual situation 
in physics in which we need to specify boundary (e.g.\ initial) 
conditions beyond the dynamical laws the system obeys.

The situation is essentially the same in the limit of AdS/CFT; we 
only need to consider the AdS boundary instead of the big bang 
singularity.  To obtain future states, it is not enough to specify 
the initial state, given by a local operator inserted at the point 
$x_{-\infty}$ corresponding to $\tau = -\infty$ on the conformally 
compactified AdS boundary.  We also have to specify other (possible) 
boundary operators inserted at points other than $x_{-\infty}$. 

String theory formulated in terms of the $S$-matrix deals with 
this issue by adopting an asymptotically Minkowski time slice in 
which all the necessary information is viewed as being in the initial 
state.  This, however, does not change the amount of information 
needed to specify the state, which is infinite in asymptotically 
Minkowski space (because one can in principle send an infinite 
amount of information into the system from past null infinity).

\subsubsection*{Closed universes---time in quantum gravity}

Consider a closed universe in which the vacuum energy is negligible 
throughout its history.  In such a universe, the area of the leaf 
changes from increasing to decreasing in the middle of its evolution. 
On the other hand, we expect that the area of the leaf for a 
``generic'' state increases monotonically, since the number of 
independent states representing spacetime with the leaf area 
${\cal A}$ goes as $e^{{\cal A}/4}$.  What does this imply?

We interpret that states representing universes like these are 
``fine-tuned,'' so that they do not obey the usual second law of 
thermodynamics as applied to the Hilbert space of quantum gravity. 
This does not mean that they are meaningless states to consider. 
Rather, it means that we need to scrutinize carefully the concept 
of time in quantum gravity.

There are at least three different views of time in quantum gravity; 
see, e.g., Ref.~\cite{Nomura:2015zda}.  First, since time parameterization 
in quantum gravity is nothing other than a gauge choice, the 
state $\ket{\tilde{\Psi}}$ of the {\it full} system---whatever 
its interpretation---satisfies the constraint~\cite{DeWitt:1967yk,%
Wheeler:1967}
\begin{equation}
  H \ket{\tilde{\Psi}} = 0,
\label{eq:WD-eq}
\end{equation}
where $H$ is the time evolution operator, in our context the generator 
of a shift in $\tau$.  In this sense, the concept of time evolution 
does not apply to the full state $\ket{\tilde{\Psi}}$.%
\footnote{Reference~\cite{DeWitt:1967yk} states that Eq.~(\ref{eq:WD-eq}) 
 need not apply in an infinite world; for example, the state of the 
 system $\ket{\Psi_\infty}$ may depend on time in asymptotically 
 Minkowski space.  We view that Eq.~(\ref{eq:WD-eq}) still applies 
 in this case by interpreting $\ket{\tilde{\Psi}}$ to represent the 
 full system, including the ``exterior'' degrees of freedom discussed 
 in Section~\ref{subsec:exterior} (the degrees of freedom corresponding 
 to $N_{\rm Minkowski}$ below Eq.~(\ref{eq:string})) as well as the 
 ``interior'' degrees of freedom represented by $\ket{\Psi_\infty}$. 
 The time evolution of $\ket{\Psi_\infty}$ is then understood as 
 correlations between the interior and exterior degrees of freedom, 
 as described below.}
However, this of course does not mean that {\it physical time} we 
perceive is nonexistent.  Time we observe can be defined as correlations 
between subsystems (e.g.\ between an object playing the role of a 
clock and the rest)~\cite{DeWitt:1967yk,Page:1983uc}, at least in 
some branch of $\ket{\tilde{\Psi}}$.  Another way to define time 
is through probability flow in $\ket{\tilde{\Psi}}$.  Suppose 
$\ket{\tilde{\Psi}}$ is expanded in a set of states $\ket{\Psi_i}$ 
each of which represents a well-defined semiclassical spacetime when 
such an interpretation is applicable:
\begin{equation}
  \ket{\tilde{\Psi}} = \sum_i c_i \ket{\Psi_i}.
\label{eq:Psi-expand}
\end{equation}
According to the discussion in Section~\ref{sec:beyond}, $\ket{\Psi_i}$'s 
are approximately orthogonal in the appropriate limit, and the constraint 
in Eq.~(\ref{eq:WD-eq}) implies
\begin{equation}
  \sum_j c_j U_{ij} = c_i,
\qquad
  U_{ij} \equiv \bra{\Psi_i} e^{-i H \delta\tau} \ket{\Psi_j},
\label{eq:exp-const}
\end{equation}
where $U_{ij}$ is (effectively) unitary
\begin{equation}
  \sum_j U_{ij} U_{kj}^* = \sum_j U_{ji} U_{jk}^* = \delta_{ik}.
\label{eq:U_ij}
\end{equation}
Multiplying Eq.~(\ref{eq:exp-const}) with its conjugate and using 
Eq.~(\ref{eq:U_ij}), we obtain
\begin{align}
  0 = & -|c_i|^2 \sum_{j \neq i} |U_{ji}|^2 
    + \sum_{j \neq i} |c_j|^2 |U_{ij}|^2 
\nonumber\\
  & {} + \sum_{j \neq i} c_i c_j^* U_{ii} U_{ij}^* 
    + \sum_{j \neq i} c_j c_i^* U_{ij} U_{ii}^* 
    + \sum_{\substack{j,k \neq i \\ j \neq k}} c_j c_k^* U_{ij} U_{ik}^*.
\label{eq:interm}
\end{align}
In the regime where the WKB approximation is applicable, the terms in 
the second line are negligible compared with those in the first line 
because of a rapid oscillation of the phases of $c_{j,k}$'s, so that
\begin{equation}
  |c_i|^2 \sum_{j \neq i} |U_{ji}|^2 
  = \sum_{j \neq i} |c_j|^2 |U_{ij}|^2,
\label{eq:prob-flow}
\end{equation}
implying that the ``current of probability'' is conserved.  We may 
regard this current as time flow.  The time defined in this way---which 
we call {\it current time}---need not be the same as the physical time 
defined through correlations, although in many cases the former agrees 
approximately with the latter or the negative of it (up to a trivial 
shift and rescaling).

In a closed universe (with a negligible vacuum energy), it is customary 
to impose the boundary condition
\begin{equation}
  c_i = 0
\quad\mbox{for}\quad
  \{ \ket{\Psi_i}\, |\, a=0 \},
\label{eq:closed-bc}
\end{equation}
i.e.\ the wavefunction vanishes when the scale factor goes to 
zero~\cite{DeWitt:1967yk}.  With this boundary condition, current 
time $\tau$ flows in a closed circuit.  The direction of the flow 
agrees with that of physical time in the branches where $da/d\tau > 0$, 
while the two are exactly the opposite in the branches where 
$da/d\tau < 0$.  (The latter statement follows, e.g., from the 
analysis in Ref.~\cite{Aguirre:2011ac}, which shows that given 
a lower entropy final condition the most likely history of a system 
is the $CPT$ conjugate of the standard time evolution.)  Our time 
evolution in earlier sections concerns the flow of current time. 
The (apparent) violation of the second law of thermodynamics then 
arises because the condition of Eq.~(\ref{eq:closed-bc}) selects 
a special, ``standing wave'' solution from the viewpoint of the 
current time flow.  This is, however, not a fine-tuning from the 
point of view of the quantum theory in a similar way as the electron 
energy levels of the hydrogen atom are not regarded as fine-tuned 
states.

The fact that current time flows toward lower entropy states does 
not mean that a physical observer living in the $da/d\tau < 0$ 
phase sees a violation of the second law of thermodynamics.  Since 
the whole system evolves as time reversal of a standard entropy 
increasing process, {\it including memory states of the observer}, 
a physical observer always finds the evolution of the system to be 
the standard one~\cite{Nomura:2011rb,Aguirre:2011ac}; in particular, 
he/she always finds that the universe is expanding.

\subsubsection*{Static quantum multiverse---selecting the state 
in the landscape}

The analysis of string theory suggests that the theory has a multitude 
of metastable vacua each of which leads to a distinct low energy 
effective theory~\cite{Bousso:2000xa,Kachru:2003aw,Susskind:2003kw,%
Douglas:2003um}.  Combining this with the fact that many of these vacua 
lead to inflation (which is future eternal at the semiclassical level) 
leads to the picture of the inflationary multiverse~\cite{Guth:1982pn,%
Vilenkin:1983xq,Linde:1986fd,Linde:1986fc}.  The picture suggests that 
our universe is one of many bubble universes, and it cannot be a closed 
universe that will eventually collapse as the one discussed above. 
How is the state of the multiverse selected then?

A naive semiclassical picture implies that the state of the multiverse 
evolves asymptotically into a superposition of supersymmetric Minkowski 
worlds and (possibly) ``singularity worlds'' resulting from the big 
crunches of AdS bubble universes~\cite{Nomura:2011rb}.  This is because 
any other universe is expected to decay eventually.  There are basically 
two possibilities for the situation in a full quantum theory.

The first possibility is that the multiverse is in a ``scattering 
state.''  This essentially preserves the semiclassical intuition.  From 
the viewpoint of the current time flow, the multiverse begins as an 
asymptotic state, experiences nontrivial cosmology at an intermediate 
stage, and then dissipates again into the asymptotic Minkowski and 
singularity worlds.  In the earlier stage of the evolution in which 
the coarse-grained entropy decreases in $\tau$, the directions of 
current and physical time flows are opposite, while in the later 
stage of increasing entropy, the flows of the two times are in 
the same direction.  The resulting picture is similar to that of 
Ref.~\cite{Carroll:2004pn}:\ the multiverse evolves asymptotically 
into both forward and backward in (current) time.  This, however, 
does not mean that a physical observer, who is a part of the system, 
sees an entropy decreasing universe; the observer always finds that 
his/her world obeys the second law of thermodynamics.

A problem with this possibility is that specifying the theory of quantum 
gravity, e.g.\ the structure of the Hilbert space and Hamiltonian, 
is not enough to obtain the state of the multiverse and hence make 
predictions.  We would need a separate theory to specify initial 
conditions.  Furthermore, having a lower course-grained entropy at 
the turn-around point (the point at which the coarse-grained entropy 
changes from decreasing to increasing in the current time evolution) 
requires a more carefully chosen initial condition.  This leads to 
the issue of understanding why we are ``ordinary observers,'' carrying 
course-grained entropies (much) smaller than that needed to have 
any consciousness---a variant of the well-known Boltzmann brain 
problem~\cite{Dyson:2002pf,Albrecht:2002uz,Page:2006dt} (the argument 
applied to space of initial conditions, rather than to a thermal system).

The alternative, and perhaps more attractive, possibility is 
that the multiverse is in a ``bound state''~\cite{Nomura:2012zb}. 
Specifically, the multiverse is in a {\it normalizable} state 
satisfying the constraint of Eq.~(\ref{eq:WD-eq}) (as well as 
any other constraints):
\begin{equation}
  \ket{\tilde{\Psi}} = \sum_i c_i \ket{\Psi_i};
\qquad
  \sum_i |c_i|^2 < \infty.
\label{eq:Psi-norm}
\end{equation}
This is a normalization condition in spacetime, rather than in space 
as in usual quantum mechanics, and it allows us to determine, in 
principle, the state of the multiverse once the theory is given.%
\footnote{If there are multiple solutions $\ket{\tilde{\Psi}_I}$, 
 it is natural to assume that the multiverse is in the maximally 
 mixed state $\rho = \frac{1}{N} \sum_{I=1}^N \ket{\tilde{\Psi}_I} 
 \bra{\tilde{\Psi}_I}$ (in the absence of more information).  Here, 
 we have taken $\ket{\tilde{\Psi}_I}$'s to be orthonormal.}
As in the case of a collapsing closed universe, current time flows 
in a closed circuit(s) to the extent that this concept is applicable. 
This suggests that the multiverse does not probe an asymptotic 
supersymmetric Minkowski region or the big crunch singularity of 
an AdS bubble.  The origin of this phenomenon must be intrinsically 
quantum mechanical as it contradicts the naive semiclassical picture. 
In fact, such a situation is not new in physics.  As is well known, 
the hydrogen atom cannot be correctly described using classical 
mechanics:\ any orbit of the electron is unstable with respect to 
the emission of synchrotron radiation.  The situation in the quantum 
multiverse may be similar---quantum mechanics is responsible for the 
very existence of the system.

Once the state of the multiverse is determined, we should be able 
to use it to give predictions or explanations.  This requires us to 
develop a prescription for extracting answers to physical questions 
about the state.  The prescription would certainly involve coarse-graining 
(as one cannot even store the information of all possible microstates 
of the multiverse within the multiverse), and it should reproduce the 
standard Born rule giving probabilistic predictions in the appropriate 
regime.  Perhaps, the normalization condition of Eq.~(\ref{eq:Psi-norm}) 
is required in order for this prescription to be well-defined.

\section*{Acknowledgments}

We thank Raphael Bousso, Zachary Fisher, Veronika Hubeny, and Mukund 
Rangamani for conversations at the 6th Berkeley Center for Theoretical 
Physics Tahoe Summit, and we thank Sumanta Chakraborty for a useful 
comment.  We are also grateful to Kavli Institute for the Physics 
and Mathematics of the Universe, University of Tokyo for hospitality 
during the visit in which a part of this work was carried out. 
This work was supported in part by the Department of Energy, Office 
of Science, Office of High Energy Physics, under contract No.\ 
DE-AC02-05CH11231 and DE-SC0011702, by the National Science Foundation 
under grant PHY-1521446, by MEXT KAKENHI Grant Number 15H05895, and 
by Foundational Questions Institute grant FQXi-RFP-1507.  N.S. was 
supported in part by the Simons Heising Physics Fellowship Fund. 
F.S. was supported in part by the DOE NNSA Stewardship Science Graduate 
Fellowship.

\appendix

\section{Spacelike Monotonicity Theorem}
\label{app:spacelike}

Let $H$ be a past holographic screen, foliated by compact marginally 
anti-trapped surfaces i.e.\ leaves, $\{ \sigma_r \}$.  Here, $r$ is 
a (non-unique) real parameter taken to be a monotonically increasing 
function of the leaf area.  For each leaf we can construct the two 
future-directed null vector fields (up to overall normalization) and 
denote them $k^a$ and $l^a$, which satisfy
\begin{equation}
  \theta_k = 0,
\qquad
  \theta_l > 0.
\label{eq:theta_cond}
\end{equation}
Now let $h^a$ a leaf-orthogonal vector field tangent to $H$ and normalized 
by the condition $h^a \partial_a r = 1$.   Note that $h^a$ must point in 
the direction of increasing area.  We can always put $h^a = \alpha l^a 
+ \beta k^a$ for some smooth real-valued functions $\alpha$ and $\beta$ 
on $H$.  The Bousso-Engelhardt area theorem implies that $\alpha > 0$ 
everywhere.  There is no restriction on the sign of $\beta$:\ it can 
even have indefinite sign on a single leaf.

Let $A_r$ be a $d-2$ dimensional region in a leaf $\sigma_r$ and let 
$\partial A_r$ denote its boundary, where $d$ is the spacetime dimension. 
This region can be transported to a region $A_{r'}$ in a nearby leaf 
$\sigma_{r'}$ by following the integral curves of the leaf-orthogonal 
vector field $h^a$.  While Ref.~\cite{Sanches:2016pga} pointed out 
that $\norm{A_r}$ is an increasing function of $r$, this by itself 
does not guarantee that $S(A_r)$ monotonically increases.  Nonetheless, 
we now show that $S(A_r)$ indeed monotonically increases if $h^a$ 
is spacelike.

\begin{theorem}
Suppose that $H$ is a past holographic screen foliated by leaves 
$\{ \sigma_r \}$ and assume that the parameter $r$ is oriented to 
increase as leaf area increases.  Assume that $H$ is spacelike on 
some particular leaf which we take to be $\sigma_0$ by shifting $r$ 
if necessary.  Let $A_0$ be a subregion of $\sigma_0$ and define 
$A_r \subset \sigma_r$ by transporting points in $A_0$ along the 
integral curves of the leaf-orthogonal vector field in $H$.  Then, 
$S(A_r)$ is a monotonically increasing function of $r$.
\end{theorem}

\begin{proof}
Let $h^a$ be the leaf-orthogonal vector field tangent to $H$ with 
$h^a \partial_a r = 1$ and note that $h^a\big|_{\sigma_0}$ is spacelike. 
The compactness of $\sigma_0$ now allows us to find $r_0 > 0$ such 
that $h^a\big|_{H[-r_0,r_0]}$ is spacelike.  Here we have introduced 
the convenient notation
\begin{equation}
  H[r_1,r_2] = \bigcup_{r_1 \leq r \leq r_2} \sigma_r.
\end{equation}

In what follows, we will assume that the extremal surface $E(A_r)$ 
anchored to $\partial A_r$ deforms smoothly as a function of $r$ at 
$r=0$.  If this is not the case, a phase transition occurs at $r=0$ 
which will give rise to a discontinuity in the derivative of $S(A_r)$. 
However, we can then note that our theorem applies at $r$ slightly 
greater than zero (where $H$ is still spacelike and where no phase 
transition occurs) and also at $r$ slightly smaller than zero. 
This implies that $S(A_r)$ is monotonically increasing at $r=0$ 
even if $E(A_r)$ ``jumps'' at $r=0$ so that the derivative of 
$\norm{E(A_r)}$ has a discontinuity.

The maximin construction of $E(A_0)$ ensures that there exists 
$\Sigma_0 \in {\mathcal C}_{\sigma_0}$ such that $E(A_0) = 
\min(A_0,\Sigma_0)$.  Here, ${\mathcal C}_\sigma$ denotes the 
collection of all complete codimension-1 achronal surfaces lying 
in $D_\sigma$ that are anchored to $\sigma$, and $\min(A,\Sigma)$ 
denote the $d-2$ dimensional surface of minimal area lying in 
$\Sigma$ that is homologous to $A$.  If $0< \epsilon < r_0$, let
\begin{equation}
  \Sigma_\epsilon = \Sigma_0 \cup H[0,\epsilon].
\end{equation}
We claim that $\Sigma_\epsilon \in {\mathcal C}_{\sigma_\epsilon}$ for 
small $\epsilon$.  First we check that $\Sigma_\epsilon$ is achronal. 
Since $\Sigma_0$ and $H[0,\epsilon]$ are achronal independently, 
we focus on their intersection at $\sigma_0$.  The definition of 
${\mathcal C}_{\sigma_0}$ requires that $\Sigma_0$ lies in $D_{\sigma_0}$ 
so that a vector pointing from $\sigma_0$ to $\Sigma_0$ has the form 
$c_1 k^a - c_2 l^a$ with $c_1, c_2 > 0$.  Meanwhile, a vector pointing 
from $\sigma_0$ to $H[0,\epsilon]$ is proportional to $h^a\big|_{\sigma_0} 
= |\alpha| l^a - |\beta| k^a$.  Here we have made use of the fact that 
$\alpha >0$ and $\beta<0$ for a spacelike past holographic screen.  We 
see now that $\Sigma_0$ lies ``inside'' $\sigma_0$ while $h^a$ points 
toward the ``outside.''  This ensures that $\Sigma_\epsilon$ is achronal 
for sufficiently small $\epsilon$.  All that is left to check is that 
$\Sigma_\epsilon$ lies inside of $D_{\sigma_\epsilon}$.  But this is clear 
because a vector pointing from $\sigma_\epsilon$ toward $\Sigma_\epsilon$ 
is proportional to $-h^a\big|_{\sigma_\epsilon} = -|\alpha| l^a + 
|\beta| k^a$ which is indeed directed into $D_{\sigma_\epsilon}$. 
That $\Sigma_\epsilon \in {\mathcal C}_{\sigma_\epsilon}$ is now 
clear for small $\epsilon$.

We now construct an $\epsilon$-dependent family of $d-2$ dimensional 
surfaces lying on $\Sigma_0$ that are anchored to $\partial A_0$, 
which we will denote by $\Xi_\epsilon$.  Begin by fixing a small 
$\epsilon$ with $0< \epsilon < r_0$ and defining a projection function 
$\pi_\epsilon: H[0,\epsilon] \to \sigma_0$ in the natural way:\ if 
$p \in H[0,\epsilon]$, follow the integral curves of $h^a$, starting 
from $p$, until a point in $\sigma_0$ is reached.  The result is 
$\pi_\epsilon(p)$.  We can now define $\Xi_\epsilon$:
\begin{equation}
  \Xi_\epsilon 
  = \Big( \min(A_\epsilon, \Sigma_\epsilon) \cap \Sigma_0 \Big) 
    \bigcup \pi_\epsilon\Big( \min(A_\epsilon, \Sigma_\epsilon) 
      \cap H[0,\epsilon] \Big).
\end{equation}
If $\epsilon$ is sufficiently small, the fact that $H[0,\epsilon]$ has 
a positive definite metric, along with the fact that $E(A_0)$ is not 
tangent to $\sigma_0$ anywhere, ensures that $\norm{\pi_\epsilon\bigl( 
\min(A_\epsilon, \Sigma_\epsilon) \cap H[0,\epsilon] \bigr)} < \norm{ 
\min(A_\epsilon, \Sigma_\epsilon) \cap H[0,\epsilon]}$.  From this it 
follows that
\begin{equation}
  \norm{\Xi_\epsilon} < \norm{\min(A_\epsilon, \Sigma_\epsilon)}.
\end{equation}
On the other hand, because $\pi_\epsilon(\partial A_\epsilon) = 
\partial A_0$, we know that $\Xi_\epsilon$ is a codimension-2 surface 
anchored to $\partial A_0$ that lies only on $\Sigma_0$.  Thus,
\begin{equation}
  4 S(A_0) = \norm{\min(A_0,\Sigma_0)} 
  \leq \norm{\Xi_\epsilon}.
\end{equation}
Noting that the maximin construction of $E(A_\epsilon)$ requires
\begin{equation}
  \norm{\min(A_\epsilon, \Sigma_\epsilon)} \leq 4 S(A_\epsilon),
\end{equation}
we find $S(A_0) < S(A_\epsilon)$.
\end{proof}

\section{Qubit Model}
\label{app:qubit}

\subsection{Model and applications to quantum gravity}

Here we describe a toy model for holographic states representing FRW 
universes, presented originally in Ref.~\cite{Nomura:2016aww}.  We 
consider a Hilbert space for $N$ ($\gg 1$) qubits $\mathcal{H} = 
({\mathbf C}^{2})^{\otimes N}$.  Let $\Delta$ ($\leq N$) be a nonnegative 
integer and consider a typical superposition of $2^\Delta$ product states
\begin{equation}
  \ket{\Psi} = \sum_{i=1}^{2^\Delta} 
    a_i\, \ket{x^i_1 x^i_2 \cdots x^i_N},
\label{eq:psi}
\end{equation}
where $\{ a_i \}$ is a normalized complex vector, and $x^i_{1,\cdots,N} 
\in \{ 0,1 \}$.  Given an integer $n$ with $1 \leq n < N$, we can break 
the Hilbert space into a subsystem $\Gamma$ for the first $n$ qubits 
and its complement $\bar{\Gamma}$.  We are interested in computing 
the entanglement entropy $S_\Gamma$ of $\Gamma$.

Suppose $n \leq N/2$.  If $\Delta \geq n$, then $i$ in Eq.~(\ref{eq:psi}) 
runs over an index that takes many more values than the dimension of 
the Hilbert space for $\Gamma$, so that Page's argument~\cite{Page:1993df} 
tells us that $\Gamma$ has maximal entanglement entropy:\ $S_\Gamma 
= n \ln 2$.  On the other hand, if $\Delta < n$ then the number of terms 
in Eq.~(\ref{eq:psi}) is much less than both the dimension of the Hilbert 
space of $\Gamma$ and that of $\bar{\Gamma}$, which limits the entanglement 
entropy:\ $S_\Gamma = \Delta \ln 2$.  We therefore obtain
\begin{equation}
  S_\Gamma = 
    \begin{cases}
      n      & n \leq \Delta, \\
      \Delta & n > \Delta,
    \end{cases}
\label{eq:S_EE-1}
\end{equation}
for $\Delta < N/2$, while
\begin{equation}
  S_\Gamma = n,
\label{eq:S_EE-2}
\end{equation}
for $\Delta \geq N/2$.  Here and below, we drop the irrelevant factor 
of $\ln 2$.  The value of $S_\Gamma$ for $n > N/2$ is obtained from 
$S_\Gamma = S_{\bar{\Gamma}}$ since $\ket{\Psi}$ is pure.

The behavior of $S_\Gamma$ in Eqs.~(\ref{eq:S_EE-1},~\ref{eq:S_EE-2}) 
models that of $S(\gamma)$ in Section~\ref{subsec:single}.  The 
correspondence is given by
\begin{align}
  \frac{n}{N} &\,\leftrightarrow\, 
    \frac{\norm{\Gamma}}{{\cal A}_*},
\label{eq:corresp-1}\\
  \frac{\Delta}{N} &\,\leftrightarrow\, 
    \frac{1}{2} Q_w\Bigl(\frac{\pi}{2}\Bigr),
\label{eq:corresp-2}
\end{align}
for $\Delta \leq N/2$.%
\footnote{States with $\Delta > N/2$ cannot be discriminated from those 
 with $\Delta = N/2$ using $S_\Gamma$ alone.  Below, we only consider 
 the states with $N/4 \leq \Delta \leq N/2$.}
The identification of Eq.~(\ref{eq:corresp-1}) is natural if we regard 
the $N = {\cal A}_*/4$ qubits as distributing over a leaf $\sigma_*$ 
with each qubit occupying a volume of $4$ in Planck units.  The quantity 
$\Delta$ controls what universe a state represents.  For fixed $\Delta$, 
different choices of the product states $\ket{x^i_1 x^i_2 \cdots x^i_N}$ 
and the coefficients $a_i$ give $e^N$ independent microstates for the 
FRW universe with $w = f(\Delta/N)$.  The function $f$ is determined 
by Eq.~(\ref{eq:corresp-2}); in particular, $f = -1$ ($> -1$) for 
$\Delta/N = 1/2$ ($< 1/2$).

This model can be used to argue for features of the holographic theories 
discussed in Section~\ref{subsec:structure}.  We consider two cases:
\begin{itemize}
\item[] {\bf Direct sum structure} 
--- In this case, each of the subspaces ${\cal H}_{*,w}$ is modeled 
by the $N$ qubit system described here.  Consider ${\cal H}_{*,w}$ 
with a fixed $w$.  States representing the FRW universe with $w$ 
then encompass $e^N$ independent microstates in this space.  These 
microstates form ``effective vector space'' in that a superposition 
involving less than $e^{O(\delta w N)}$ of them leads only to another 
microstate representing the same FRW universe with $w$.  (We say 
that these states comprise ``fat'' preferred axes.)  Most of the 
states in ${\cal H}_{*,w}$, containing more than $e^{O(\delta w N)}$ 
of the $w$ microstates, are regarded as non-semiclassical, i.e.\ 
firewall or unphysical, states.
\item[] {\bf Russian doll structure} 
--- In this case, the entire ${\cal H}_*$ space is modeled by the 
$N$ qubits, and the states representing various FRW universes are all 
elements of this single Hilbert space of dimension $e^N$.  An important 
point is that the set of states with {\it any} fixed $\Delta_w$ 
provide a complete basis for the whole Hilbert space, where $\Delta_w 
\equiv N f^{-1}(w)$.  This implies that we can obtain a state with 
any $w' < w$ by superposing $e^{\Delta_{w'} - \Delta_w}$ states 
with $\Delta_w$, and we can also obtain a state with $w' > w$ as 
a superposition of carefully chosen $e^{\Delta_w}$ states with 
$\Delta_w$.  We call this the ``Russian doll'' structure, which 
is depicted schematically in Fig.~\ref{fig:possib}.
\end{itemize}

\subsection{Effective incoherence of superpositions}

We now focus on the latter case and consider a normalized superposition
\begin{equation}
  \ket{\Psi} = c_1 \ket{\Psi_1} + c_2 \ket{\Psi_2},
\label{eq:app-superp}
\end{equation}
of two states
\begin{alignat}{2}
  \ket{\Psi_1} &= \sum_{i = 1}^{2^{\Delta_1}} 
    a_i\, \ket{x^i_1 x^i_2 \cdots x^i_N}
\qquad
  && \Biggl( \sum_{i = 1}^{2^{\Delta_1}} |a_i|^2 = 1 \Biggr),
\label{eq:app-Psi_1}\\
  \ket{\Psi_2} &= \sum_{i = 1}^{2^{\Delta_2}} 
    b_i\, \ket{y^i_1 y^i_2 \cdots y^i_N}
\qquad
  && \Biggl( \sum_{i = 1}^{2^{\Delta_2}} |b_i|^2 = 1 \Biggr),
\label{eq:app-Psi_2}
\end{alignat}
with $\Delta_1 \neq \Delta_2$ and
\begin{equation}
  \Delta_1, \Delta_2 \leq \frac{N}{2}.
\end{equation}
Here, the coefficients $a_i$ and $b_i$ are random, as are the binary 
values $x^i_{1,\cdots,N}$ and $y^i_{1,\cdots,N}$, and $|c_1|^2 + 
|c_2|^2 = 1$ up to an exponentially suppressed correction arising from 
$\inner{\Psi_1}{\Psi_2} \neq 0 \lesssim O(2^{-|\Delta_1-\Delta_2|/2})$. 
We are interested in the reduced density matrix
\begin{equation}
  \rho_{1 \cdots n} = {\rm Tr}_{n+1 \cdots N}\, \rho,
\label{eq:app-reduced}
\end{equation}
obtained by performing a partial trace on
\begin{equation}
  \rho = \ket{\Psi}\bra{\Psi} 
  = |c_1|^2 \ket{\Psi_1} \bra{\Psi_1} + |c_2|^2 \ket{\Psi_2} \bra{\Psi_2} 
    + c_1 c_2^* \ket{\Psi_1} \bra{\Psi_2} 
    + c_2 c_1^* \ket{\Psi_2} \bra{\Psi_1},
\label{eq:app-rho}
\end{equation}
over the subsystem consisting of the first $n$ qubits.  We will only 
consider the case where $n < N/2$.

We begin our analysis by considering ${\rm Tr}_{n+1 \cdots N} 
\ket{\Psi_1} \bra{\Psi_1}$.  It is convenient to write
\begin{equation}
  \ket{\Psi_1} \bra{\Psi_1} 
  = \sum_{i=1}^{2^{\Delta_1}} |a_i|^2\, 
      \ket{x^i_1 \cdots x^i_N} \bra{x^i_1 \cdots x^i_N} 
    + \sum_{\substack{i,j=1 \\ i\neq j}}^{2^{\Delta_1}} 
      a_i a_j^*\, \ket{x^i_1 \cdots x^i_N} \bra{x^j_1 \cdots x^j_N}.
\label{eq:Psi1_density}
\end{equation}
Upon performing the partial trace over $\ket{\Psi_1} \bra{\Psi_1}$, the 
first sum gives a diagonal contribution to the reduced density matrix
\begin{equation}
  D_{11} = \sum_{i=1}^{2^{\Delta_1}} |a_i|^2\, 
    \ket{x^i_1 \cdots x^i_n} \bra{x^i_1 \cdots x^i_n}.
\end{equation}
The second sum gives a correction
\begin{equation}
  \tilde{D}_{11} = \sum_{\substack{i,j=1 \\ i\neq j}}^{2^{\Delta_1}} 
    a_i a_j^*\, \ket{x^i_1 \cdots x^i_n} \bra{x^j_1 \cdots x^j_n}\, 
    \delta_{x^i_{n+1}, x^j_{n+1}} \cdots \delta_{x^i_N, x^j_N}.
\label{eq:corr11}
\end{equation}
We now consider two cases:
\begin{itemize}
\item[(i)] $\Delta_1 > n$.
\\
Because $2^{\Delta_1} \gg 2^n$, it is clear that $D_{11}$ is 
a $2^n \times 2^n$ diagonal matrix with every diagonal entry 
approximately given by
\begin{equation}
  \frac{2^{\Delta_1}}{2^n} \left< |a_i |^2 \right> = 2^{-n}.
\end{equation}
(Note that $\left< |a_i |^2 \right> = 2^{-\Delta_1}$ because 
$\ket{\Psi_1}$ is normalized and random.)  Thus, $D_{11}$ is 
a fully mixed state.  Now observe that $\tilde{D}_{11}$  consists 
of almost all zeros.  In fact, looking at Eq.~(\ref{eq:corr11}) 
we see that there are $2^{2 \Delta_1 - N + n}$ nonzero entries 
of average absolute value $2^{-\Delta_1}$.  Given that $\Delta_1 
\leq N/2$, we conclude that $\tilde{D}_{11}$ has exponentially 
fewer nonzero entries than $D_{11}$, and that each nonzero entry 
has exponentially smaller size than the entries of $D_{11}$.
\item[(ii)] $\Delta_1 \leq n$.
\\
In this case, $D_{11}$ is a diagonal matrix having $2^{\Delta_1}$ 
nonzero entries of order $2^{-\Delta_1}$.  The number of nonzero 
entries in $\tilde{D}_{11}$ is, again, $2^{2 \Delta_1 - N + n}$, 
each having the average absolute value $2^{-\Delta_1}$.  The effect 
of $\tilde{D}_{11}$ is highly suppressed because its number of 
nonzero entries is exponentially smaller than that of $D_{11}$. 
In fact, for the number of nonzero entries in $\tilde{D}_{11}$ 
to compete with that in $D_{11}$, we would need $2\Delta_1 - 
N + n \geq \Delta_1$, which, however, mean
\begin{equation}
  \Delta_1 \geq N-n > \frac{N}{2},
\end{equation}
a contradiction.
\end{itemize}
Summarizing, ${\rm Tr}_{n+1 \cdots N} \ket{\Psi_1} \bra{\Psi_1} 
= D_{11} + \tilde{D}_{11}$ is a diagonal matrix having 
$2^{{\rm min}\{ \Delta_1, n \}}$ nonzero entries of order 
$2^{-{\rm min}\{ \Delta_1, n \}}$, up to exponentially suppressed 
effects.  The same analysis obviously applies to ${\rm Tr}_{n+1 \cdots N} 
\ket{\Psi_2} \bra{\Psi_2} = D_{22} + \tilde{D}_{22}$ with 
$\Delta_1 \rightarrow \Delta_2$.

We now turn our attention to the matrix ${\rm Tr}_{n+1 \cdots N} 
\ket{\Psi_1} \bra{\Psi_2}$, which we denote as $\tilde{D}_{12}$:
\begin{equation}
  \tilde{D}_{12} = \sum_{i=1}^{2^{\Delta_1}} \sum_{j=1}^{2^{\Delta_2}} 
    a_i b_j^*\, \ket{x^i_1 \cdots x^i_n} \bra{y^j_1 \cdots y^j_n} 
    \delta_{x^i_{n+1}, y^j_{n+1}} \cdots \delta_{x^i_N, y^j_N}.
\label{eq:corr12}
\end{equation}
We argue, along similar lines to the above, that $\tilde{D}_{12}$ is 
exponentially smaller than $|c_1|^2 D_{11} + |c_2|^2 D_{22}$, unless 
$|c_1|$ or $|c_2|$ is exponentially suppressed.  Once again, we have 
several cases:
\begin{itemize}
\item[(i)] $\Delta_1, \Delta_2 \leq n$.
\\
In this case, $|c_1|^2 D_{11} + |c_2|^2 D_{22}$ is a diagonal matrix having 
$2^{\Delta_1}$ nonzero entries of order $2^{-\Delta_1}$ and $2^{\Delta_2}$ 
nonzero entries of order $2^{-\Delta_2}$.  Considering Eq.~(\ref{eq:corr12}), 
$\tilde{D}_{12}$ consists of zeros except for $2^{\Delta_1+\Delta_2-N+n}$ 
nonzero entries with the average absolute value $\left< |a_i b_j^*| \right> 
= 2^{-(\Delta_1+\Delta_2)/2}$.  The number of these entries, however, is 
exponentially smaller than $2^{\Delta_1}$, since having $\Delta_1 + \Delta_2 
- N + n \geq \Delta_1$ would require $\Delta_2 \geq N-n > N/2$; similarly, 
it is also exponentially smaller than $2^{\Delta_2}$.  Moreover the 
changes of the exponentially rare eigenvalues affected are at most 
of $O(1)$.  We conclude that the effect of $\tilde{D}_{12}$ is 
exponentially suppressed.
\item[(ii)] $\Delta_1, \Delta_2 > n$.
\\
In this case, the condition that $|c_1|^2 + |c_2|^2 = 1$ ensures that 
$|c_1|^2 D_{11} + |c_2|^2 D_{22}$ is a $2^n \times 2^n$ unit matrix 
multiplied by $2^{-n}$.  Meanwhile, $\tilde{D}_{12}$ consists of zeros 
except for $2^{\Delta_1 + \Delta_2 - N + n} \ll 2^n$ nonzero entries 
of size $2^{-(\Delta_1+\Delta_2)/2} \ll 2^{-n}$.
\item[(iii)] $\Delta_1 \leq n < \Delta_2$.
\\
In this case, $D_{22}$ is a $2^n \times 2^n$ unit matrix multiplied by 
$2^{-n}$ while $D_{11}$ is a diagonal matrix having $2^{\Delta_1}$ nonzero 
entries of order $2^{-\Delta_1}$.  Once again, the number of nonzero entries 
in $\tilde{D}_{12}$ is exponentially smaller than $2^{\Delta_1}$, since 
$\Delta_1 + \Delta_2 - N + n \geq \Delta_1$ would require $\Delta_2 \geq 
N - n > N/2$, and the fractional corrections to eigenvalues from 
these entries are of order $2^{-(\Delta_2-n)}$.  This implies that 
the effect of $\tilde{D}_{12}$ is negligible.  The same argument 
also applies to the case that $\Delta_2 \leq n < \Delta_1$.
\end{itemize}

We conclude that for $n < N/2$, we find
\begin{equation}
  \rho_{1 \cdots n} = |c_1|^2 D_{11} + |c_2|^2 D_{22} 
  = \sum_{i=1}^{2^{\Delta_1}} |a_i|^2\, 
      \ket{x^i_1 \cdots x^i_n} \bra{x^i_1 \cdots x^i_n} 
    + \sum_{i=1}^{2^{\Delta_2}} |b_i|^2\, 
      \ket{y^i_1 \cdots y^i_n} \bra{y^i_1 \cdots y^i_n},
\label{eq:app-rho_n}
\end{equation}
up to effects exponentially suppressed in $N \approx O({\cal A}_*)$.  This 
implies that the reduced density matrix for the state $\ket{\Psi}$ takes 
the form of an incoherent classical mixture
\begin{equation}
  \rho_{1 \cdots n} 
  = |c_1|^2 \rho^{(1)}_{1 \cdots n} + |c_2|^2 \rho^{(2)}_{1 \cdots n},
\end{equation}
where $\rho^{(k)}_{1 \cdots n} = {\rm Tr}_{n+1 \cdots N} \ket{\Psi_k} 
\bra{\Psi_k}$ ($k = 1, 2$) are the reduced density matrices we would 
obtain if the state were $\ket{\Psi_k}$.

The form of Eq.~(\ref{eq:app-rho_n}) also implies that the entanglement entropy
\begin{equation}
  S_{1 \cdots n} = - {\rm Tr}_{1 \cdots n} 
    (\rho_{1 \cdots n} \ln \rho_{1 \cdots n}),
\end{equation}
obeys a similar linear relation
\begin{equation}
  S_{1 \cdots n} 
  = |c_1|^2 S^{(1)}_{1 \cdots n} + |c_2|^2 S^{(2)}_{1 \cdots n} + O(1),
\end{equation}
unless $|c_1|$ or $|c_2|$ is exponentially small.  Here, 
$S^{(k)}_{1 \cdots n} = -{\rm Tr}_{1 \cdots n} (\rho^{(k)}_{1 \cdots n} 
\ln \rho^{(k)}_{1 \cdots n})$.  This can be seen by considering the same 
three cases as above.  If $\Delta_1,\Delta_2 \leq n$, $\rho_{1 \cdots n}$ 
is a diagonal matrix having $2^{\Delta_1}$ nonzero entries with average 
value $|c_1|^2 2^{-\Delta_1}$ and $2^{\Delta_2}$ nonzero entries with 
average value $|c_2|^2 2^{-\Delta_2}$.  In this case, 
\begin{equation}
  S_{1 \cdots n} 
  = -|c_1|^2 \ln \frac{|c_1|^2}{2^{\Delta_1}} 
    - |c_2|^2 \ln \frac{|c_2|^2}{2^{\Delta_2}} 
  = |c_1|^2 \Delta_1 \ln 2 + |c_2|^2 \Delta_2 \ln 2 + O(1),
\end{equation}
while we have $S^{(k)}_{1 \cdots n} = \Delta_k \ln 2$.  The $O(1)$ 
correction from linearity is the entropy of mixing, given by
\begin{equation}
  S_{1 \cdots n, {\rm mix}} = -|c_1|^2 \ln |c_1|^2 -|c_2|^2 \ln |c_2|^2.
\end{equation}
If $\Delta_1, \Delta_2 > n$, then $\rho_{1 \cdots n}$ is a unit matrix 
multiplied by $2^{-n}$.  From this it follows that $S_{1 \cdots n} 
= n \ln 2 = |c_1|^2 n \ln 2 + |c_2|^2 n \ln 2$, which is desirable 
given that $S^{(k)}_{1 \cdots n} = n \ln 2$ for $\Delta_k > n$.%
\footnote{The absence of the mixing contribution in this case is 
 an artifact of the specific qubit model considered here, arising 
 from the fact that two universes cannot be discriminated unless 
 $n$ is larger than one of $\Delta_{1,2}$; see Eq.~(\ref{eq:S_EE-1}). 
 In realistic cases, the mixing contribution should always exist 
 for any macroscopic region in the holographic space as two 
 different universes can be discriminated in that region; see, 
 e.g., Fig.~\ref{fig:Q-gamma}.}
Finally, if $\Delta_1 < n < \Delta_2$, $\rho^{(1)}_{1 \cdots n}$ 
has $2^{\Delta_1}$ nonzero entries of mean value $2^{-\Delta_1}$ 
while $\rho^{(2)}_{1 \cdots n}$ is a unit matrix multiplied by 
$2^{-n}$.  Because $2^{-\Delta_1} \gg 2^{-n}$ the total density 
matrix $\rho_{1 \cdots n}$ given by Eq.~(\ref{eq:app-rho_n}) is 
diagonal and has $2^{\Delta_1}$ entries of size $|c_1|^2 2^{-\Delta_1}$ 
and $2^n$ entries of size $|c_2|^2 2^{-n}$.  We thus find that 
$S_{1 \cdots n} = |c_1|^2 \Delta_1 \ln 2 + |c_2|^2 n \ln 2 + 
S_{1 \cdots n, {\rm mix}} = |c_1|^2 S^{(1)}_{1 \cdots n} + |c_2|^2 
S^{(2)}_{1 \cdots n} + O(1)$.  (This expression is valid for 
$\Delta_1 = n < \Delta_2$ as well.)

\end{document}